\documentclass[12pt]{article} 

\usepackage{array,epsfig,fancyheadings,rotating}
\usepackage{sectsty, secdot}
\sectionfont{\fontsize{12}{14pt plus.8pt minus .6pt}\selectfont}
\renewcommand{\theequation}{\thesection\arabic{equation}}
\subsectionfont{\fontsize{12}{14pt plus.8pt minus .6pt}\selectfont}

\RequirePackage[authoryear]{natbib}

\usepackage{xr}

\usepackage{avant}
\usepackage[T1]{fontenc}
\usepackage[latin9]{inputenc}
\usepackage{geometry}[margin = 1 inch]
\usepackage{babel}

\usepackage{stmaryrd}

\usepackage{breakurl}

\makeatletter
\@ifundefined{date}{}{\date{}}
\usepackage{babel}
\usepackage{amsthm}\usepackage{bm}\usepackage{graphicx}\usepackage{rotating}\usepackage{array}\usepackage{booktabs}\usepackage{color}\usepackage{algorithmic}
\usepackage{algorithm}\usepackage{multirow}\usepackage{enumerate}\usepackage{graphicx}\usepackage{caption}\usepackage{subcaption}

\newcommand*{\addFileDependency}[1]{
  \typeout{(#1)}
  \@addtofilelist{#1}
  \IfFileExists{#1}{}{\typeout{No file #1.}}
}

\floatstyle{ruled}
\newfloat{algorithm}{tbp}{loa}
\providecommand{\algorithmname}{Algorithm}
\floatname{algorithm}{\protect\algorithmname}

\@ifundefined{definecolor}{color}{}
\RequirePackage[displaymath]{lineno}

\@ifundefined{definecolor}{color}{}
\usepackage{amsfonts}\usepackage{latexsym}

\usepackage{mathrsfs}\usepackage{amsthm}\usepackage{latexsym}\usepackage{booktabs}\@ifundefined{definecolor}{color}{}

\@ifundefined{definecolor}{color}{}

\newcommand{\independent}{\protect\mathpalette{\protect\independenT}{\perp}}
\def\independenT#1#2{\mathrel{\rlap{$#1#2$}\mkern2mu{#1#2}}}

\RequirePackage{amsthm,amsmath,amsfonts,amssymb}
\usepackage[unicode=true,
bookmarks=false,
breaklinks=false,pdfborder={0 0 1},backref=false,colorlinks=false]
{hyperref}
\RequirePackage{graphicx}

\newtheorem{lemma}{Lemma}
\newtheorem{proposition}{Proposition}
\newtheorem{theorem}{Theorem}

\newtheorem{corollary}{Corollary}
\newtheorem{assumption}{Assumption}

\newcommand{\T}{^{\mbox{{\sf T}}}} 
\def\T{{ \mathrm{\tiny T} }}

\usepackage{babel}

\def\E{{\bf E}}
\def\cov{{\mbox{cov}}}

\makeatother

\begin{document}

\newcommand{\bSigma}{{\bm \Sigma}}

\newcommand{\bX}{\mathbf{X}}
\newcommand{\bY}{\mathbf{Y}}
\newcommand{\bx}{\mathbf{x}}
\newcommand{\mby}{\mathbf{y}}
\newcommand{\bV}{\mathbf{V}}
\newcommand{\bM}{\mathbf{M}}
\newcommand{\bN}{\mathbf{N}}
\newcommand{\bbeta}{{\bm\beta}}
\newcommand{\bgamma}{{\bm\gamma}}
\newcommand{\bGamma}{{\bm\Gamma}}
\newcommand{\bA}{\mathbf{A}}
\newcommand{\bB}{\mathbf{B}}
\newcommand{\bC}{\mathbf{C}}
\newcommand{\bD}{\mathbf{D}}
\newcommand{\bE}{\mathbf{E}}
\newcommand{\bI}{\mathbf{I}}
\newcommand{\bS}{\mathbf{S}}
\newcommand{\bs}{\mathbf{s}}
\newcommand{\bu}{\mathbf{u}}
\newcommand{\bv}{\mathbf{v}}
\newcommand{\f}{\mathbf{f}}
\newcommand{\bt}{\mathbf{t}}
\newcommand{\by}{\mathbf{y}}

\newcommand{\bh}{\mathbf{h}}
\newcommand{\bzeta}{{\bm \zeta}}
\newcommand{\bdelta}{{\bm \delta}}
\newcommand{\balpha}{{\bm \alpha}}
\newcommand{\btheta}{{\bm \theta}}
\newcommand{\bTheta}{{\bm \Theta}}
\newcommand{\bmeta}{{\bm \eta}}
\newcommand{\bPhi}{{\bm \Phi}}
\newcommand{\bPsi}{{\bm \Psi}}
\newcommand{\bOmega}{{\bm \Omega}}
\newcommand{\bP}{\mathbf{P}}
\newcommand{\bT}{\mathbf{T}}
\newcommand{\bW}{\mathbf{W}}
\newcommand{\bZ}{\mathbf{Z}}
\newcommand{\bepsilon}{\bm \epsilon}
\newcommand{\bmu}{\bm \mu}
\newcommand{\bxi}{\bm \xi}
\newcommand{\bPi}{\bm \Pi}

\global\long\def\mbA{\mathbf{A}}
 \global\long\def\hatmbA{\widehat{\mathbf{A}}}
 \global\long\def\mbB{\mathbf{B}}
  \global\long\def\mbC{\mathbf{C}}
 \global\long\def\mbD{\mathbf{D}}
 \global\long\def\mbN{\mathbf{N}}
 \global\long\def\mbV{\mathbf{V}}
 \global\long\def\hatmbV{\widehat{\mathbf{V}}}
 \global\long\def\mbv{\mathbf{v}}
 \global\long\def\hatmbu{\widehat{\mathbf{u}}}
 \global\long\def\mbu{\mathbf{u}}
 \global\long\def\hatmbv{\widehat{\mathbf{v}}}
 \global\long\def\tilmbv{\widetilde{\mbv}}
 \global\long\def\mbX{\mathbf{X}}
 \global\long\def\mbY{\mathbf{Y}}
 \global\long\def\mbU{\mathbf{U}}
 \global\long\def\mbW{\mathbf{W}}
 \global\long\def\mbP{\mathbf{P}}
 \global\long\def\mbQ{\mathbf{Q}}
 \global\long\def\mbO{\mathbf{O}}
 \global\long\def\mbM{\mathbf{M}}
  \global\long\def\bM{\mathbf{M}}
 \global\long\def\mbR{\mathbf{R}}
 \global\long\def\barmbV{\overline{\mathbf{V}}}
 \global\long\def\mbS{\mathbf{S}}
 \global\long\def\mbT{\mathbf{T}}
 
 \global\long\def\mbf{\mathbf{f}}
 \global\long\def\mbG{\mathbf{G}}
 \global\long\def\mbI{\mathbf{I}}
 \global\long\def\hatmbG{\widehat{\mbG}}
 \global\long\def\mbg{\mathbf{g}}
 \global\long\def\hatmbg{\widehat{\mbg}}
 \global\long\def\hatmbM{\widehat{\mbM}}
 \global\long\def\hatmbW{\widehat{\mbW}}
 \global\long\def\hatmbU{\widehat{\mbU}}

  \global\long\def\bolepsilon{\boldsymbol{\epsilon}}

 \global\long\def\mbw{\mathbf{w}}
 \global\long\def\hatmbw{\widehat{\mbw}}
 \global\long\def\bolGamma{\boldsymbol{\Gamma}}
 \global\long\def\hatbolGamma{\widehat{\bolGamma}}
 \global\long\def\bolOmega{\boldsymbol{\Omega}}
 \global\long\def\bolmu{\boldsymbol{\mu}}
 \global\long\def\bolnu{\boldsymbol{\nu}}
 \global\long\def\bolLambda{\boldsymbol{\Lambda}}
 \global\long\def\hatbolLambda{\widehat{\boldsymbol{\Lambda}}}
 \global\long\def\bolPhi{\boldsymbol{\Phi}}
 \global\long\def\mbbR{\mathbb{R}}
 \global\long\def\mbbS{\mathbb{S}}
 \global\long\def\calE{\mathcal{E}}
 \global\long\def\boltheta{\boldsymbol{\theta}}
 \global\long\def\boleta{\boldsymbol{\eta}}
 \global\long\def\bolphi{\boldsymbol{\phi}}
 \global\long\def\bolpsi{\boldsymbol{\psi}}
 \global\long\def\hatboltheta{\widehat{\boltheta}}
 \global\long\def\bolbeta{\boldsymbol{\beta}}
  \global\long\def\bbeta{\boldsymbol{\beta}}
 \global\long\def\hatbolbeta{\widehat{\bolbeta}}
 \global\long\def\bolgamma{\boldsymbol{\gamma}}
 \global\long\def\hatbolgamma{\widehat{\bolgamma}}
 \global\long\def\hatboleta{\widehat{\boleta}}
 \global\long\def\hatbolpsi{\widehat{\bolpsi}}
 \global\long\def\hatlambda{\widehat{\lambda}}
 \global\long\def\hatbolOmega{\widehat{\bolOmega}}
 \global\long\def\bolvarepsilon{\boldsymbol{\varepsilon}}
 \global\long\def\bolalpha{\boldsymbol{\alpha}}
 \global\long\def\bolTheta{\boldsymbol{\Theta}}
  \global\long\def\bTheta{\boldsymbol{\Theta}}

 \global\long\def\bolSigma{\boldsymbol{\Sigma}}
 \global\long\def\bolDelta{\boldsymbol{\Delta}}
 \global\long\def\hatbolSigma{\widehat{\bolSigma}}
 \global\long\def\calR{{\cal R}}
 \global\long\def\calRp{\calR^{\perp}}
 \global\long\def\calG{\mathcal{G}}
 \global\long\def\calS{\mathcal{S}}
 \global\long\def\calU{\mathcal{U}}
 \global\long\def\hatphi{\widehat{\phi}}
 \global\long\def\hatcalE{\widehat{\calE}}

\global\long\def\vecc{\mathrm{vec}}
 \global\long\def\Prob{\mathrm{Pr}}
 \global\long\def\E{\mathrm{E}}
 \global\long\def\Cov{\mathrm{cov}}
 \global\long\def\Corr{\mathrm{corr}}
 \global\long\def\F{\mathrm{F}}
 \global\long\def\J{\mathrm{J}}
 \global\long\def\I{\mathcal{I}_{n}}
 \global\long\def\II{\mathcal{I}_{n}^{\mathrm{1D}}}
 \global\long\def\Jn{\mathrm{J}_{n}}
 \global\long\def\Ln{\ell_{n}}
 \global\long\def\Var{\mathrm{var}}
 \global\long\def\dimension{\mathrm{dim}}
 \global\long\def\spn{\mathrm{span}}
 \global\long\def\vech{\mathrm{vech}}
  \global\long\def\supp{\mathrm{supp}}

\global\long\def\Env{\mathrm{Env}}
 \global\long\def\tr{\mathrm{trace}}
  \global\long\def\cov{\mathrm{cov}}
 \global\long\def\dg{\mathrm{diag}}
 \global\long\def\asyVar{\mathrm{avar}}
\global\long\def\HT{\mathrm{HT}}

\global\long\def\MenvU{\calE_{\mbM}(\mbU)}
 \global\long\def\hatMenvU{\widehat{\calE}_{\mbM}(\mbU)}
 \global\long\def\MDDM{\mathrm{MDDM}}
 \global\long\def\nMDDM{\mathrm{MDDM}_{n}}
 
\newcommand{\beq}{\begin{equation}}
\newcommand{\eeq}{\end{equation}}
\newcommand{\beqn}{\begin{equation*}}
\newcommand{\eeqn}{\end{equation*}}
\newcommand{\bea}{\begin{eqnarray}}
\newcommand{\eea}{\end{eqnarray}}
\newcommand{\bean}{\begin{eqnarray*}}
\newcommand{\eean}{\end{eqnarray*}}

\def\Tr{{ \mathrm{Tr} }}




\renewcommand{\baselinestretch}{1.9}
\renewcommand{\sectionmark}[1]{\markright{\thesection. #1}}
\markright{ \hbox{\footnotesize\rm Statistica Sinica
	}\hfill\\[-13pt]
	\hbox{\footnotesize\rm
	}\hfill }

\markboth{\hfill{\footnotesize\rm Mai et al.} \hfill}
{\hfill {\footnotesize\rm Slicing-free Inverse Regression in High-dimensional Sufficient Dimension Reduction} \hfill}

\renewcommand{\thefootnote}{}
$\ $\par


\fontsize{12}{14pt plus.8pt minus .6pt}\selectfont \vspace{0.8pc}
\centerline{\large\bf Slicing-free Inverse Regression in High-dimensional }
\vspace{2pt} 
\centerline{\large\bf Sufficient Dimension Reduction}
\vspace{.4cm} 
\centerline{Qing Mai$^1$, Xiaofeng Shao$^2$, Runmin Wang$^3$ and Xin Zhang$^1$} 
\vspace{.4cm} 
\centerline{\it Florida State University$^1$}
\centerline{\it University of Illinois at Urbana-Champaign$^2$}
\centerline{\it Texas A\&M University$^3$}
\vspace{.55cm} \fontsize{9}{11.5pt plus.8pt minus.6pt}\selectfont


\begin{quotation}
	\noindent {\it Abstract:}
		Sliced inverse regression \citep[SIR,][]{SIR} is a pioneering work and the most recognized method in sufficient dimension reduction.  While promising progress has been made in theory and methods of high-dimensional  SIR, two remaining challenges are still nagging high-dimensional multivariate applications. First, choosing the number of slices in SIR is a difficult problem, and it depends on the sample size, the distribution of variables, and other practical considerations. Second, the extension of SIR from univariate response to multivariate is not trivial. Targeting at the same dimension reduction subspace as SIR, we propose a new slicing-free method that provides a unified solution to  sufficient dimension reduction with high-dimensional covariates and univariate or multivariate response. We achieve this by adopting the recently developed martingale difference divergence matrix \citep[MDDM,][]{LS18} and penalized eigen-decomposition algorithms. To establish the consistency of our method with a high-dimensional predictor and a multivariate response, we develop a new concentration inequality for sample MDDM around its population counterpart using theories for U-statistics, which may be of independent interest. Simulations and real data analysis demonstrate the favorable finite sample performance of the proposed method.

	\vspace{9pt}
	\par
\end{quotation}\par

\def\thefigure{\arabic{figure}}
\def\thetable{\arabic{table}}

\renewcommand{\theequation}{\thesection.\arabic{equation}}

\fontsize{12}{14pt plus.8pt minus .6pt}\selectfont

\section{Introduction}

Sufficient dimension reduction (SDR) is an important statistical analysis tool for data visualization, summary and inference. It extracts low-rank projections of the predictors $X$ that contain all the information about the response $Y$, without specifying a parametric model beforehand. The semi-parametric nature of SDR leads to great flexibility and convenience in practice. After SDR, we can model the conditional distributions of the response given the lower dimensional projected covariate using existing parametric or non-parametric methods. A salient feature with SDR  is that the low-rank projection space can be accurately estimated at a parametric rate with the nonparametric part treated as an infinitely dimensional nuisance parameter.  For example, in multi-index models, SDR can estimate the multiple projection directions without estimating the unspecified link function. 

A cornerstone of SDR is the SIR (sliced inverse regression), pioneered by \cite{SIR}, who first discovered the inner connection between the low-rank projection space and the eigen-space of $\cov(\E(\mbX\mid Y))$, under suitable assumptions. 
SIR is performed by slicing the response $Y$ and aggregating the conditional mean of the predictor $\mbX$ given the response $Y$ within each slice. To illustrate the idea, we consider a univariate response $Y$. Slicing involves picking $K+1$ constants $-\infty=a_0<a_1<\ldots<a_K=\infty$, and defining a new random variable $H$, where $H=k$ if and only if $a_{k-1}<Y\le a_{k}$.   Upon a centering and standardization of the covariate, i.e., $\mbX\rightarrow \widetilde{\mbX}=\bolSigma_{\mbX}^{-1/2}(\mbX-\E(\mbX))$, a simple eigen-decomposition can be conducted to find linear projections that explain most of the variability in the conditional expectation of the transformed predictor given the response across slices, that is, $\cov(\E(\widetilde{\mbX}\mid H))$. As an important variation of SIR, sliced average variance estimation \citep{SAVE} utilizes the conditional variance across slices. 
A key step in these inverse regression methods is apparently the choice of the slicing scheme. If $Y$ is sliced too coarsely, we may not be able to capture the full dependence of $Y$ on the predictors, which could lead to a large bias in the estimation of $\cov(\E(\widetilde{\mbX}\mid Y))$. If $Y$ is sliced too finely, the with-in-slice sample size is small, leading to a large variability in estimation. Although \citet{SIR,HC92} showed that SDR can still be consistent in large sample even when the slicing scheme is chosen poorly, \cite{ZN95} argued that the choice of slicing scheme is critical to achieve high estimation efficiency. To the best of our knowledge, there seems no universal guidance on the choice of the slicing scheme provided in the literature.

\cite{zhu2010dimension} and \cite{CZ13} showed that it is beneficial to aggregate multiple slicing schemes rather than relying on a single one. However, proposals in the above-mentioned papers have their own limitations as they exclusively focused on the univariate response. In many real life problems, it is common to encounter multi-response data.
 Component-wise analysis may not be sufficient for multi-response data  because it does not fully make use of the componentwise dependence in the response.   But slicing multivariate response is notoriously hard due to the curse of dimensionality, a common problem in multivariate nonparametric smoothing. As the dimension for the response becomes moderately large, it is increasingly difficult to make sure that each slice contains a decent number of samples, and the estimation can be unstable in practice. Hence, it is highly desirable to develop new SDR methods that do not involve slicing.

An important line of research in the recent SDR literature is to develop SDR methods for datasets with high-dimensional covariates, as motivated by many \allowdisplaybreaks contemporary applications. The idea of SDR is naturally attractive for high-dimensional datasets, as an effective reduction of the dimension in $X$ facilitates the use of existing modeling and inference methods that are tailored for low-dimensional covariates.  However,  most classical SDR methods are not directly applicable to the large $p$ small $n$ setting, where $p$ is the dimension of $X$ and $n$ is the sample size. To overcome the challenges  with high-dimensional covariates, several methods have been proposed  recently. In \cite{LZL18}, they show that the SIR estimator is consistent if and only if $\lim p/n=0$. When the dimension $p$ is larger than $n$, they propose a diagonal thresholding screening SIR (DT-SIR) algorithm, and show that it is consistent to recover the dimension reduction space under certain sparsity assumptions on both the covariance matrix of predictors and the loadings of the directions. In \cite{LZL19}, they further introduce a simple Lasso regression method to obtain an estimate of the SDR space by  constructing artificial response variables made up from top eigenvectors of the estimated conditional covariance matrix. In \cite{TWLZ18}, they propose a two-stage computational framework to solve
the sparse generalized eigenvalue problem, which includes the high-dimensional SDR as a special case, and propose a truncated Rayleigh flow method (namely, RIFLE) to estimate the leading generalized eigenvector.
Also see \cite{LLHL20} and \cite{tan2018convex} for related recent work. 
These methods provide valuable tools to tackle high-dimensional SDR problem. However, all these methods still rely on the SIR as an important component in their methodology and involve choosing a single slicing scheme, but little guideline on the slicing scheme is provided. Consequently, these methods cannot be easily applied to data with multivariate response, and the impact from the choice of slicing scheme is unclear.


In this article, we propose a novel slicing-free SDR method in the high-dimensional setting. Our proposal is inspired by a recent nonlinear dependence metric: the so-called martingale difference divergence matrix \citep[MDDM,][]{LS18}. The MDDM was developed by \cite{LS18} as a matrix-valued extension of MDD (martingale difference divergence) in \cite{SZ14}, which measures the (conditional) mean dependence of a response variable given a covariate, and was used for dimension reduction of a multivariate time series. As recently revealed by \cite{ZLS20}, at the population level, the eigenvectors (or generalized eigenvectors) of the MDDM are always contained in the central subspace. Building on these prior works, we propose a penalized eigen-decomposition on MDDM to perform SDR in high dimensions. For the case when the covariance matrix of the predictor is identity matrix, we adopt the truncated power method with hard thresholding to estimate the top-$K$ eigenvectors of MDDM. In the case of more general covariance structure, we adopt the RIFLE algorithm \citep{TWLZ18} and apply to the sample MDDM instead of sample SIR estimator of $\cov(\E(\mbX\mid Y))$. With the use of sample MDDM, this approach is completely slicing-free and allows  to treat the univariate response and multivariate response in a unified way, thus the practical difficulty of selecting the number of slices (especially for multivariate response) is circumvented. On the theory front,  we derive a concentration inequality for sample MDDM around its population counterpart by using theories for U statistics, and obtain a rigorous non-asymptotic theoretical justification for the estimated central subspaces for both settings. Simulations and real data analysis confirm that PMDDM outperforms slicing-based methods in estimation accuracy.



The rest of this paper is organized as follows. {In Section \ref{sec:MDDM}, we give a brief review of the martingale difference divergence matrix (MDDM) and then present a new concentration inequality for the sample MDDM around its population counterpart.  In Section~\ref{sec:method}, we present our general methodology of adopting MDDM in both model-free and model-based SDR problems, where we establish population level connections between the central subspace and the eigen-decomposition and the generalized eigen-decomposition of MDDM. Algorithms for regularized eigen-decomposition and generalized eigen-decomposition problems are proposed in Sections~\ref{subsec:PMDDM} and \ref{subsec:GEPMDDM}, respectively. Theoretical properties are established in Section~\ref{sec:theory}. Section~\ref{sec:numerical} contains numerical studies. Finally, Section~\ref{sec:discussion} concludes the paper with a short discussion. The Supplementary Materials collect all additional technical details and numerical results. }


\section{MDDM and its concentration inequality}\label{sec:MDDM}

Consider a pair of random vectors $\mbV\in\mbbR^{p}$, $\mbU\in\mbbR^{q}$ such that $\E(\|\mbU\|^2+\|\mbV\|^2)<\infty$. We use $\|\mbU\|=|\mbU|_{q}$ to denote the Euclidean norm in $\mbbR^q$. 
Define
\[\MDDM(\mbV\mid\mbU)=-\E\left[\{\mbV-\E(\mbV)\}\{\mbV'-\E(\mbV')\}^{\T}\|\mbU-\mbU'\|\right]\in\mbbR^{p\times p},\]
where $(\mbV',\mbU')$ is an independent copy of $(\mbV,\mbU)$.  \cite{LS18} established the following key properties of  $\MDDM(\mbV\mid\mbU)$: (i) It is  symmetric and positive semi-definite; (ii) $\E(\mbV\mid\mbU)=\E(\mbV)$ almost surely, is equivalent to $\MDDM(\mbV\mid\mbU)=0$; (iii) For any $p\times d$ matrix $\mbA$, $\MDDM(\mbA^\T\mbV\mid\mbU)=\mbA^\T\MDDM(\mbV\mid\mbU)\mbA$; (iv) There exist $p-d$ linearly independent combinations of $\mbV$ such that they
are (conditionally) mean independent of $\mbU$ if and only if $\mbox{rank} (\MDDM(\mbV|\mbU)) = d$.

Given a random sample of size $n$, i.e., $(\mbU_k,\mbV_k)_{k=1}^{n}$, the sample estimate of $\MDDM(\mbV\mid\mbU)$, denoted by $\nMDDM(\mbV\mid\mbU)$, is defined as 
\begin{equation}
\nMDDM(\mbV\mid\mbU)=-\dfrac{1}{n^{2}}\sum_{j,k=1}^{n}(\mbV_{j}-\barmbV_{n})(\mbV_{k}-\barmbV_{n})^\T\vert\mbU_{j}-\mbU_{k}\vert_{q},\label{MDDMn}
\end{equation}
where $\barmbV_{n}=n^{-1}\sum_{k=1}^{n}\mbV_k$ is the sample mean. 

In the following, we present a concentration inequality for the sample MDDM around its population counterpart, which plays an instrumental role in our consistency proof for the proposed penalized MDDM method later. To this end, we 
let $\mbV=(V_1,\cdots,V_p)^\T\in \mbbR^p$ and  assume the following condition.


\begin{enumerate}
	\item[(C1)] There exists two positive constants $\sigma_0$ and $C_0$ such that
	\begin{align}
		\begin{split}
			&\sup_{p}\max_{1\le j\le p}\E\{\exp(2\sigma_0V_j^2)\}\le C_0,\\
			&\E\{\exp(2\sigma_0\Vert \mbU\Vert_q^2)\}\le C_0.
		\end{split}
	\end{align}

\end{enumerate}
For a matrix $A=(a_{ij})$, we denote its max norm as $\|A\|_{max}=\max_{ij} |a_{ij}|$. 
\begin{theorem}
	\label{th:main}
	Suppose that Condition (C1) holds.
	There exists a positive integer
	$n_0=n_0(\sigma_0,C_0,q)<\infty$, $\gamma=\gamma(\sigma_0,C_0,q)\in (0,1/2)$ and a finite positive constant $D_0=D_0(\sigma_0,C_0,q)<\infty$ such that when $n\ge n_0$ and
	$16>\epsilon>D_0 n^{-\gamma}$, we have
	\[P(\|\MDDM_n(\mbV|\mbU)-\MDDM(\mbV|\mbU)\|_{\max}>12\epsilon)
	\le 54 p^2\exp\left\{-\frac{\epsilon^2n}{36\log^3(n)}\right\}.\]
	
\end{theorem}

The above bound is non-asymptotic and holds for all $(n,p,\epsilon)$ as long as the condition is satisfied. The exponent $\dfrac{\epsilon^2n}{\log^3(n)}$ is due to the use of a truncation argument along with Hoeffding's inequality for U-statistic, and 
seems hard to improve. Nevertheless we are able to achieve an exponential type bound under a uniform sub-Gaussian condition on both $\mbV$ and $\mbU$. This result may be of independent theoretical interest. For example, in the time series dimension reduction problem studied by \cite{LS18}, our Theorem~\ref{th:main} could potentially help extend the theory there from low-dimensional multivariate time series to higher dimensions. 

\section{Slicing-free Inverse Regression via MDDM}\label{sec:method}

\subsection{Inverse regression subspace in sufficient dimension reduction}

Sufficient dimension reduction (SDR) methods aim to identify the central subspace that preserves all the information in the predictors. In this paper, we consider the SDR problem of a multivariate response $\mbY\in\mathbb{R}^q$ on a multivariate predictor $\mbX\in\mathbb{R}^{p}$. The central subspace $\cal S_{\mbY\mid\mbX}$ is defined as the intersection of all the subspaces $\cal S$ such that $\mbY\independent \mbX\mid \mbP_{\cal S}\mbX$, where $\mbP_{\cal S}$ is the projection matrix onto $\cal S$. By construction, the central subspace $\cal S_{\mbY\mid\mbX}$ is the smallest dimension reduction subspace that contains all the information in the conditional distribution of $\mbY$ given $\mbX$. Many methods have been proposed for the recovery of the central subspace or a portion of the central subspace \citep{SIR,SAVE,BC01,ccl02,yc03,CN2005,lw07,zh08}. See \citet{Li2018} for a  comprehensive review. Although the central subspace is well defined for both univariate and multivariate response, most existing SDR methods consider the case with univariate response, while extension to multivariate response is non-trivial. 

The definition of central subspace is not very constructive, as it requires taking intersection of \emph{all} subspace $\calS\subseteq\mbbR^{p}$ such that $\mbY\independent \mbX\mid \mbP_{\cal S}\mbX$. It is indeed a very ambitious goal to estimate the central subspace without specifying a model between $\mbY$ and $\mbX$. To achieve this, we often need additional assumptions such as the linearity and the coverage conditions. The linearity condition requires that, for any basis of the central subspace $\bolbeta$, we must have that $\E(\mbX\mid \bbeta^\T\mbX)$ is linear in $\bbeta^\T\mbX$. The linearity condition is guaranteed if $\mbX$ is elliptically contoured, and allows us to connect the central subspace to the conditional expectation $\E(\mbX\mid\mbY)$. Define $\bolSigma_{\mbX}$ as the covariance of $\mbX$ and the \emph{inverse regression subspace}
\begin{multline}
	\label{IRsubspace}
	\calS_{\E(\mbX\mid \mbY)}\equiv\spn\{\E(\mbX\mid \mbY=\mby)-\E(\mbX):  \mby\in\mbbR^q \mathrm{\ such\ that}\ \E(\mbX\mid \mbY=\mby)\ \mathrm{exists} \}.
\end{multline}
Then following property is well-known and often adopted in developing SDR methods.
\begin{proposition}\label{prop.linear}
Under linearity condition, we have $\calS_{\E(\mbX\mid \mbY)}\subseteq\bolSigma_{\mbX}\calS_{\mbY\mid\mbX}\subseteq\mbbR^p$.
\end{proposition}

The coverage condition further assumes that $\calS_{\E(\mbX\mid \mbY)}=\bolSigma_{\mbX}\calS_{\mbY\mid\mbX}$. It follows that we can estimate the central subspace by modeling the conditional expectation of $\mbX$. Indeed, many SDR methods approximate $\E(\mbX\mid\mbY)$. For example, the most classical SDR method, sliced inverse regression (SIR), slices univariate $Y$ into several categories and estimate the mean of $\mbX$ within each slice. Most later methods also follow this slice-and-estimate procedure. Apparently, the slicing scheme  is important to the estimation. If there are too few slices, we may not be able to fully capture the dependence of $\mbX$ on $Y$; however, if there are too many slices, there is a lack of  enough samples within each slice to allow accurate estimation.

\subsection{MDDM in SDR}

In this section, we lay the foundation for the application of MDDM in SDR. We show that the subspace spanned by MDDM coincides with the inverse regression subspace in \eqref{IRsubspace}. In particular, we have the following Proposition~\ref{prop1}, which was used in \cite{ZLS20}, without a proof, in the context of multivariate linear regression. 

\begin{proposition}\label{prop1}
For multivariate $\mbX\in\mbbR^{p}$ and  $\mbY\in \mbbR^{q}$, assuming
the existence of $\E(\mbX)$, $\E(\mbX\mid \mbY)$ and $\MDDM(\mbX\mid \mbY)$,
we have $
\calS_{\E(\mbX\mid \mbY)}=\spn\{\MDDM(\mbX\mid \mbY)\}$.
\end{proposition}

Therefore, the rank of $\MDDM(\mbX\mid \mbY)$ is the dimensionality of the inverse regression subspace; and the non-trivial eigenvectors of $\MDDM(\mbX\mid \mbY)$ contain all the information for $\calS_{\E(\mbX\mid \mbY)}$. 
Combining Proposition\ref{prop.linear}\&~\ref{prop1},  we immediately have that (i) under the linearity condition, $\bolSigma_{\mbX}^{-1}\spn\{\MDDM(\mbX\mid \mbY)\}\subseteq\calS_{\mbY\mid\mbX}$; and (ii) under the linearity and coverage conditions, $\bolSigma_{\mbX}^{-1}\spn\{\MDDM(\mbX\mid \mbY)\}=\calS_{\mbY\mid\mbX}$. 

Henceforth, we assume both the linearity and coverage conditions, which are assumed either explicitly or implicitly in inverse regression type dimension reduction methods \citep[e.g.,][]{SIR,CN2005, zhu2010dimension, CZ13}. Then the central subspace is related to the eigen-decomposition of $\MDDM(\mbX\mid \mbY)$. Specifically, we have the following scenarios.

If $\Cov(\mbX)=\sigma^2\mbI_{p}$ for some $\sigma^2>0$, then obviously $\spn\{\MDDM(\mbX\mid \mbY)\}=\calS_{\mbY\mid\mbX}$. This includes single index and multiple index models with uncorrelated predictors. Let $K$ be the rank of $\MDDM(\mbX\mid \mbY)$, then the dimension of the central subspace is $K$; and the first $K$ eigenvectors of $\MDDM(\mbX\mid \mbY)$ span the central subspace. 

If $\Cov(\mbX\mid \mbY)=\sigma^{2}\mbI_{p}$ for some $\sigma^2>0$, then we have $\bolSigma_\mbX=\sigma^2\mbI_p + \Cov\{\E(\mbX\mid \mbY)\}$. Because $\spn[\Cov\{\E(\mbX\mid \mbY)\}]= \calS_{\E(\mbX\mid\mbY)}$, we can show that $\calS_{\mbY\mid\mbX}=\bolSigma_\mbX^{-1}\spn\{\MDDM(\mbX\mid \mbY)\}=\spn\{\MDDM(\mbX\mid \mbY)\}$. To see this, let $\Cov\{\E(\mbX\mid \mbY)\}=\mbU\mbU^\T$ for some $\mbU\in\mbbR^{p\times K}$, then $\spn(\mbU)=\spn[\Cov\{\E(\mbX\mid \mbY)\}]=\spn\{\MDDM(\mbX\mid \mbY)\}$ and we may also write $\MDDM(\mbX\mid \mbY)=\mbU\boldsymbol{\Psi}\mbU^{\T}$ for some symmetric positive definite matrix $\boldsymbol{\Psi}\in\mbbR^{K\times K}$. Then the result follows by applying the Woodbury matrix identity to $\bolSigma_\mbX^{-1}=(\sigma^2\mbI_p + \mbU\mbU^\T)^{-1}=\sigma^{-2}\mbI_p -\sigma^{-2} \mbU(\sigma^2\mbI_K + \mbU^{\T}\mbU)^{-1}\mbU^{\T}$. The non-trivial eigenvectors of $\MDDM(\mbX\mid \mbY)$  again span the central subspace.

For general covariance structure, the $d$-dimensional central subspace $\calS_{\mbY\mid\mbX}=\bolSigma_\mbX^{-1}\spn\{\MDDM(\mbX\mid \mbY)\}$ can be obtained via generalized eigen-decomposition. Specifically, consider the generalized eigenvalue problem
\begin{equation}\label{GEV}
\MDDM(\mbX\mid \mbY)\mbv_i = \varphi_i\bolSigma_{\mbX}\mbv_i,\ \varphi_i\geq0,\ \mbv_i\in\mbbR^p,
\end{equation}
where $\mbv_i^\T\bolSigma_{\mbX}\mbv_j=0$ for $i\neq j$. Then, similar to \citep{li2007sparse,chen2010coordinate}, it is straightforward to show that the generalized eigenvector spans the central subspace, $\calS_{\mbY\mid\mbX}=\spn(\mbv_1,\dots,\mbv_K)$.
 
Existing works in SDR often focus on the eigen-decomposition or the generalized eigen-decomposition of $\Cov\{\E(\mbX\mid Y=y)\}$, where non-parametric estimates of $\E(\mbX\mid Y=y)$ are obtained from
slicing the support of the univariate response $Y$. Comparing to these approaches, the MDDM approach requires no tuning parameter selection (i.e. specifying slicing schemes). Moreover, high-dimensional theoretical study of MDDM is easier and does not require additional assumptions on the conditional mean function $\E(\mbX\mid \mbY)$ such as smoothness in the empirical mean function of $\mbX$ given $Y$ (e.g.~sliced stable condition in \citet{LZL18}).

\subsection{MDDM for model-based SDR}
 
So far, we have discussed model-free SDR. Another important research area in SDR is model-based methods, which provide invaluable intuition for the use of inverse regression estimation under the assumption that the conditional distribution of $\mbX\mid \mbY$ is normal. In this section, we consider the principal fitted component (PFC) model, which was discussed in details in \cite{CF09} and \cite{C07}, and generalize it from univariate response to multivariate response. 
We argue that (generalized) eigen-decomposition of MDDM is potentially advantageous to the likelihood-based approaches under PFC model. This is somewhat surprising but reasonable, considering that the advantages of MDDM over least squares and likelihood-based estimation was demonstrated in \cite{ZLS20} under the multivariate linear model.

Let $\mbX_{\mby}\sim\mbX\mid(\mbY=\mby)$ denote the conditional variable,
then the PFC model is
\begin{equation}
\mbX_{\mby}=\bolmu+\bolGamma\bolnu_{\mby}+\bolvarepsilon,\quad \bolvarepsilon\sim N(0,\bolDelta),
\end{equation}
where $\bolGamma\in\mbbR^{p\times K}$, $K<p$, is a non-stochastic
orthogonal matrix, $\bolnu_{\mby}\in\mbbR^{K}$ is the latent variable
that depends on $\mby$. 
%
Then the latent variable $\bolnu_{\mby}$ is fitted as $\bolnu_{\mby}=\bolalpha\mbf_{\mby}$ with some user-specified functions $\mbf_{\mby}=(f_{1}(\mby),\ldots,f_{m}(\mby))^{\T}\in\mbbR^{m}$,
$m\geq K$, that maps $q$-dimensional response to $m$-dimensional. In the univariate PFC model, $q=1$, so the $m$ functions can be viewed as an expansion of the response (similar to slicing). For our multivariate extensions of the PFC model, there is no requirement of $m\geq q$. The PFC model can be written as
\begin{equation}
\mbX_{y}=\bolmu+\bolGamma\bolalpha\mbf_{\mby}+\bolvarepsilon,
\end{equation}
where $\bolGamma$ and $\bolalpha$ are estimated similarly to the multivariate
reduced-rank regression with $\mbX\in\mbbR^{p}$ being the response
and $\mbf_{\mby}\in\mbbR^{m}$ being the predictor. Finally, the central
subspace under this PFC model is $\bolDelta^{-1}\spn(\bolGamma)$,
which simplifies to $\spn(\bolGamma)$ if we further assume isotropic
error (i.e. isotropic PFC model) $\bolDelta=\Cov(\mbX\mid\mbY)=\sigma^{2}\mbI_{p}$. 

For the PFC model, our MDDM approach is the same as the model-free MDDM counterpart, and has two main advantages over the likelihood-based PFC estimation: (i)
there is no need to specify the functions $\mbf_{\mby}$, and thus no
risk of mis-specification; (ii) extensions to high-dimensional setting is much more straightforward.
Moreover, under the isotropic PFC model, the central subspace $\calS_{\mbY\mid\mbX}=\spn(\bolGamma)$ is exactly the first $K$ eigenvectors of $\MDDM(\mbX\mid \mbY)$.

%
%


%

\section{Estimation}\label{sec:Estimation}

\subsection{Penalized decomposition of MDDM}\label{subsec:PMDDM}

Based on the results in the last section, penalized eigen-decomposition of MDDM can be used for estimating the central subspace in high dimension when the covariance $\bolSigma_\mbX$ or the conditional covariance $\Cov(\mbX\mid \mbY)$ is proportional to the identity matrix $\mbI_p$. We consider the construction of such an estimate. It is worth mentioning that the penalized decomposition of MDDM we develop here is immediately applicable to the dimension reduction of multivariate stationary time series in \citep{LS18}, which is beyond the scope of this article.
Moreover, it is well-known that $\bolSigma_{\mbX}^{-1}$ is not easy to estimate in high dimensions. Then, even when for general covariance structure, the eigen-decomposition of MDDM provides estimate of the inverse regression subspace (though may differ from the central subspace) that is useful for exploratory data analysis (e.g.~detecting and visualizing non-linear mean function).


As such, we consider the estimation of the eigenvectors of  $\MDDM(\mbX\mid \mbY)$. We assume that $\MDDM(\mbX\mid \mbY)$ has $K$ nontrivial eigenvectors, denoted by $\bolbeta_1,\ldots,\bolbeta_K$. We use the shorthand notation $\mbM=\MDDM(\mbX\mid \mbY)$. Also, we note that, given the first $k-1$ eigenvectors, $\bolbeta_{k}$ is the top eigenvector of $\mbM_{k}$, where $\mbM_{k}=\mbM-\sum_{l<k}(\bolbeta_l^\T\mbM\bolbeta_l)\bolbeta_l\bolbeta_l^\T$.

It is well-known that the eigenvectors cannot be accurately estimated in high dimensions without additional assumptions. We adopt the popular sparsity assumption that many entries in $\bolbeta_k$ are zero. To estimate these sparse eigenvectors, denote $\widehat{\mbM}_1=\MDDM_n(\mbX\mid \mbY)$, where the sample $\MDDM_n$ is defined in \eqref{MDDMn}. We find $\widehat\bolbeta_k,k=1,\ldots,K$ as follows:
\begin{equation}\label{Penalized Eigen Decompose}
\widehat\bolbeta_k=\arg\max_{\bolbeta}\bolbeta^\T\widehat{\mbM}_k\bolbeta \mbox{ s.t. $\bolbeta^\T\bolbeta=1,\Vert\bolbeta\Vert_0\le s,$}
\end{equation}
where $\widehat{\mbM}_1=\MDDM_n(\mbX\mid \mbY)$, $\widehat{\mbM}_k=\widehat{\mbM}_1-\sum_{l<k}\delta_l\widehat\bolbeta_l\widehat\bolbeta_l^\T$ for $k>1$ with $\delta_l=\widehat\bolbeta_l^\T\widehat{\mbM}_1\widehat\bolbeta_l$, and $s$ is a tuning parameter.

We solve the above problem by combining the truncated power method with hard thresholding. For a vector $\mbv\in \mathbb{R}^p$ and a positive integer $s$, denote $v_s^*$ as the $s$-th largest value of $\vert v_j\vert ,j=1,\ldots, p$. The hard-thresholding operator is $\HT(\mbv,s)=(v_1 I(\vert v_1\vert\geq v_s^*),\ldots,v_p I(\vert v_p\vert\geq v_s^*))^\T$, which sets the $p-s$ elements in $\mbv$ to zero. We solve \eqref{Penalized Eigen Decompose} by Algorithm~\ref{alg1}, where the initialization $\widehat\bolbeta_1^{(0)}$ may be randomly generated. Note that \citet{yuan2013truncated} proposed Algorithm~\ref{alg1} to perform principal component analysis through penalized eigen-decomposition on the sample covariance. 

{In our algorithm, we require a pre-specified sparsity level $s$ and subspace dimension $K$. In theory, we show that our estimators for $\bolbeta_k$, $k=1,\dots,K$, are all consistent for their population counterparts when the sparsity $s$ is sufficiently large (i.e.~larger than the population sparsity level) and the number of directions $K$ is no bigger than the true dimension of the central subspace. Therefore, our method is flexible in the sense that the pre-specified $s$ and $K$ do not have to be exactly correct. In practice, especially in exploratory data analysis, the number of sequentially extracted directions are often set to be small (i.e.~$K=1,2 ~\mbox{or}~3$), while the determination of true central subspace dimension is a separate and important research topic in SDR \citep[e.g.][]{bura2011dimension,luo2016combining} and is beyond the scope of this paper. Moreover, the pre-specified sparsity level $s$ combined with $\ell_0$ regularization is potentially convenient for post-dimension reduction inference \citep{kim2020post}, as seen in the post-selection inference of canonical correlation analysis that is done over subsets of variables of pre-specified cardinality \citep{mckeague2020significance}.  }

As pointed out by a referee, other sparse principal component analysis (PCA) methods can potentially be applied to decompose MDDM. We choose to extend the algorithm in \citet{yuan2013truncated} to facilitate computation and theoretical development. For computationally efficient sparse PCA methods such as \citet{zou2006sparse,witten2009penalized}, their theoretical properties are unfortunately unknown. Hence, we expect the theoretical study of their MDDM-variants to be very challenging. On the other hand, for the theoretically justified sparse PCA methods such as \citet{vu2013minimax,cai2013sparse}, the computation is less efficient.

\renewcommand{\baselinestretch}{1}

\begin{algorithm}[t!]
\begin{enumerate}
\item Input: $s, K,\widehat{\mbM}_1=\widehat\mbM=\MDDM_n(\mbX\mid \mbY)$. 

\item Initialize $\widehat\bolbeta_1^{(0)}$. 

\item For $k=1,\ldots, K$, do

\begin{enumerate}
\item Iterate over $t$ until convergence:
\begin{enumerate}
         \item Set $\widehat\bolbeta_k^{(t)}=\widehat{\mbM}_k\widehat\bolbeta_k^{(t-1)}$. 
         \item If $\Vert\widehat\bolbeta_k^{(t)}\Vert_0\le s$, set 
         $$\widehat\bolbeta_k^{(t)}=\dfrac{\widehat\bolbeta_k^{(t)}}{\Vert\widehat\bolbeta_k^{(t)}\Vert_2};$$ 
         else 
         $$\widehat\bolbeta_k^{(t)}=\dfrac{\HT(\widehat\bolbeta_k^{(t)},s)}{\Vert \HT(\widehat\bolbeta_k^{(t)},s)\Vert_2}$$

\end{enumerate}
\item Set $\widehat\bolbeta_k=\widehat\bolbeta_k^{(t)}$ at convergence and $\widehat{\mbM}_{k+1}=\widehat{\mbM}_k-\widehat\bolbeta_k^\T\widehat{\mbM}\widehat\bolbeta_k\cdot \widehat\bolbeta_k\widehat\bolbeta_k^\T$.
\end{enumerate}
\item Output $\widehat\calS_{\mbY\mid\mbX}=\spn(\widehat\bolbeta_1,\ldots,\widehat\bolbeta_K)$.
\end{enumerate}
\caption{Penalized eigen-decomposition of MDDM.}\label{alg1}
\end{algorithm}

\renewcommand{\baselinestretch}{2}

\subsection{Generalized eigenvalue problems with MDDM}\label{subsec:GEPMDDM}

Now we consider the general (arbitrary) covariance structure $\bolSigma_{\mbX}$. We continue to use $\bolbeta_1,\ldots,\bolbeta_K$ to denote the nontrivial eigenvectors of $\bolSigma_{\mbX}^{-1}\spn(\MDDM(\mbX\mid \mbY))$ so that the central subspace is spanned by the $\bolbeta$'s. Again, we assume that these eigenvectors are sparse. In principle, we could assume that $\bolSigma_{\mbX}^{-1}$ is also sparse and construct its estimate accordingly. However, $\bolSigma_{\mbX}^{-1}$ is a nuisance parameter for our ultimate goal, and additional assumptions on it may unnecessarily limit the applicability of our method. Hence, we take a different approach as follows.

To avoid estimating $\bolSigma_{\mbX}^{-1}$, we note that $\bolbeta_1,\ldots,\bolbeta_K$ can also be viewed as the generalized eigenvectors defined as follows, which is equivalent to \eqref{GEV},
\begin{equation}
    \label{GEP1}
	\bolbeta_k=\arg\max_{\bolbeta}\bolbeta^\T\mbM\bolbeta, \mbox { s.t. $\bolbeta^\T\bolSigma_{\mbX}\bolbeta=1, \bolbeta_l^\T\bolSigma_{\mbX}\bolbeta=0$ for any $l<k$.}
\end{equation}
	
Directly solving the generalized eigen-decomposition problem in \eqref{GEP1} is not easy if we want to further impose penalties, because it is difficult to satisfy the orthogonality constraints. Therefore, we further consider another form for \eqref{GEP1} that does not involve the orthogonal constraints. This alternative form is based on the following lemma.

\begin{lemma}\label{sub}
Let $\lambda_j=\bolbeta_j^\T\mbM\bolbeta_j$ and $\mbM_k=\mbM-\bolSigma_{\mbX}(\sum_{j<k}\lambda_j\bolbeta_j\bolbeta_j^\T)\bolSigma_{\mbX}$. We have 
\beq
\bolbeta_{k}=\arg\max_{\bolbeta}\bolbeta^\T\mbM_k\bolbeta,\quad\mbox{ s.t. $\bolbeta^\T\bolSigma_{\mbX}\bolbeta=1$.}
\eeq
\end{lemma}

Motivated by Lemma~\ref{sub}, we consider the penalized problem that 
$\bolbeta_k=\arg\max_{\bolbeta}\bolbeta^\T\widehat{\mbM}_k\bolbeta$ such that  $\bolbeta^\T\bolSigma_{\mbX}\bolbeta=1,\Vert\bolbeta\Vert_0\le s$, where $\widehat{\mbM}_1=\MDDM_n(\mbX\mid\mbY)$ and $\widehat{\mbM}_k=\widehat{\mbM}_1-\widehat{\bolSigma}_\mbX\left(\sum_{l<k}\delta_l\widehat{\bolbeta}_l\widehat{\bolbeta}_l^\T\right)\widehat{\bolSigma}_\mbX$ for $k>1$ with $\delta_l=\widehat{\bolbeta}_l^\T\widehat{\mbM}\widehat{\bolbeta}_l$, and $s$ is a tuning parameter. We adopt the RIFLE algorithm in \citet{TWLZ18} to solve this problem. See the details in Algorithm~\ref{alg2}.  
In our simulation studies, we considered randomly generated initial value $\widehat\bolbeta_1^{(0)}$ and fixed step size $\eta=1$, and observed reasonably good performance.

Although Algorithm~\ref{alg2} is a generalization of the RIFLE Algorithm in \citet{TWLZ18}, there are important differences between these two. On one hand, the RIFLE Algorithm only extracts the first generalized eigenvector, whereas Algorithm~\ref{alg2} is capable of estimating multiple generalized eigenvectors by properly deflating the MDDM. In sufficient dimension reduction problems, the central subspace often has a structural dimension greater than 1, and it is necessary to find more than one  generalized eigenvector. Hence, Algorithm~\ref{alg2} is potentially more useful than the RIFLE algorithm in practice. On the other hand, the usefulness of RIFLE Algorithm has been demonstrated in several statistical applications, including sparse sliced inverse regression. Here Algorithm~\ref{alg2} decomposes the MDDM, which is the first time the penalized generalized eigenvector problem is used to perform sufficient dimension reduction in a slicing-free manner in high dimensions. A brief analysis of the computation complexity is included in Section~S3 in the Supplementary Materials.

\renewcommand{\baselinestretch}{1}

\begin{algorithm}[t!]
\begin{enumerate}
\item Input: $s, K, \widehat{\mbM}_1=\widehat\mbM$ and step size $\eta>0$. 
\item Initialize $\widehat\bolbeta_1^{(0)}$. 

\item For $k=1,\ldots, K$, do

\begin{enumerate}

\item Iterate over $t$ until convergence:
\begin{enumerate}
         \item Set $\rho^{(t-1)}=\dfrac{(\widehat\bolbeta_k^{(t-1)})^\T\widehat{\mbM}_k\widehat{\bolbeta}_k^{(t-1)}}{(\widehat\bolbeta_k^{(t-1)})^\T\widehat{\bolSigma}_\mbX\widehat{\bolbeta}_k^{(t-1)}}$.
         \item $\mbC=\mbI+(\eta/\rho^{(t-1)})\cdot(\widehat{\mbM}_k-\rho^{(t-1)}\widehat{\bolSigma}_\mbX)$
         \item $\widetilde{\bolbeta}_k^{(t)}=\mbC\widehat{\bolbeta}_k^{(t-1)}/\Vert\mbC\widehat{\bolbeta}_k^{(t-1)}\Vert_2$.
         \item $\widehat{\bolbeta}_k^{(t)}=\dfrac{\HT(\widetilde{\bolbeta}_k, s)}{\Vert\HT(\widetilde{\bolbeta}_k, s)\Vert_2}$

\end{enumerate}
\item Set $\widetilde\bolbeta_k=\widehat\bolbeta_k^{(t)}$ at convergence and scale it to obtain $\widehat\bbeta_k=\dfrac{\widetilde{\bolbeta}_k}{\sqrt{\widetilde\bolbeta_k^\T\widehat{\bSigma}_\mbX\widetilde\bolbeta_k}}$.

\item Set $\widehat{\mbM}_{k+1}=\widehat{\mbM}_k-\widehat{\bolSigma}_\mbX\widehat\bolbeta_k^\T\widehat{\mbM}\widehat\bolbeta_k\cdot \widehat\bolbeta_k\widehat\bolbeta_k^\T\widehat{\bolSigma}_\mbX$.
\end{enumerate}
\item Output $\widehat\calS_{\mbY\mid\mbX}=\spn(\widehat\bolbeta_1,\ldots,\widehat\bolbeta_K)$.
\end{enumerate}
\caption{Generalized eigen-decomposition of MDDM.}\label{alg2}
\end{algorithm}
\renewcommand{\baselinestretch}{2}

\section{Theoretical properties}\label{sec:theory}

In this section, we consider theoretical properties of the generalized eigenvectors of $(\MDDM(\mbX\mid\mbY),\bSigma_{\bX})$. Recall that, if we know that $\bSigma_{\bX}=\bI$, the generalized eigenvectors reduce to eigenvectors, and can be estimated by Algorithm~\ref{alg1}. If we do not have any information about $\bSigma_{\bX}$, we can find the generalized eigenvectors with Algorithm~\ref{alg2}. Either way, we let $\bbeta_k, k=1,\ldots, K,$ be the first $K$ (generalized) eigenvectors of $\MDDM(\mbX\mid\mbY)$. Throughout the proof, we let $C$ denote a generic constant that can vary from line to line. We show the consistency of $\widehat\bbeta_k$ by 
proving that $\eta_k=|\sin\Theta(\widehat\bbeta_k,\bbeta_k)|\le Cs\epsilon$.  We assume that $K$ is fixed, and $s\epsilon\le 1$. {Recall that we define 
$\lambda_j=\bolbeta_j^\T\mbM\bolbeta_j$ as the (generalized) eigenvalue. Further define $d=\max_{k=1}^K\{\Vert\bbeta_k\Vert_{0}\}$. When we study Algorithm~\ref{alg1} or Algorithm~\ref{alg2}, we assume that $s=d+2s'$, where $s'=Cd$ for a sufficiently large $C$. To apply the concentration inequalities for $\MDDM$, we restate below Condition (C1) in terms of $\mbX$ and $\mbY$ as Condition (C1'), along with other suitable conditions:

\renewcommand{\baselinestretch}{1}

\begin{enumerate}
	\item[(C1')] There exist two positive constants $\sigma_0$ and $C_0$ such that
	$\E\{\exp(2\sigma_0\Vert \mbY\Vert_q^2)\}\le C_0$ and
	$\sup_{p}\max_{1\le j\le p}\E\{\exp(2\sigma_0X_j^2)\}\le C_0$.
\item[(C2)] There exist $\Delta>0$ such that $\min_{k=1,\ldots,K}(\lambda_k-\lambda_{k+1})\ge\Delta$.
\item[(C3)] There exists constants $U,L$ that do not depend on $n,p$ such that $L\le \lambda_{K}\le \lambda_1\le U$.
\item[(C4)] As $n\rightarrow \infty$, $dn^{-1/2} {(\log p)^{1/2}(\log{n})^{3/2}}\rightarrow 0$.
\end{enumerate}

\renewcommand{\baselinestretch}{2}

Condition (C2) guarantees that the eigenvectors are well-defined. Condition (C3) imposes bounds on the eigenvalues of $\bM$. Researchers often impose similar assumptions on the covariance matrix to achieve consistent estimation. Condition (C4) restricts the growth rate of $p,d$ with respect to $n$. {Note that $d$ is the population sparsity level of $\bolbeta_k$'s, while $s$ is the user-specified sparsity level in Algorithms~\ref{alg1} and \ref{alg2}.} If we fix $d$, the dimension is allowed to grow at the rate $\log{p}=o(n\log^{-3}{n})$. When we allow $d$ to diverge, we require it to diverge more slowly than $\left\{{n}/{(\log{p}\log^3{n})}\right\}^{\frac{1}{2}}$. 

{We present the non-asymptotic results for Algorithm~\ref{alg1} in the following theorem, where the constants $D_1,D_2,\sigma_0,\gamma,C_0$ are defined previously in Theorem 1 under Condition (C1).}

\begin{theorem}\label{consistency}
Assume that Conditions (C1'), (C2) \& (C3) hold and $\bSigma_{\bX}=\bI$. Further assume that, there exists $\theta\in(0,1/2)$, such that for $k=1,\ldots, K$, we have $
(\widehat{\bolbeta}_k^0)^\T\bolbeta_k\ge 2\theta,$ and 
\begin{multline}
	\mu=\sqrt{[1+2\{(\frac{d}{s'})^{1/2}+\frac{d}{s'}\}]\{1-0.5\theta(1+\theta)(1-(\gamma^*)^2)\}}<1,
\end{multline}
where $\gamma^*=\dfrac{\lambda_K-\frac{3}{4}\Delta}{\lambda_K-\frac{1}{4}\Delta}$. Then there exists a positive integer $n_0=n_0(\sigma_0,C_0,q)<\infty$, $\gamma=\gamma(\sigma_0,C_0,q)\in (0,1/2)$ and a finite positive $D_0=D_0(\sigma_0,C_0,q)$ such that when $n> n_0$, we have $D_0n^{-\gamma}<\dfrac{\Delta}{4s}$ and for any
$D_0 n^{-\gamma}<\epsilon<\min\{\dfrac{\Delta}{4s},\theta\}$, with a probability greater than $1- 54p^2\exp\left\{-\dfrac{\epsilon^2 n}{36\log^3{n}}\right\}
$, 
\beq
|\sin\bTheta(\widehat\bbeta_k,\bbeta_k)|\le Cs\epsilon,\quad k=1,\dots,K.
\eeq
\end{theorem}

{Let $n^{-1/2} {(\log p)^{1/2}\log{^{3/2}n}} \ll\epsilon\ll d^{-1}$, then Theorem~\ref{consistency} directly implies the following asymptotic result that justifies the consistency of our estimator.}
\begin{corollary}\label{cor.consistency}
Assume that Conditions (C1'), (C2)--(C4) hold. Suppose there exists $\gamma>0$ such that $d \le s\ll \min\{n^{\gamma},\dfrac{n
^{\frac{1}{2}}}{(\log p)^{\frac{1}{2}}(\log{n})^{\frac{3}{2}}}\}$. Under the conditions in Theorem~\ref{consistency}, the quantities $|\sin\bTheta(\widehat\bbeta_k,\bbeta_k)|\rightarrow 0$ with a probability tending to 1, for $k=1,\ldots,K$. 
\end{corollary}

Corollary~\ref{cor.consistency} reveals that, without specifying a model, Algorithm~\ref{alg1} can achieve consistency when $p$ grows at an exponential rate of $n$. To be exact, we can allow $\log{p}=o\{n/(d^2\log^{3}{n})\}$. Here the theoretical results are established for the output of Algorithm~\ref{alg1}, instead of the solution of the optimization problem \eqref{Penalized Eigen Decompose}. Note that it is possible to have some gap between the theoretical optimal solution of \eqref{Penalized Eigen Decompose} and the estimate we use in practice, because the optimization problem is nonconvex, and numerically we might not achieve the global maximum.  Thus it might be more meaningful to study the property of the estimate obtained as the output of Algorithm~\ref{alg1}. The above theorem  guarantees that the estimate we use in practice has the desired theoretical properties. 

{In the meantime, although our rate in Theorem~\ref{consistency} is not as high as in the case of sparse sliced inverse regression, as established very recently by \citet{LZL18} when $\Sigma_X=I$ and for general $\Sigma_X$ in \citet{LZL19}, and \citet{tan2020sparse}, we have some unique advantages over these proposals. For simplicity, we assume that $d$ is fixed in the subsequent discussion. First, both sliced inverse regression methods require estimation of with-in-slice means rather than the MDDM. As shown in Theorem~\ref{th:main}, MDDM converges to its population counterpart at a slower rate than the sample with-in-slice mean. However, by adopting MDDM, we no longer need to determine the slicing scheme, and  we do not encounter the curse-of-dimensionality problem when slicing multivariate response. Second, \citet{LZL18} only achieves the optimal rate when $p=o(n^2)$, and  cannot handle ultra-high dimensions. In contrast, Algorithm~\ref{alg1} allows $p$ to diverge at an exponential rate of $n$, and is more suitable for ultra-high-dimensional data. Third, although \citet{tan2020sparse} achieves consistency when $\log{p}=o(n)$, they make much more restrictive model assumptions. For example, they assume that $Y$ is categorical and $\bX$ is normal within each slice of $Y$; they randomly split the dataset to form independent batches to facilitate their proofs, which is not done in their numerical studies. The theoretical properties for their proposal is unclear beyond the (conditionally) Gaussian model and without the sample splitting. In contrast, our method makes no model assumption between $\bX$ and $Y$, and our theory requires no sample splitting. Thus, our results are more widely applicable, and the rates we obtain seem hard to improve. Also, unlike the theory in \citet{tan2020sparse}, our theoretical result characterizes the exact method we use in practice. Moreover, the convergence rate of our method has an additional factor of $\log^3(n)$ compared to \citet{tan2020sparse}, which grows at a slow rate of $n$ that only imposes mild restriction on the dimensionality. For example, for any positive constant $\xi\in (0,1)$, if $\log{p}=O(n^{1-\xi})$, our method is consistent. In this sense, although we cannot handle the optimal dimensionality of $\log{p}=o(n)$, the gap is very small.

Next, we further consider the penalized generalized eigen-decomposition in Algorithm~\ref{alg2}. We assume that the step size $\eta$ satisfies $\eta\lambda_{\max}(\bSigma_{\mbX})<1/2$ and
\begin{equation}
	\sqrt{[1+2\{(\frac{d}{s'})^{1/2}+\frac{d}{s'}\}][1-\dfrac{\eta\lambda_{\min}(\bSigma_{\mbX})(1-\frac{\lambda_2}{\lambda_1})}{16\kappa(\bSigma_{\mbX})+16\frac{\lambda_2}{\lambda_1}}]} < 1,
\end{equation}
{where $\lambda_{\max}(\bolSigma_{\mbX})$, $\lambda_{\min}(\bolSigma_{\mbX})$ and $\kappa(\bSigma_{\mbX})$ are respectively the largest eigenvalue, the smallest eigenvalue and the condition number of $\bolSigma_{\mbX}$. The non-asymptotic results are as follows.}
\begin{theorem}\label{consistency.GEP}
Assume that Conditions (C1'), (C2) \& (C3) hold. Suppose there exists $\gamma\in (0,1/2)$ such that $d\le s=o(n^{\gamma})$, and there exists a constant $\theta(\kappa(\bSigma_{\bX}),\lambda_{\max}(\bSigma_{\bX}),\Delta, \lambda_1,\lambda_K,\eta)\in (0,1)$ such that $\dfrac{(\widehat\bbeta_k^0)^\T\bbeta_k}{\Vert\widehat\bbeta_k^0\Vert_2}\ge 1-\theta$. Then there exists a positive integer $n_0=n_0(s_0,C_0)<\infty$ and four finite positive constants $D_0=D_0(\gamma,\sigma_0,C_0)\in (0,\infty)$, $D_1=D_1(C_0)\in (0,\infty)$,
$D_2=D_2(\sigma_0,C_0)\in (0,\infty)$ and $\epsilon_0=\epsilon_0(\lambda_1,\lambda_2,\lambda_{\min}(\bSigma),\Delta)$ such that for any $\epsilon$ that satisfies $s\epsilon<\epsilon_0$ and $D_0n^{-\gamma}<\epsilon\le 1$, with a probability greater than $1-D_1p^2n\exp\{-D_2\epsilon^2n/\log^3{n}\}$, we have $|\sin\Theta(\widehat{\bbeta}_k,\bbeta_k)|\le Cs\epsilon$, for $k=1,\ldots,K$.
\end{theorem}

Theorem~\ref{consistency.GEP} is proved by showing that $\widehat{\bM}_k$ and $\widehat{\bSigma}_{\bX}$ are close to their counterparts in the sense that $\bu^\T\widehat{\bM}_k\bu,\bu^\T\widehat{\bSigma}_{\bX}\bu$ are close to $\bu^\T{\bM}_k\bu,\bu^\T{\bSigma}_{\bX}\bu$ for any $\bu$ with only $s$ nonzero elements, respectively. Then we use the fact that Algorithm~\ref{alg2} is a generalization of the RIFLE algorithm [\citet{TWLZ18}], and some  properties of the latter allow us to establish  the consistency of Algorithm~\ref{alg2}. By comparison, our proofs are significantly more involved than the one in \citet{TWLZ18}, because we have to estimate $K$ generalized eigenvectors instead of just the first one. We need to carefully control the error bounds to guarantee that the estimation errors do not accumulate to a higher order beyond the first generalized eigenvector.

Analogous to Corollary 1, we can easily obtain asymptotic consistency results by translating Theorem~\ref{consistency.GEP}. 

\begin{corollary}\label{cor.consistency.GEP}
Assume that Conditions (C1)--(C4) hold. Suppose there exists $\gamma\in (0,1/2)$ such that $d \le s\ll \min\{n^{\gamma},\dfrac{n
^{\frac{1}{2}}}{(\log p\log^3{n})^{\frac{1}{2}}}\}$. Under the conditions in Theorem~\ref{consistency.GEP}, the quantities $|\sin\bTheta(\widehat\bbeta_k,\bbeta_k)|\rightarrow 0$ with a probability tending to 1, for $k=1,\ldots,K$. 
\end{corollary}

Corollary~\ref{cor.consistency.GEP} shows that Algorithm~\ref{alg2} produces consistent estimates of the generalized eigenvectors $\beta_k$ even when $p$ grows at an exponential rate of the sample size $n$, and thus is suitable for ultra-high-dimensional problems. Similar to Corollary~\ref{cor.consistency}, Corollary~\ref{cor.consistency.GEP} has no gap between theory and the numerical outputs, as it is a result concerning the outputs of Algorithm~\ref{alg2}. We note that the dimensionality in Corollary~\ref{cor.consistency.GEP} is the same as that in Corollary~\ref{cor.consistency}. Thus, with a properly chosen step size $\eta$, the penalized generalized eigen-decomposition is intrinsically no more difficult than the penalized eigen-decomposition. But if we have knowledge about $\bSigma_{\bX}$ being the identity matrix, it is still beneficial to exploit such information and use Algorithm~\ref{alg1}, because Algorithm~\ref{alg1} does not involve the step size and is more convenient in practice. Also, although Algorithm~\ref{alg2} does not achieve the same rate of convergence as recent sparse sliced inverse regression proposals, it has many practical and theoretical advantages just as Algorithm~\ref{alg1}, which we do not repeat.

Finally, we note that our theoretical studies require conditions on the initial value. Specifically, we require the initial value to be non-orthogonal to the truth. This is a common technical condition for iterative algorithms; see \citet{yuan2013truncated,TWLZ18} for example.  Such conditions do not seem critical for our algorithms to work in practice. In our numerical studies to be presented, we use randomly generated initial values, and the performance of our methods appears to be   competitive.

\section{Numerical studies} \label{sec:numerical}

\subsection{Simulations}

We compare our slicing-free approaches to the state-of-the-art high-dimensional extensions of sliced inverse regression estimators. Both univariate response and multivariate response settings are considered. Specifically, for univariate response simulations, we include Rifle-SIR \citep{TWLZ18} and Lasso-SIR \citep{LZL19} as two main competitors; for multivariate response simulations, we mainly compare our method with the projective resampling approach to SIR \citep[PR-SIR,][]{PR}, which is a computationally expensive method that repeatedly projects the multivariate response to one-dimensional subspaces. For Rifle-SIR, we adopt the Rifle algorithm to estimate the leading eigenvector of the sample matrix $\cov\{\E(\mbX\mid Y)\}$ based on slicing. In addition, we include the oracle-SIR as a benchmark method, where we perform SIR on the subset of truly relevant variables (hence a low-dimensional estimation problem). For all these SIR-based methods, we include two different slicing schemes by setting the number of slices to be $3$ and $10$, where $3$ is the minimal number of slices required to obtain our two-dimensional central subspace and $10$ is a typical choice used in the literature. To evaluate the performances of these SDR methods, we use the subspace estimation error defined as $\mathcal{D}(\widehat{\bolbeta},\bolbeta) = \|\mbP_{\widehat{\bolbeta}} - \mbP_{\bolbeta}\|_F/\sqrt{2K}$, where $\hatbolbeta, \bolbeta\in\mbbR^{p\times K}$ are the estimated and the true basis matrices of the central subspace and $\mbP_{\hatbolbeta},\mbP_{\bolbeta}\in\mbbR^{p\times p}$ are the corresponding projection matrices. This subspace estimation error is always between $0$ and $1$, and a small value indicates a good estimation.





First, we consider the following six models for univariate response regression: $\mathcal{M}_1$ and $\mathcal{M}_2$ are single index models (i.e.~$K=1$), $\mathcal{M}_3$--$\mathcal{M}_5$ are multiple index models (i.e.~$K=2$), that are widely used in SDR literature; $\mathcal{M}_6$ is isotropic PFC model with $K=1$. Specifically,
\begin{align*}
&\mathcal{M}_1: Y = (\bolbeta_1^\T\mbX) + \sin(\bolbeta_1^\T\mbX) + \epsilon,\quad
\mathcal{M}_2: Y = 2\arctan(\bolbeta_1^\T\mbX) + 0.1(\bolbeta_1^\T\mbX)^3 + \epsilon,\\
&\mathcal{M}_3: Y = \frac{\bolbeta_1^\T\mbX}{0.5 + (1.5+\bolbeta_2^\T\mbX)^2} + 0.2\epsilon,\quad
\mathcal{M}_4: Y = \bolbeta_1^\T\mbX + (\bolbeta_1^\T\mbX)\cdot(\bolbeta_2^\T\mbX)+ 0.3\epsilon,\\
&\mathcal{M}_5: Y = sign(\bolbeta_1^\T\mbX)\cdot \log(|\bolbeta_2^\T\mbX + 5|) + 0.2\epsilon,\quad\mathcal{M}_6: \mbX = 2\bolbeta_1\exp(Y)/3 + 0.5\bolepsilon,
\end{align*}
where $\mbX \sim N_p(0,\bolSigma_{\mbX})$ and $\epsilon \sim N(0,1)$ for $\mathcal{M}_1$--$\mathcal{M}_5$, and $Y\sim N(0,1), \bolepsilon\sim N_p(0,\mbI_p)$ for the isotropic PFC model ($\mathcal{M}_6$). 
The sparse directions in the central subspace $\bolbeta_1,\bolbeta_2\in\mbbR^{p}$ are orthogonal as we let the first $s=6$ elements in $\bolbeta_1$ and the $6$-th to $12$-th elements in $\bolbeta_2$ to be $1/\sqrt{6}$ (while all other elements are zero). For  $\mathcal{M}_1$--$\mathcal{M}_5$, we consider both the independent predictor setting with $\bolSigma_{\mbX}=\mbI_p$ and the correlated predictor setting with auto-regressive correlation that $\Sigma_X(i,j) = 0.5^{|i-j|}$ for $i,j = 1,2,...,p$. 
For each of model settings, we vary the sample size $n \in  \{200,500,800\}$ and predictor dimension $p \in \{200,500,800,1200,2000\}$ and simulate 1000 independent data sets. 

For our method, we applied the generalized eigen-decomposition algorithm (Algorithm~\ref{alg2}) in all these six models (even when the covariance $\mbX$ is identity matrix). In the single index models $\mathcal{M}_1$ and $\mathcal{M}_2$, we use the random initialization ($\hatbolbeta^{(0)}$ is randomly generated from $p$-dimensional standard normal) for our algorithm and Rifle-SIR to demonstrate the robustness to initialization. The step size in the algorithm is simply fixed as $\eta=1$. For more challenging multiple index models, $\mathcal{M}_3 - \mathcal{M}_6$, we consider the best case scenarios for each method, therefore true parameter $\bolbeta$ is used as the initial value and an optimal $\eta\in\{0.1,0.2,\dots,1.0\}$ is selected from a separate training sample with 400 observations. The results based on 1000 replications for $n = 200$ and $p = 800$ are summarized in Table~\ref{tab:URM}, while the rest of the results can be found in the Supplemental Materials.
Overall, the slicing-free MDDM approach is much more accurate than existing SIR-based methods. It is almost as accurate as the oracle-SIR. Moreover, it is clear that SIR-type methods are rather sensitive to the choice of the number of slices.

\begin{table*}[t!]
	\resizebox{\textwidth}{!}{
		
		\begin{tabular}{cccc|cc|cc|cc|cc|cc|cc}
			\hline
			$\bolSigma_{\mbX}$&\multicolumn{1}{c}{}&  \multicolumn{2}{c}{MDDM}                           & \multicolumn{2}{c}{Oracle-SIR(3)}                         & \multicolumn{2}{c}{Oracle-SIR(10)}                         & \multicolumn{2}{c}{Rifle-SIR(3)}                    & \multicolumn{2}{c}{Rifle-SIR(10)}                                       & \multicolumn{2}{c}{LassoSIR(3)}                   & \multicolumn{2}{c}{LassoSIR(10)}                                         \\
			
			&\multicolumn{1}{c}{}    & \multicolumn{1}{c}{Error} & \multicolumn{1}{c}{SE} & \multicolumn{1}{c}{Error} & \multicolumn{1}{c}{SE} & \multicolumn{1}{c}{Error} & \multicolumn{1}{c}{SE} &  \multicolumn{1}{c}{Error} & \multicolumn{1}{c}{SE} & \multicolumn{1}{c}{Error} & \multicolumn{1}{c}{SE} & \multicolumn{1}{c}{Error} & \multicolumn{1}{c}{SE} & \multicolumn{1}{c}{Error} & \multicolumn{1}{c}{SE} \\\cline{2-16}
			\multirow{5}{*}{$\mbI_p$}&{$\mathcal{M}_1$} & 10.1 & 0.1 & 12.5 & 0.1 & 10.3 & 0.1 & 25.2 & 1.0 & 53.7 & 1.4 & 37.9 & 0.4 & 59.9 & 0.7  \\
			
			&{$\mathcal{M}_2$}&10.3 & 0.1 & 13.1 & 0.1 & 10.6 & 0.1 & 26.1 & 1.0 & 54.7 & 1.4 & 40.1 & 0.4 & 61.5 & 0.7  \\
			
			&$\mathcal{M}_3$ &17.7 & 0.2 & 40.8 & 0.2 & 27.7 & 0.2 & 71.3 & 0.0 & 71.2 & 0.0 & 76.5 & 0.2 & 85.0 & 0.2\\
			
			&$\mathcal{M}_4$&23.0 & 0.2 & 45.8 & 0.3 & 36.4 & 0.3 & 71.9 & 0.0 & 71.6 & 0.0 & 85.2 & 0.2 & 91.5 & 0.2 \\
			
			&$\mathcal{M}_5$&30.8 & 0.6 & 28.8 & 0.2 & 22.1 & 0.1 & 71.6 & 0.0 & 71.2 & 0.0 & 71.2 & 0.3 & 81.3 & 0.3\\\cline{2-16} 
			\multirow{6}{*}{AR}&$\mathcal{M}_1$&18.7 & 0.3 & 21.0 & 0.2 & 17.6 & 0.2 & 34.7 & 0.8 & 39.8 & 1.1 & 35.3 & 0.3 & 35.5 & 0.3  \\ 
			&$\mathcal{M}_2$&14.2 & 0.2 & 20.7 & 0.2 & 14.8 & 0.2 & 33.1 & 0.7 & 33.6 & 1.1 & 34.6 & 0.3 & 30.5 & 0.3  \\
			&$\mathcal{M}_3$&25.2 & 0.3 & 44.6 & 0.2 & 34.1 & 0.2 & 71.5 & 0.0 & 71.3 & 0.0 & 54.8 & 0.2 & 47.1 & 0.3 \\
			&$\mathcal{M}_4$&59.1 & 0.5 & 75.1 & 0.2 & 69.9 & 0.3 & 81.0 & 0.2 & 78.7 & 0.2 & 89.7 & 0.2 & 92.1 & 0.2 \\
			&$\mathcal{M}_5$ &46.2 & 0.6 & 46.4 & 0.2 & 35.5 & 0.2 & 73.8 & 0.1 & 72.4 & 0.0 & 66.5 & 0.2 & 61.4 & 0.3\\\cline{2-16}
			PFC &$\mathcal{M}_{6}$&34.6&0.6&48.9&0.5&33.4&0.5&40.1&0.7&30.8&0.6&70.7&0.0&70.7&0.0\\\hline
			
		\end{tabular}}
	
	\caption{ Averaged subspace estimation errors and the corresponding standard errors (after multiplied by $100$) for univariate response models ($n = 200, p = 800$).} \label{tab:URM}
\end{table*}

Next, we further consider the following three multivariate response models, where the response dimension is $q=4$. These three models are respectively a multivariate linear model, a single-index heteroschedastic error model, and an isotropic PFC model. The predictors satisfy $\mbX  \sim N_p(0,\mbI_p)$ in the following two forward regression model. Therefore, we applied Algorithm~\ref{alg1} for our method under models $\mathcal{M}_7$ and $\mathcal{M}_8$. For the isotropic PFC model $\mathcal{M}_9$, where $\mbX\mid\mbY  \sim N_p(\bolbeta f(\mbY),\mbI_p)$, we still apply Algorithm~\ref{alg2} to be consistent with the univariate case. For the projective resampling methods, PR-SIR and PR-Oracle-SIR, we generated a sufficiently large number of  $n\log(n)$ random projections so that the PR methods reach their fullest potential. 
\begin{enumerate}
	\item[$\mathcal{M}_7$]: $Y_1 = \bolbeta_1^{\T}X + \epsilon_1$, $Y_2 = \bolbeta_2^\T\mbX + \epsilon_2$,  $Y_3 = \epsilon_3$ and $Y_4 = \epsilon_4$. The errors $(\epsilon_1,\dots,\epsilon_4)$ are independent standard normal except for $\cov(\epsilon_1,\epsilon_2)=-0.5$. For this model, the central subspace is spanned by $\bolbeta_1 = (1,0,0,0,...,0)^\T$ and $\bolbeta_2 = (0,2,1,0,...,0)^\T$.
	\item[$\mathcal{M}_8$]: $Y_1 = \exp(\epsilon_1)$ and $Y_i = \epsilon_i$ for $i = 2,3,4$, where $(\epsilon_1,\dots,\epsilon_4)$ are independent standard normal except for $\cov(\epsilon_1,\epsilon_2)=\sin(\bolbeta^{\T}\mbX)$. For this model, the central subspace is $\bolbeta = (0.8,0.6,0,0,\dots,0)^\T$. Note that marginally each response is independent of $\mbX$.
	\item[$\mathcal{M}_{9}$]: $\mbX = \bolbeta\left(\frac{1}{3}\sin(Y_1) + \frac{2}{3}\exp(Y_2) + Y_3\right) + \bolepsilon$, where $\bolbeta=(1/\sqrt{6}\cdot \mathrm{1}_6,\mathrm{0}_{p-6})$ , and $\bolepsilon\sim N(0,\mbI_p)$. Hence, $\calS_{\mbY\mid\mbX}=\spn(\bolbeta)$.
\end{enumerate}
Again we considered various sample size and predictor dimension setups, each with 1000 replicates. We summarize the subspace estimation errors in the Table~\ref{tab:MRM}. For $p = 800$ and $1200$, the results are gathered in the supplement. It is clear that the proposed MDDM approach is much better than PR-SIR, and also improves much faster than PR-SIR when we increase the sample size. The MDDM method performed better in inverse regression models such as the isotropic PFC model than forward regression models such as the linear model and index models. This finding is more apparent in the multivariate response simulations than in the univariate response simulations. This is expected, as the MDDM directly targets at the inverse regression subspace, which is more directly driven by the response in the isotropic PFC models.

\begin{table*}[t!]
	\centering
	\resizebox{\textwidth}{!}{
		\begin{tabular}{cccccccc|cccccc|cccccc}
			\hline
			&        & \multicolumn{6}{c}{$n = 100$}        & \multicolumn{6}{c}{$n = 200$}  & \multicolumn{6}{c}{$n = 400$} \\  \cline{3-20}
			&        &\multicolumn{2}{c}{$p = 100$}&\multicolumn{2}{c}{$p = 200$}&\multicolumn{2}{c}{$p = 400$}&\multicolumn{2}{c}{$p = 100$}&\multicolumn{2}{c}{$p = 200$}&\multicolumn{2}{c}{$p = 400$}&\multicolumn{2}{c}{$p = 100$}&\multicolumn{2}{c}{$p = 200$}&\multicolumn{2}{c}{$p = 400$} \\
			&&Error&SE&Error&SE&Error&SE&Error&SE&Error&SE&Error&SE&Error&SE&Error&SE&Error&SE\\\cline{2-20}
			\multirow{6}{*}{$\mathcal{M}_7$}& MDDM   & 37.1 & 0.5 & 39.8 & 0.5 & 42.5 & 0.5 & 24.0 & 0.4 & 25.3 & 0.4 & 26.9 & 0.4 & 16.1 & 0.3 & 17.3 & 0.3 & 18.6 & 0.3  \\  
			& PR-Oracle-SIR(3) &12.6&0.2&12.2&0.2&12.0&0.2&8.8&0.1&8.5&0.1&8.7&0.1&5.9&0.1&5.8&0.1&5.8&0.1\\
			& PR-Oracle-SIR(10) & 16.2&0.3&15.7&0.3&15.6&0.3&9.6&0.2&9.4&0.2&95.2&0.2&6.0&0.1&6.0&0.1&6.0&0.1\\
			& PR-SIR(3) & 79.9 & 0.1 & 88.2 & 0.1 & 93.5 & 0.0 & 67.9 & 0.1 & 79.3 & 0.1 & 87.8 & 0.0 & 54.6 & 0.1 & 67.6 & 0.1 & 79.0 & 0.0  \\  
			& PR-SIR(10) & 83.5 & 0.1 & 90.6 & 0.1 & 94.9 & 0.0 & 70.1 & 0.1 & 81.6 & 0.1 & 90.1 & 0.1 & 55.3 & 0.1 & 68.2 & 0.1 & 80.2 & 0.1  \\  \cline{2-20} 
			\multirow{6}{*}{$\mathcal{M}_8$}& MDDM   &79.4 & 0.9 & 85.8 & 0.8 & 90.0 & 0.7 & 55.9 & 1.2 & 61.0 & 1.2 & 68.4 & 1.2 & 27.1 & 0.9 & 30.3 & 1.0 & 31.0 & 1.0  \\ 
			& PR-Oracle-SIR(3) &40.9&0.9&41.3&0.9&41.4&0.9&26.0&0.7&24.9&0.7&25.0&0.6&14.9&0.4&14.9&0.4&15.0&0.4\\
			& PR-Oracle-SIR(10) & 44.1&0.9&43.8&0.9&43.5&0.9&25.1&0.6&23.7&0.6&24.1&0.6&13.1&0.3&13.0&0.3&13.2&0.3\\ 
			& PR-SIR(3) & 99.3 & 0.0 & 99.7 & 0.0 & 99.8 & 0.0 & 99.2 & 0.0 & 99.7 & 0.0 & 99.8 & 0.0 & 98.8 & 0.0 & 99.6 & 0.0 & 99.8 & 0.0 \\  
			& PR-SIR(10) & 99.3 & 0.0 & 99.7 & 0.0 & 99.9 & 0.0 & 99.1 & 0.0 & 99.6 & 0.0 & 99.8 & 0.0 & 98.4 & 0.1 & 99.6 & 0.0 & 99.8 & 0.0  \\ \cline{2-20} 
			\multirow{6}{*}{$\mathcal{M}_{9}$}& MDDM   & 15.3 & 0.3 & 15.4 & 0.3 & 15.7 & 0.3 & 9.9 & 0.1 & 10.1 & 0.1 & 10.0 & 0.1 & 7.1 & 0.1 & 7.2 & 0.1 & 7.1 & 0.1 \\
			& PR-Oracle-SIR(3) &15.2 & 0.2 & 15.2 & 0.2 & 14.9 & 0.2 & 10.5 & 0.1 & 10.6 & 0.1 & 10.5 & 0.1 & 7.5 & 0.1 & 7.6 & 0.1 & 7.4 & 0.1 \\ 
			& PR-Oracle-SIR(10) &13.8 & 0.2 & 13.9 & 0.2 & 13.6 & 0.2 & 9.4 & 0.1 & 9.7 & 0.1 & 9.6 & 0.1 & 6.8 & 0.1 & 6.8 & 0.1 & 6.7 & 0.1 \\ 
			& PR-SIR(3) & 58.5 & 0.2 & 72.3 & 0.2 & 84.0 & 0.2 & 44.6 & 0.1 & 58.2 & 0.1 & 71.4 & 0.1 & 33.1 & 0.1 & 44.6 & 0.1 & 57.9 & 0.1 \\
			& PR-SIR(10) & 54.8 & 0.2 & 68.5 & 0.2 & 80.6 & 0.2 & 41.1 & 0.2 & 54.3 & 0.2 & 67.7 & 0.2 & 30.2 & 0.1 & 41.0 & 0.1 & 54.2 & 0.1 \\   \hline			
	\end{tabular}}
	\caption{Averaged subspace estimation errors and the corresponding standard errors (after multiplied by $100$) for multivariate response models.} \label{tab:MRM}
\end{table*}

\subsection{Real Data Illustration}

	\begin{figure*}[ht!]
		\centering
		\includegraphics[scale = 0.5]{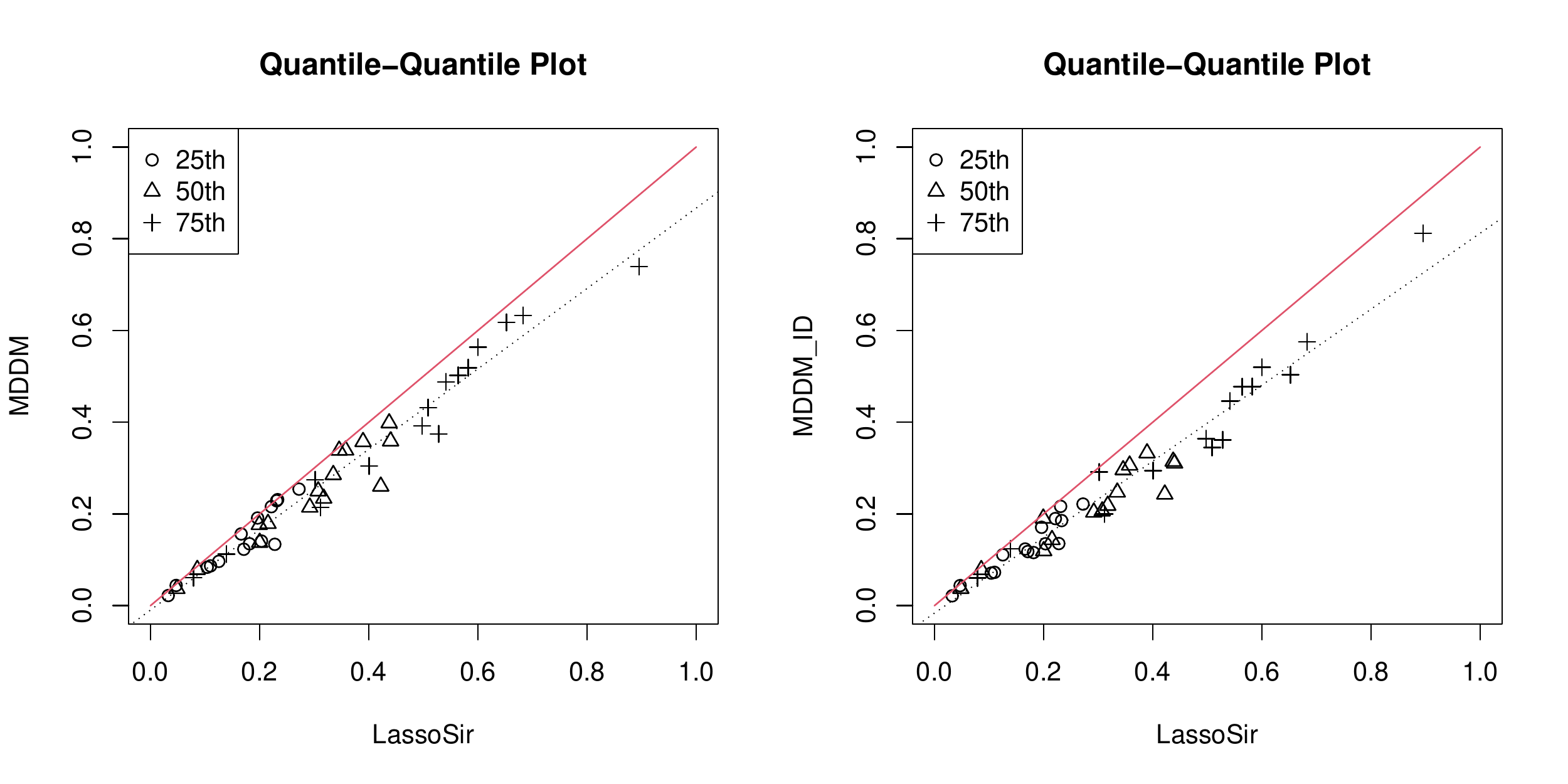}
		\caption{Quantile-quantile plots for prediction error comparisons between MDDM and Lasso-SIR (left panel), and between MDDM-ID and Lasso-SIR (right panel). Each point corresponds to the prediction mean squared errors for one of the $q=15$ response variables, where different shapes represents different quantiles. }\label{fig:quantile}
	\end{figure*}
	
	
	\begin{figure}[ht!]
		\centering
		\includegraphics[scale = 0.5]{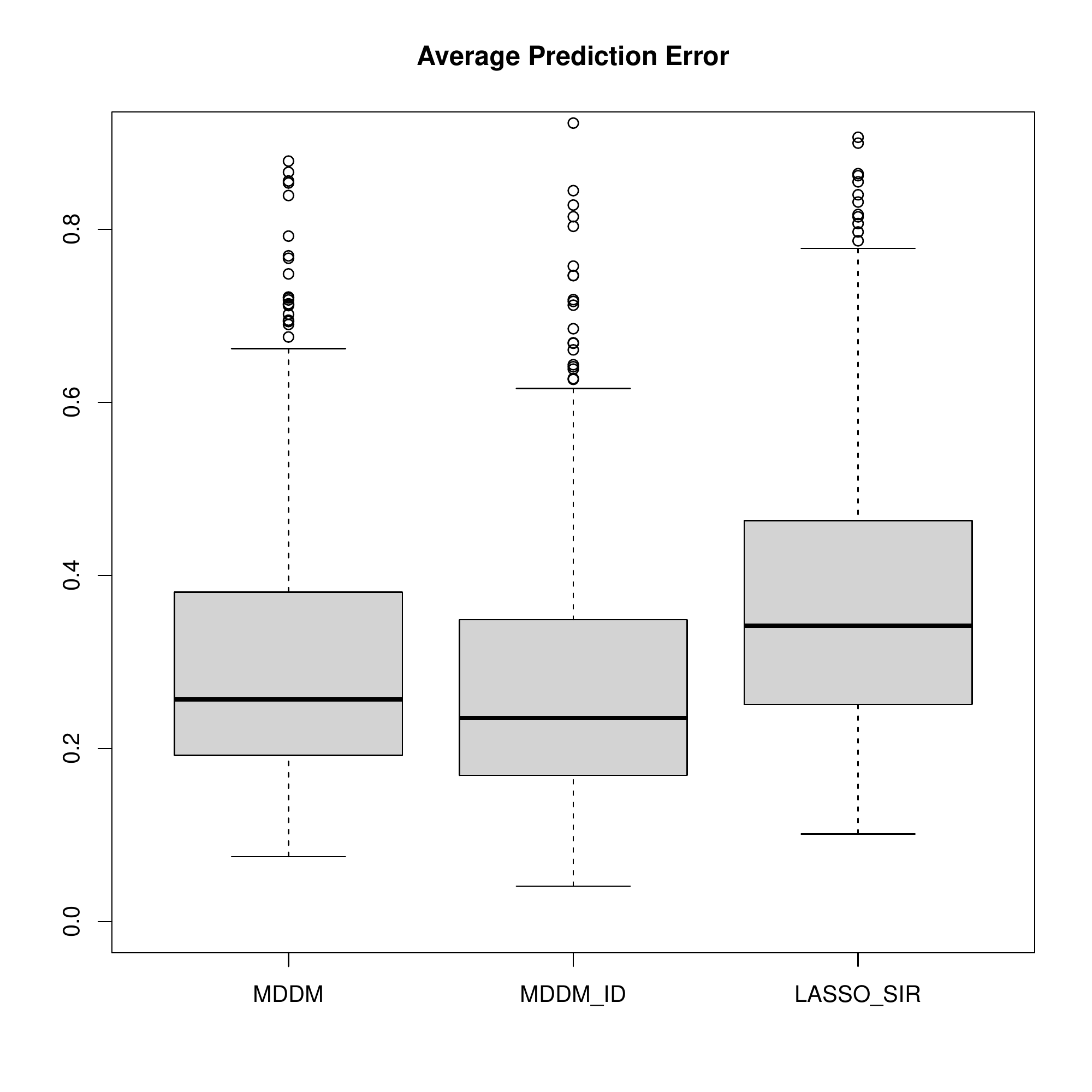}
		\caption{The averaged prediction error over $500$ training-testing sample splits and over $q=15$ response variables. }\label{fig:avgbox}
	\end{figure}

	In this section we use our method to analyze the NCI-60 data set \citep{NCI60} that contains the microRNA expression profiles and cancer drug activity measurements on the NCI-60 cell lines. The multivariate response is the cancer drug activities of $q = 15$ drugs; the predictor is $p = 365$ different microRNA; the sample size is $n = 60$.
	
	First, we examine the predictive performance of our method by $500$ random training-testing sample splits; each time we randomly pick $5$ observations to form the test set. We consider $K=5$ for all methods. For MDDM, we included both the eigen-decomposition (Algorithm~\ref{alg1}) and the generalized eigen-decomposition (Algorithm~\ref{alg2}). To distinguish the two versions of MDDM, we have ``MDDM-ID'' for the eigen-decomposition approach because it implicitly assumes the covariance of $\mbX$ or the conditional covariance of $\mbX\mid\mbY$ is constant times identity matrix. We use random initial values, and choose the sparsity level to be $s=25$ in the way described in Section~S2 in the Supplementary Materials.
  Then the five reduced predictors $\bolbeta_k^{\T}\mbX$, $k=1,\dots,5$, are fed into a generalized additive model for each drug. Finally, we evaluate the mean squared prediction error based on the test sample. The Rifle-SIR can only estimate a one-dimensional subspace, which did not yield accurate prediction in this data set. Hence for comparison, we compute five leading directions from Lasso-SIR. The $25$th, $50$th and $75$th percentiles of the squared prediction errors for each of the 15 responses for all three models are obtained and we construct quantile-quantile plots in Figure \ref{fig:quantile}. The red line is the $y=x$ line, and the black dashed line is a simple linear regression fit for the results indicated by the y-axis label against that indicated by the x-axis. Clearly, for all the quantiles and for all the response variables, the MDDM results (MDDM or MDDM-ID) are better than Lasso-SIR in terms of prediction. In addition, we construct side-by-side boxplots of the prediction error averaged over all response variables in Figure \ref{fig:avgbox} to evaluate the overall improvement. Interestingly, the MDDM-ID is slightly better than the MDDM approach. This is likely due to the small sample size -- with only $55$ training sample,  the sample covariance of $p=365$ variables is difficult to estimate accurately. We further include additional real data analysis results in Section S4 in Supplementary Materials.


\section{Discussion}\label{sec:discussion}

In this paper, we propose a slicing-free high-dimensional SDR method by a penalized eigen-decomposition of sample MDDM. Our proposal is motivated by the usefulness of MDDM for dimension reduction and yields a relatively straightforward implementation in view of the recently developed  RIFLE algorithm \citep{TWLZ18} by simply replacing the slicing-based estimator with sample MDDM. Our methodology and implementation involve no slicing and treats the univariate and multivariate response in a unified fashion. Both theoretical support and finite sample investigations provide convincing evidence that MDDM is a very competitive alternative compared to SIR and may be used as a surrogate to SIR-based estimator routinely for many related sufficient dimension reduction problems.

As with most SDR methods, our proposal requires the linearity condition, the violation of which can make SDR very challenging to tackle.  It seems that existing proposals that relax the linearity condition are often practically difficult due to excessive computational costs, and cannot be easily extended to high dimensions \citep{cook1994reweighting,ma2012semiparametric}. One potentially useful approach is to transform data before SDR to alleviate obvious violation of the linearity assumption \citep{mai2015nonparametric}. But an in-depth study along this line is beyond the scope of the current paper. In addition, we observe from our simulation studies that RIFLE requires the choice of several tuning parameters, such as the step size and the initial value, and the optimization  error could depend on these tuning parameters in a nontrivial way.  Further investigation on the optimization error and data-driven choice for these tuning parameters would be desirable and is left for future research. 

As pointed out by a referee, many SDR methods beyond SIR involve slicing. It will be interesting to study how to perform them in a slicing-free fashion as well. For example, \citet{SAVE} attempt to perform dimension reduction by estimating the conditional covariance of $\bX$, while \citet{yc03} consider the conditional third moment. These methods slice the response to estimate the conditional moments. In the future, one can develop slicing-free methods to estimate these higher-order moments and conduct SDR.

\section*{Supplementary Materials}
Supplement to ``Slicing-free Inverse Regression in High-dimensional Sufficient Dimension Reduction''. In the supplement, we present additional simulation results and proofs.
\section*{Acknowledgements}
The authors are grateful to the Editor, Associate Editor and anonymous referees, whose suggestions led great improvement of this work.
The authors contributed equally to this work and are listed in alphabetical order. Mai and Zhang's research in this article is supported in part by National Science Foundation grants CCF-1908969. Shao's research in this article is supported in part by National Science Foundation grants  DMS-1607489. 



\baselineskip=15pt
\bibliography{ref}
\bibliographystyle{agsm}

\vskip .65cm
\noindent
Qing Mai, Florida State University
\vskip 2pt
\noindent
E-mail: qmai@fsu.edu
\vskip 2pt

\noindent
Xiaofeng Shao, University of Illinois at Urbana-Champaign
\vskip 2pt
\noindent
E-mail: xshao@illinois.edu

\noindent
Runmin Wang, Texas A\&M University
\vskip 2pt
\noindent
E-mail: runminw@tamu.edu

\noindent
Xin Zhang, Florida State University
\vskip 2pt
\noindent
E-mail: xzhang8@fsu.edu

\clearpage
\newpage


\global\long\def\theequation{S\arabic{equation}}%
\global\long\def\thetheorem{S\arabic{theorem}}%
\global\long\def\thesection{S\arabic{section}}%
\global\long\def\thetable{S\arabic{table}}%
\global\long\def\thefigure{S\arabic{figure}}%
\renewcommand{\thelemma}{S\arabic{lemma}}
\renewcommand{\theproposition}{S\arabic{proposition}}

\renewcommand{\theassumption}{S\arabic{assumption}}

\setcounter{equation}{0}
\setcounter{section}{0}
\setcounter{figure}{0}
\setcounter{table}{0}

\renewcommand{\baselinestretch}{1.5}
\fontsize{12}{14.5}\selectfont

\title{{Supplementary Materials for ``Slicing-free Inverse Regression in  High-Dimensional Sufficient Dimension Reduction''}}
\author{Qing Mai, Xiaofeng Shao, Runmin Wang and Xin Zhang}
\maketitle

In the supplement, we present some additional simulation results in Section~\ref{sec:simulation}, additional real data analysis results in Section~\ref{sec:realdata},  some discussion of computational complexity in Section~\ref{sec:computation}, 
proofs of Proposition~\ref{prop1} \& Lemma~\ref{sub}
in Section~\ref{sec:proplemma}, proofs of Theorem~\ref{th:main} in Section~\ref{sec:theorem1}, proofs of Theorem~\ref{consistency} \& \ref{consistency.GEP} in Section~\ref{sec:consistency.proof}. 

	\section{Additional Simulation Results} \label{sec:simulation} 
 In this section, we shall present all simulation results for models described in Section \ref{sec:numerical}, except the ones which have already been presented in the main paper. Specifically, Table \ref{tab:single1} and \ref{tab:single2} contain the results for single index models ($\mathcal{M}_1$ and 
 $\mathcal{M}_2$), Table \ref{tab:multi1} and \ref{tab:multi2} contain results for multiple index models ($\mathcal{M}_3$ - $\mathcal{M}_5$). Results for PFC model (univariate response) $\mathcal{M}_6$ are summarized in Table \ref{tab:pfc}, and we gather the rest of 
 results for models with multivariate response variables ($\mathcal{M}_7$-$\mathcal{M}_9$) in Table \ref{tab:MRM}.

 The patterns are similar to those presented in the paper. Overall the newly proposed method outperforms the competitors in most scenarios, especially when the dimensionality is significantly larger than the sample size (i.e., high-dimensional setting). It is also observed that SIR-based methods are rather sensitive to the choice of number of slices, whereas our method is slicing-free and  is thus easier to use in practice.
	
	\begin{table}[h!]
	\resizebox{\textwidth}{!}{
		
		\begin{tabular}{ccccc|cc|cc|cc|cc|cc|cc}
			\hline
			\multicolumn{3}{c}{}&  \multicolumn{2}{c}{MDDM}                           & \multicolumn{2}{c}{Oracle-SIR(3)}                         & \multicolumn{2}{c}{Oracle-SIR(10)}                         & \multicolumn{2}{c}{Rifle-SIR(3)}                    & \multicolumn{2}{c}{Rifle-SIR(10)}                                       & \multicolumn{2}{c}{LassoSIR(3)}                   & \multicolumn{2}{c}{LassoSIR(10)}\\
			
			\multicolumn{3}{c}{}    & \multicolumn{1}{c}{Error} & \multicolumn{1}{c}{SE} & \multicolumn{1}{c}{Error} & \multicolumn{1}{c}{SE} & \multicolumn{1}{c}{Error} & \multicolumn{1}{c}{SE} &  \multicolumn{1}{c}{Error} & \multicolumn{1}{c}{SE} & \multicolumn{1}{c}{Error} & \multicolumn{1}{c}{SE} & \multicolumn{1}{c}{Error} & \multicolumn{1}{c}{SE} & \multicolumn{1}{c}{Error} & \multicolumn{1}{c}{SE} \\\cline{2-17}
			\multirow{15}{*}{$\mathcal{M}_1$} & \multirow{5}{*}{n = 200} & p = 200 &10.3 & 0.1 & 12.6 & 0.1 & 10.5 & 0.1 & 16.4 & 0.6 & 30.0 & 1.2 & 28.5 & 0.2 & 30.6 & 0.3 \\ 
			&                          & p = 500 &10.4 & 0.1 & 12.7 & 0.1 & 10.6 & 0.1 & 22.1 & 0.9 & 47.1 & 1.4 & 34.6 & 0.3 & 47.7 & 0.6\\
			&                          & p = 800 &10.1 & 0.1 & 12.5 & 0.1 & 10.3 & 0.1 & 25.2 & 1.0 & 53.7 & 1.4 & 37.9 & 0.4 & 59.9 & 0.7\\
   &                          & p = 1200 &10.0 & 0.1 & 12.4 & 0.1 & 10.3 & 0.1 & 27.5 & 1.0 & 54.3 & 1.4 & 42.4 & 0.5 & 71.4 & 0.7  \\ 
   &                          & p = 2000 & 10.1 & 0.1 & 12.6 & 0.1 & 10.5 & 0.1 & 29.7 & 1.1 & 63.4 & 1.4 & 48.4 & 0.6 & 81.7 & 0.6\\ 
   \cline{4-17}
			& \multirow{5}{*}{n = 500} & p = 200 &6.3 & 0.1 & 7.6 & 0.1 & 6.3 & 0.1 & 8.9 & 0.4 & 12.1 & 0.7 & 15.0 & 0.1 & 12.8 & 0.1 \\
			&                          & p = 500 &6.4 & 0.1 & 7.9 & 0.1 & 6.5 & 0.1 & 11.8 & 0.6 & 21.6 & 1.1 & 16.3 & 0.1 & 14.6 & 0.1 \\ 
			&                          & p = 800 &6.2 & 0.1 & 7.7 & 0.1 & 6.4 & 0.1 & 13.2 & 0.7 & 26.1 & 1.2 & 16.6 & 0.1 & 16.2 & 0.2 \\
   	&                          & p = 1200 &6.4 & 0.1 & 7.6 & 0.1 & 6.4 & 0.1 & 13.4 & 0.7 & 30.3 & 1.3 & 17.3 & 0.2 & 17.9 & 0.2\\
    	&                          & p = 2000 & 6.3 & 0.1 & 7.7 & 0.1 & 6.3 & 0.1 & 14.9 & 0.8 & 32.9 & 1.3 & 18.3 & 0.2 & 21.8 & 0.3\\\cline{4-17}
			& \multirow{5}{*}{n =800}  & p = 200 &5.0 & 0.1 & 6.0 & 0.1 & 5.0 & 0.1 & 7.4 & 0.4 & 8.7 & 0.6 & 11.1 & 0.1 & 9.3 & 0.1\\
			&                          & p = 500 & 5.2 & 0.1 & 6.1 & 0.1 & 5.1 & 0.1 & 8.3 & 0.4 & 12.9 & 0.8 & 11.9 & 0.1 & 10.1 & 0.1\\ 
			&                          & p = 800 &5.1 & 0.1 & 6.1 & 0.1 & 5.1 & 0.1 & 9.5 & 0.6 & 20.1 & 1.1 & 12.4 & 0.1 & 11.1 & 0.1\\
   	&                          & p = 1200 &5.1 & 0.1 & 6.1 & 0.1 & 5.1 & 0.1 & 9.5 & 0.6 & 20.1 & 1.1 & 12.4 & 0.1 & 11.1 & 0.1\\
    	&                          & p = 2000 & 4.9 & 0.1 & 6.1 & 0.1 & 5.0 & 0.1 & 10.6 & 0.6 & 23.1 & 1.2 & 12.8 & 0.1 & 12.2 & 0.1\\\cline{2-17}
			\multirow{15}{*}{$\mathcal{M}_2$} & \multirow{5}{*}{n = 200} & p = 200 &10.4 & 0.1 & 13.1 & 0.1 & 10.7 & 0.1 & 17.5 & 0.6 & 29.7 & 1.2 & 30.3 & 0.2 & 31.6 & 0.3 \\ 
			&                          & p = 500 &10.6 & 0.1 & 13.3 & 0.1 & 10.8 & 0.1 & 23.8 & 0.9 & 48.9 & 1.4 & 36.7 & 0.3 & 49.9 & 0.6\\ 
			&                          & p = 800 &10.3 & 0.1 & 13.1 & 0.1 & 10.6 & 0.1 & 26.1 & 1.0 & 54.7 & 1.4 & 40.1 & 0.4 & 61.5 & 0.7 \\
   && p = 1200 &54.7 & 0.8 & 12.9 & 0.1 & 10.4 & 0.1 & 74.4 & 0.8 & 95.1 & 0.4 & 45.0 & 0.5 & 71.4 & 0.7 \\ 

   && p = 2000 &55.3 & 0.8 & 13.1 & 0.1 & 10.6 & 0.1 & 76.5 & 0.7 & 96.8 & 0.3 & 51.2 & 0.6 & 82.7 & 0.6\\ 
\cline{4-17}
			& \multirow{5}{*}{n = 500} & p = 200 &6.4 & 0.1 & 8.0 & 0.1 & 6.5 & 0.1 & 9.7 & 0.4 & 12.0 & 0.7 & 15.8 & 0.1 & 13.4 & 0.1 \\ 
			&                          & p = 500 &6.7 & 0.1 & 8.2 & 0.1 & 6.6 & 0.1 & 12.4 & 0.6 & 21.2 & 1.1 & 17.1 & 0.1 & 15.1 & 0.1 \\ 
			&                          & p = 800 &6.5 & 0.1 & 8.0 & 0.1 & 6.5 & 0.1 & 13.6 & 0.7 & 25.1 & 1.2 & 17.6 & 0.2 & 16.7 & 0.2 \\
   			&                          & p = 1200 &17.1 & 0.7 & 8.0 & 0.1 & 6.5 & 0.1 & 37.4 & 1.1 & 74.6 & 1.0 & 18.2 & 0.2 & 18.3 & 0.2\\ 
			&                          & p = 2000 &16.9 & 0.8 & 8.0 & 0.1 & 6.5 & 0.1 & 38.3 & 1.2 & 77.4 & 0.9 & 19.0 & 0.2 & 22.2 & 0.3\\ 
\cline{4-17}
			& \multirow{5}{*}{n =800}  & p = 200 &5.1 & 0.1 & 6.3 & 0.1 & 5.1 & 0.1 & 6.9 & 0.3 & 8.3 & 0.5 & 11.6 & 0.1 & 9.6 & 0.1\\ 
			&                          & p = 500 &5.3 & 0.1 & 6.4 & 0.1 & 5.2 & 0.1 & 7.8 & 0.4 & 13.1 & 0.8 & 12.5 & 0.1 & 10.5 & 0.1\\
			&                          & p = 800 & 5.0 & 0.1 & 6.3 & 0.1 & 5.1 & 0.1 & 8.4 & 0.4 & 16.9 & 1.0 & 12.6 & 0.1 & 10.8 & 0.1 \\
   			&                          & p = 1200 &10.8 & 0.6 & 6.4 & 0.1 & 5.2 & 0.1 & 23.9 & 1.0 & 53.9 & 1.3 & 13.1 & 0.1 & 11.4 & 0.1 \\
			&                          & p = 2000 &11.3 & 0.6 & 6.3 & 0.1 & 5.1 & 0.1 & 26.6 & 1.1 & 57.6 & 1.2 & 13.6 & 0.1 & 12.4 & 0.1 \\\hline
		\end{tabular}
	}
	\caption{$d(V,\hat{V})$ and corresponding standard errors (in $10^{-2}$) for single index models (identity variance) \label{tab:single1}}
\end{table}

\begin{table}[h!]
	\resizebox{\textwidth}{!}{
		
		\begin{tabular}{ccccc|cc|cc|cc|cc|cc|cc}
			\hline
			\multicolumn{3}{c}{}&  \multicolumn{2}{c}{MDDM}                           & \multicolumn{2}{c}{Oracle-SIR(3)}                         & \multicolumn{2}{c}{Oracle-SIR(10)}                         & \multicolumn{2}{c}{Rifle-SIR(3)}                    & \multicolumn{2}{c}{Rifle-SIR(10)}                                       & \multicolumn{2}{c}{LassoSIR(3)}                   & \multicolumn{2}{c}{LassoSIR(10)}  \\
			
			\multicolumn{3}{c}{}    & \multicolumn{1}{c}{Error} & \multicolumn{1}{c}{SE} & \multicolumn{1}{c}{Error} & \multicolumn{1}{c}{SE} & \multicolumn{1}{c}{Error} & \multicolumn{1}{c}{SE} &  \multicolumn{1}{c}{Error} & \multicolumn{1}{c}{SE} & \multicolumn{1}{c}{Error} & \multicolumn{1}{c}{SE} & \multicolumn{1}{c}{Error} & \multicolumn{1}{c}{SE} & \multicolumn{1}{c}{Error} & \multicolumn{1}{c}{SE} \\\cline{2-17}
			\multirow{15}{*}{$\mathcal{M}_1$} & \multirow{5}{*}{n = 200} & p = 200 &17.6 & 0.2 & 20.8 & 0.2 & 17.8 & 0.2 & 26.2 & 0.5 & 25.3 & 0.7 & 32.1 & 0.3 & 28.4 & 0.3\\ 
			&                          & p = 500 &18.3 & 0.3 & 21.1 & 0.2 & 18.0 & 0.2 & 32.9 & 0.7 & 35.0 & 1.0 & 34.5 & 0.3 & 33.1 & 0.3 \\ 
			&                          & p = 800 &18.7 & 0.3 & 21.0 & 0.2 & 17.6 & 0.2 & 34.7 & 0.8 & 39.8 & 1.1 & 35.3 & 0.3 & 35.5 & 0.3 \\
   && p = 1200 &26.1 & 0.6 & 21.3 & 0.2 & 18.0 & 0.2 & 47.0 & 1.0 & 81.2 & 0.7 & 36.7 & 0.3 & 41.2 & 0.4  \\ 
   &  & p = 2000 &26.3 & 0.7 & 21.0 & 0.2 & 17.7 & 0.2 & 47.8 & 1.0 & 81.4 & 0.7 & 39.1 & 0.4 & 49.6 & 0.6 \\ \cline{4-17}
			& \multirow{5}{*}{n = 500} & p = 200 &10.8 & 0.1 & 13.2 & 0.1 & 11.0 & 0.1 & 13.7 & 0.2 & 12.7 & 0.4 & 18.6 & 0.2 & 15.3 & 0.1  \\ 
			&                          & p = 500 &11.0 & 0.1 & 13.3 & 0.1 & 11.2 & 0.1 & 14.3 & 0.3 & 15.7 & 0.6 & 19.5 & 0.2 & 16.5 & 0.2  \\ 
			&                          & p = 800 &10.9 & 0.1 & 13.5 & 0.1 & 11.1 & 0.1 & 15.1 & 0.4 & 18.4 & 0.8 & 20.1 & 0.2 & 17.1 & 0.2 \\
   &                          & p = 1200 &14.2 & 0.5 & 13.4 & 0.1 & 11.1 & 0.1 & 23.1 & 0.8 & 49.9 & 1.1 & 20.3 & 0.2 & 17.8 & 0.2\\
	
			&                          & p = 2000 &13.6 & 0.5 & 13.5 & 0.1 & 11.2 & 0.1 & 26.9 & 0.9 & 53.8 & 1.2 & 20.7 & 0.2 & 18.6 & 0.2\\\cline{4-17}
			& \multirow{5}{*}{n =800}  & p = 200 &8.8 & 0.1 & 10.7 & 0.1 & 8.9 & 0.1 & 10.6 & 0.1 & 9.7 & 0.3 & 14.5 & 0.1 & 11.9 & 0.1\\ 
			&                          & p = 500 &8.7 & 0.1 & 10.5 & 0.1 & 8.9 & 0.1 & 11.1 & 0.3 & 9.9 & 0.3 & 14.8 & 0.1 & 12.5 & 0.1\\ 
			&                          & p = 800 &8.6 & 0.1 & 10.5 & 0.1 & 8.8 & 0.1 & 11.7 & 0.4 & 12.4 & 0.6 & 15.0 & 0.1 & 12.5 & 0.1 \\
   		&                          &p = 1200 &10.2 & 0.4 & 10.4 & 0.1 & 8.8 & 0.1 & 18.3 & 0.8 & 37.9 & 1.1 & 15.1 & 0.1 & 12.9 & 0.1\\

			&                          & p = 2000 &11.1 & 0.5 & 10.6 & 0.1 & 8.8 & 0.1 & 20.5 & 0.8 & 38.2 & 1.1 & 15.8 & 0.1 & 13.4 & 0.1\\\cline{2-17}
			\multirow{15}{*}{$\mathcal{M}_2$} & \multirow{5}{*}{n = 200} & p = 200 &14.0 & 0.2 & 20.8 & 0.2 & 14.8 & 0.2 & 25.4 & 0.5 & 21.4 & 0.7 & 31.8 & 0.3 & 23.3 & 0.2  \\ 
			&                          & p = 500 &14.3 & 0.2 & 21.1 & 0.2 & 15.0 & 0.2 & 31.4 & 0.7 & 29.0 & 1.0 & 34.1 & 0.3 & 27.4 & 0.3 \\ 
			&                          & p = 800 &14.2 & 0.2 & 20.7 & 0.2 & 14.8 & 0.2 & 33.1 & 0.7 & 33.6 & 1.1 & 34.6 & 0.3 & 30.5 & 0.3 \\
   & & p = 1200 &23.9 & 0.7 & 21.1 & 0.2 & 14.9 & 0.2 & 45.1 & 1.0 & 75.9 & 0.9 & 36.9 & 0.3 & 34.9 & 0.4\\ 
		   &  & p = 2000 &24.9 & 0.8 & 20.6 & 0.2 & 14.6 & 0.2 & 47.9 & 1.0 & 77.2 & 0.9 & 38.6 & 0.4 & 42.5 & 0.5 \\ \cline{4-17}
			& \multirow{5}{*}{n = 500} & p = 200 &8.7 & 0.1 & 13.2 & 0.1 & 9.0 & 0.1 & 13.7 & 0.3 & 10.4 & 0.4 & 18.6 & 0.2 & 12.6 & 0.1 \\ 
			&                          & p = 500 &8.9 & 0.1 & 13.3 & 0.1 & 9.3 & 0.1 & 14.3 & 0.3 & 13.4 & 0.6 & 19.4 & 0.2 & 13.4 & 0.1 \\ 
			&                          & p = 800 &8.8 & 0.1 & 13.3 & 0.1 & 9.1 & 0.1 & 14.0 & 0.3 & 14.9 & 0.7 & 19.9 & 0.2 & 13.8 & 0.1 \\ 
   &                          & p = 1200 &12.8 & 0.5 & 13.3 & 0.1 & 9.2 & 0.1 & 24.3 & 0.8 & 45.9 & 1.2 & 20.1 & 0.2 & 14.6 & 0.1\\ 
						&                          & p = 2000 &12.5 & 0.5 & 13.5 & 0.1 & 9.3 & 0.1 & 27.0 & 0.9 & 49.7 & 1.2 & 20.6 & 0.2 & 15.1 & 0.1\\ 
\cline{4-17}
			& \multirow{5}{*}{n =800}  & p = 200 &7.1 & 0.1 & 10.6 & 0.1 & 7.4 & 0.1 & 10.6 & 0.1 & 7.9 & 0.2 & 14.4 & 0.1 & 9.8 & 0.1\\ 
			&                          & p = 500 &7.1 & 0.1 & 10.6 & 0.1 & 7.3 & 0.1 & 10.9 & 0.2 & 9.3 & 0.4 & 15.0 & 0.1 & 10.2 & 0.1\\ 
			&                          & p = 800 & 7.1 & 0.1 & 10.6 & 0.1 & 7.3 & 0.1 & 11.6 & 0.3 & 10.8 & 0.6 & 15.1 & 0.1 & 10.2 & 0.1 \\ 
   &                          & p = 1200 &8.8 & 0.4 & 10.4 & 0.1 & 7.2 & 0.1 & 17.5 & 0.7 & 32.8 & 1.1 & 15.0 & 0.1 & 10.5 & 0.1 \\

			&                          & p = 2000 &9.7 & 0.5 & 10.5 & 0.1 & 7.3 & 0.1 & 19.7 & 0.8 & 34.3 & 1.2 & 15.9 & 0.1 & 10.8 & 0.1  \\\hline
		\end{tabular}
	}
	\caption{$d(V,\hat{V})$ and corresponding standard errors (in $10^{-2}$) for single index models (AR-type variance) \label{tab:single2}}
\end{table}

\begin{table}[h!]
\centering
	\resizebox{0.85\textwidth}{!}{
		
		\begin{tabular}{ccccc|cc|cc|cc|cc|cc|cc}
			\hline
			\multicolumn{3}{c}{}&  \multicolumn{2}{c}{MDDM}                           & \multicolumn{2}{c}{Oracle-SIR(3)}                         & \multicolumn{2}{c}{Oracle-SIR(10)}                         & \multicolumn{2}{c}{Rifle-SIR(3)}                    & \multicolumn{2}{c}{Rifle-SIR(10)}                                       & \multicolumn{2}{c}{LassoSIR(3)}                   & \multicolumn{2}{c}{LassoSIR(10)}  \\
			
			\multicolumn{3}{c}{}    & \multicolumn{1}{c}{Error} & \multicolumn{1}{c}{SE} & \multicolumn{1}{c}{Error} & \multicolumn{1}{c}{SE} & \multicolumn{1}{c}{Error} & \multicolumn{1}{c}{SE} &  \multicolumn{1}{c}{Error} & \multicolumn{1}{c}{SE} & \multicolumn{1}{c}{Error} & \multicolumn{1}{c}{SE} & \multicolumn{1}{c}{Error} & \multicolumn{1}{c}{SE} & \multicolumn{1}{c}{Error} & \multicolumn{1}{c}{SE} \\\cline{2-17}
			\multirow{15}{*}{$\mathcal{M}_3$} & \multirow{5}{*}{n = 200} & p = 200 &17.5 & 0.2 & 40.6 & 0.2 & 27.7 & 0.2 & 71.3 & 0.0 & 71.2 & 0.0 & 69.6 & 0.2 & 67.0 & 0.3 \\
			&                          & p = 500 & 18.1 & 0.2 & 40.7 & 0.2 & 28.2 & 0.2 & 71.3 & 0.0 & 71.2 & 0.0 & 74.8 & 0.2 & 80.1 & 0.3 \\
			&                          & p = 800 &17.7 & 0.2 & 40.8 & 0.2 & 27.7 & 0.2 & 71.3 & 0.0 & 71.2 & 0.0 & 76.5 & 0.2 & 85.0 & 0.2\\
   & & p = 1200 &17.9 & 0.2 & 40.7 & 0.2 & 28.0 & 0.2 & 31.7 & 0.4 & 18.6 & 0.2 & 78.4 & 0.2 & 88.8 & 0.2  \\ 
			  &  & p = 2000 &18.1 & 0.2 & 40.8 & 0.2 & 28.0 & 0.2 & 32.1 & 0.4 & 18.9 & 0.2 & 80.6 & 0.2 & 92.5 & 0.2  \\ 
\cline{4-17}
			& \multirow{5}{*}{n = 500} & p = 200 &10.6 & 0.1 & 27.8 & 0.2 & 17.0 & 0.1 & 70.9 & 0.0 & 70.9 & 0.0 & 48.6 & 0.3 & 30.1 & 0.2 \\
			&                          & p = 500 &10.8 & 0.1 & 27.5 & 0.2 & 17.2 & 0.1 & 70.9 & 0.0 & 70.9 & 0.0 & 53.7 & 0.3 & 39.1 & 0.3\\
			&                          & p = 800 &10.8 & 0.1 & 27.4 & 0.2 & 17.3 & 0.1 & 70.9 & 0.0 & 70.9 & 0.0 & 57.1 & 0.2 & 46.0 & 0.4 \\
    &                          & p = 1200 &10.7 & 0.1 & 27.4 & 0.2 & 17.1 & 0.1 & 18.9 & 0.2 & 11.0 & 0.1 & 59.1 & 0.2 & 53.0 & 0.4 \\

			&                          & p = 2000 &10.7 & 0.1 & 27.6 & 0.2 & 17.1 & 0.1 & 19.0 & 0.2 & 11.2 & 0.1 & 62.1 & 0.2 & 60.7 & 0.4 \\\cline{4-17}
			& \multirow{5}{*}{n =800}  & p = 200 &8.2 & 0.1 & 22.0 & 0.1 & 13.5 & 0.1 & 70.8 & 0.0 & 70.8 & 0.0 & 36.0 & 0.2 & 20.8 & 0.1 \\
			&                          & p = 500 &8.3 & 0.1 & 21.9 & 0.1 & 13.5 & 0.1 & 70.8 & 0.0 & 70.8 & 0.0 & 40.4 & 0.2 & 24.0 & 0.2 \\
			&                          & p = 800 &8.2 & 0.1 & 21.9 & 0.1 & 13.3 & 0.1 & 70.8 & 0.0 & 70.8 & 0.0 & 42.2 & 0.3 & 26.2 & 0.2 \\
   &                          & p = 1200 &8.3 & 0.1 & 22.0 & 0.1 & 13.4 & 0.1 & 14.9 & 0.1 & 8.7 & 0.1 & 44.7 & 0.2 & 29.8 & 0.3 \\
			&                          & p = 2000 &8.3 & 0.1 & 22.2 & 0.1 & 13.5 & 0.1 & 14.8 & 0.1 & 8.6 & 0.1 & 47.9 & 0.3 & 36.4 & 0.4 \\
  \cline{2-17}
			\multirow{15}{*}{$\mathcal{M}_4$} & \multirow{5}{*}{n = 200} & p = 200 &23.1 & 0.2 & 46.2 & 0.3 & 36.3 & 0.3 & 72.0 & 0.0 & 71.6 & 0.0 & 78.1 & 0.2 & 78.2 & 0.3  \\
			&                          & p = 500 &22.8 & 0.2 & 45.8 & 0.3 & 35.9 & 0.3 & 72.1 & 0.0 & 71.6 & 0.0 & 83.0 & 0.2 & 87.8 & 0.2 \\
			&                          & p = 800 &23.0 & 0.2 & 45.8 & 0.3 & 36.4 & 0.3 & 71.9 & 0.0 & 71.6 & 0.0 & 85.2 & 0.2 & 91.5 & 0.2 \\
    & & p = 1200 &23.2 & 0.3 & 45.8 & 0.3 & 36.3 & 0.3 & 38.1 & 0.4 & 25.2 & 0.3 & 87.3 & 0.2 & 93.8 & 0.2\\ 
			
    &  & p = 2000 &23.1 & 0.3 & 45.7 & 0.3 & 36.2 & 0.3 & 54.0 & 0.5 & 34.1 & 0.5 & 89.2 & 0.2 & 95.9 & 0.1\\ \cline{4-17}
			& \multirow{5}{*}{n = 500} & p = 200 &13.4 & 0.1 & 31.0 & 0.2 & 21.7 & 0.1 & 71.2 & 0.0 & 71.0 & 0.0 & 55.3 & 0.3 & 43.2 & 0.3 \\
			&                          & p = 500 &13.5 & 0.1 & 31.0 & 0.2 & 21.6 & 0.1 & 71.2 & 0.0 & 71.0 & 0.0 & 61.3 & 0.3 & 55.3 & 0.4 \\
			&                          & p = 800 &13.4 & 0.1 & 31.0 & 0.2 & 21.8 & 0.1 & 71.2 & 0.0 & 71.0 & 0.0 & 64.7 & 0.2 & 62.8 & 0.4\\
   &                          & p = 1200 &13.4 & 0.1 & 30.8 & 0.2 & 21.6 & 0.1 & 21.5 & 0.2 & 14.3 & 0.1 & 67.2 & 0.2 & 67.8 & 0.3\\
			
			&                          & p = 2000 &13.5 & 0.1 & 31.0 & 0.2 & 21.8 & 0.1 & 22.0 & 0.2 & 14.4 & 0.1 & 69.7 & 0.2 & 73.3 & 0.3 \\
\cline{4-17}
			& \multirow{5}{*}{n =800}  & p = 200 &10.5 & 0.1 & 25.2 & 0.2 & 17.2 & 0.1 & 71.0 & 0.0 & 70.9 & 0.0 & 42.7 & 0.2 & 28.7 & 0.2 \\
			&                          & p = 500 &10.4 & 0.1 & 24.9 & 0.2 & 17.1 & 0.1 & 71.0 & 0.0 & 70.9 & 0.0 & 47.3 & 0.3 & 34.2 & 0.3\\
			&                          & p = 800 &10.5 & 0.1 & 25.1 & 0.2 & 17.1 & 0.1 & 71.0 & 0.0 & 70.9 & 0.0 & 50.2 & 0.3 & 39.9 & 0.4
			\\
   &                          & p = 1200 &10.3 & 0.1 & 24.9 & 0.2 & 17.0 & 0.1 & 16.9 & 0.2 & 11.0 & 0.1 & 52.4 & 0.3 & 45.8 & 0.4\\
			&                          & p = 2000 &10.6 & 0.1 & 25.2 & 0.2 & 17.2 & 0.1 & 17.3 & 0.2 & 11.3 & 0.1 & 56.5 & 0.3 & 53.8 & 0.4  \\\cline{2-17}
			\multirow{15}{*}{$\mathcal{M}_5$} & \multirow{5}{*}{n = 200} & p = 200 &30.6 & 0.6 & 29.1 & 0.1 & 22.0 & 0.1 & 71.6 & 0.0 & 71.2 & 0.0 & 58.8 & 0.3 & 56.6 & 0.3  \\
			&                          & p = 500 &30.4 & 0.6 & 28.9 & 0.2 & 22.2 & 0.1 & 71.5 & 0.0 & 71.2 & 0.0 & 67.4 & 0.3 & 73.5 & 0.4  \\
			&                          & p = 800 &30.8 & 0.6 & 28.8 & 0.2 & 22.1 & 0.1 & 71.6 & 0.0 & 71.2 & 0.0 & 71.2 & 0.3 & 81.3 & 0.3\\
   && p= 1200 &31.0 & 0.6 & 29.1 & 0.1 & 22.3 & 0.1 & 20.0 & 0.2 & 14.6 & 0.1 & 74.1 & 0.3 & 86.4 & 0.3\\ 

    & & p = 2000 &31.3 & 0.6 & 28.7 & 0.2 & 22.1 & 0.1 & 19.3 & 0.2 & 14.4 & 0.1 & 77.8 & 0.3 & 90.3 & 0.2\\ 
\cline{4-17}
			& \multirow{5}{*}{n = 500} & p = 200 &12.4 & 0.2 & 18.4 & 0.1 & 13.6 & 0.1 & 71.1 & 0.0 & 70.9 & 0.0 & 31.7 & 0.2 & 24.5 & 0.2\\
			&                          & p = 500 &11.8 & 0.2 & 18.4 & 0.1 & 13.6 & 0.1 & 71.0 & 0.0 & 70.9 & 0.0 & 35.9 & 0.2 & 29.2 & 0.2\\
			&                          & p = 800 &11.9 & 0.2 & 18.4 & 0.1 & 13.6 & 0.1 & 71.0 & 0.0 & 70.9 & 0.0 & 37.9 & 0.3 & 32.9 & 0.3 \\
   &                          & p = 1200 &11.6 & 0.2 & 18.4 & 0.1 & 13.7 & 0.1 & 12.1 & 0.1 & 8.6 & 0.1 & 39.7 & 0.3 & 37.3 & 0.3\\ 
			
			&                          & p = 2000 &12.0 & 0.2 & 18.5 & 0.1 & 13.6 & 0.1 & 12.1 & 0.1 & 8.6 & 0.1 & 42.1 & 0.3 & 46.3 & 0.4 \\
\cline{4-17}
			& \multirow{5}{*}{n = 800} & p = 200 &7.9  & 0.1 & 14.7 & 0.1 & 10.7 & 0.1 & 70.9 & 0.0 & 70.8 & 0.0 & 23.8 & 0.1 & 17.4 & 0.1 \\
			&                          & p = 500 &7.8  & 0.1 & 14.5 & 0.1 & 10.6 & 0.1 & 70.9 & 0.0 & 70.8 & 0.0 & 25.7 & 0.2 & 19.2 & 0.1\\
			&                          & p = 800 &7.8  & 0.1 & 14.6 & 0.1 & 10.7 & 0.1 & 70.9 & 0.0 & 70.8 & 0.0 & 27.1 & 0.2 & 20.7 & 0.2  \\
   &                          & p = 1200 &7.7 & 0.1 & 14.6 & 0.1 & 10.6 & 0.1 & 9.3 & 0.1 & 6.5 & 0.1 & 28.3 & 0.2 & 21.6 & 0.2 \\
			&                          & p = 2000  &7.9 & 0.1 & 14.4 & 0.1 & 10.7 & 0.1 & 9.4 & 0.1 & 6.6 & 0.1 & 29.9 & 0.2 & 25.0 & 0.2 \\\hline
		\end{tabular}
	}
	\caption{$d(V,\hat{V})$ and corresponding standard errors (in $10^{-2}$) for multiple index models (identity variance) \label{tab:multi1}}
\end{table}

\begin{table}[h!]
\centering
	\resizebox{0.85\textwidth}{!}{
		
		\begin{tabular}{ccccc|cc|cc|cc|cc|cc|cc}
			\hline
			\multicolumn{3}{c}{}&  \multicolumn{2}{c}{MDDM}                           & \multicolumn{2}{c}{Oracle-SIR(3)}                         & \multicolumn{2}{c}{Oracle-SIR(10)}                         & \multicolumn{2}{c}{Rifle-SIR(3)}                    & \multicolumn{2}{c}{Rifle-SIR(10)}                                       & \multicolumn{2}{c}{LassoSIR(3)}                   & \multicolumn{2}{c}{LassoSIR(10)}  \\
			
			\multicolumn{3}{c}{}    & \multicolumn{1}{c}{Error} & \multicolumn{1}{c}{SE} & \multicolumn{1}{c}{Error} & \multicolumn{1}{c}{SE} & \multicolumn{1}{c}{Error} & \multicolumn{1}{c}{SE} &  \multicolumn{1}{c}{Error} & \multicolumn{1}{c}{SE} & \multicolumn{1}{c}{Error} & \multicolumn{1}{c}{SE} & \multicolumn{1}{c}{Error} & \multicolumn{1}{c}{SE} & \multicolumn{1}{c}{Error} & \multicolumn{1}{c}{SE} \\\cline{2-17}
			\multirow{15}{*}{$\mathcal{M}_3$} & \multirow{5}{*}{n = 200} & p = 200 &41.4 & 0.4 & 58.8 & 0.2 & 50.2 & 0.2 & 72.6 & 0.0 & 72.2 & 0.0 & 67.0 & 0.2 & 62.8 & 0.2\\
			&                          & p = 500 & 42.1 & 0.4 & 59.0 & 0.2 & 50.4 & 0.2 & 72.9 & 0.1 & 72.4 & 0.0 & 70.3 & 0.2 & 71.1 & 0.2\\
			&                          & p = 800 &43.0 & 0.4 & 59.1 & 0.2 & 50.4 & 0.2 & 72.7 & 0.0 & 72.3 & 0.0 & 72.1 & 0.2 & 75.1 & 0.2\\
   && p = 1200 &42.7 & 0.4 & 59.4 & 0.2 & 50.8 & 0.2 & 53.0 & 0.4 & 39.9 & 0.4 & 73.0 & 0.2 & 78.6 & 0.2  \\ 

   & & p = 2000 &43.5 & 0.4 & 59.2 & 0.2 & 50.8 & 0.2 & 53.5 & 0.4 & 40.2 & 0.4 & 74.8 & 0.2 & 82.1 & 0.2\\ \cline{4-17}
			& \multirow{5}{*}{n = 500} & p = 200 &25.6 & 0.3 & 44.6 & 0.2 & 34.3 & 0.2 & 71.4 & 0.0 & 71.3 & 0.0 & 51.5 & 0.2 & 42.4 & 0.2\\
			&                          & p = 500 &25.8 & 0.3 & 44.4 & 0.2 & 34.5 & 0.2 & 71.5 & 0.0 & 71.3 & 0.0 & 53.8 & 0.2 & 45.7 & 0.2\\
			&                          & p = 800 &25.2 & 0.3 & 44.6 & 0.2 & 34.1 & 0.2 & 71.5 & 0.0 & 71.3 & 0.0 & 54.8 & 0.2 & 47.1 & 0.3\\
   &                          & p = 1200 &25.9 & 0.3 & 44.4 & 0.2 & 34.4 & 0.2 & 33.8 & 0.4 & 23.8 & 0.2 & 55.8 & 0.2 & 50.1 & 0.3  \\

			&                          & p = 2000 &25.7 & 0.3 & 44.3 & 0.2 & 34.4 & 0.2 & 35.0 & 0.4 & 23.8 & 0.2 & 56.9 & 0.2 & 54.0 & 0.3\\ \cline{4-17}
			& \multirow{5}{*}{n =800}  & p = 200 &20.1 & 0.2 & 37.0 & 0.2 & 27.8 & 0.2 & 71.2 & 0.0 & 71.0 & 0.0 & 43.8 & 0.2 & 34.9 & 0.2\\
			&                          & p = 500 &19.8 & 0.2 & 37.8 & 0.2 & 27.6 & 0.2 & 71.2 & 0.0 & 71.1 & 0.0 & 46.1 & 0.2 & 36.6 & 0.2 \\
			&                          & p = 800 &20.0 & 0.2 & 37.3 & 0.2 & 27.9 & 0.2 & 71.2 & 0.0 & 71.0 & 0.0 & 47.0 & 0.2 & 37.9 & 0.2\\
   &                          & p = 1200 &19.8 & 0.2 & 37.1 & 0.2 & 27.8 & 0.2 & 26.7 & 0.3 & 18.5 & 0.2 & 47.6 & 0.2 & 38.6 & 0.2 \\
			&                          & p = 2000 &19.9 & 0.2 & 37.4 & 0.2 & 27.7 & 0.2 & 27.5 & 0.3 & 18.7 & 0.2 & 48.8 & 0.2 & 40.7 & 0.2  \\\cline{2-17}
			\multirow{15}{*}{$\mathcal{M}_4$} & \multirow{5}{*}{n = 200} & p = 200 &58.1 & 0.4 & 75.9 & 0.2 & 70.1 & 0.3 & 80.4 & 0.2 & 78.2 & 0.2 & 85.7 & 0.2 & 84.6 & 0.2 \\
			&                          & p = 500 &59.3 & 0.5 & 75.8 & 0.2 & 70.2 & 0.3 & 95.2 & 0.1 & 94.2 & 0.1 & 88.8 & 0.2 & 90.2 & 0.2\\
			&                          & p = 800 &59.1 & 0.5 & 75.1 & 0.2 & 69.9 & 0.3 & 81.0 & 0.2 & 78.7 & 0.2 & 89.7 & 0.2 & 92.1 & 0.2\\
   & & p = 1200 &59.5 & 0.5 & 75.3 & 0.2 & 69.9 & 0.3 & 73.9 & 0.3 & 62.3 & 0.5 & 90.5 & 0.2 & 93.9 & 0.2\\ 

   && p = 2000 &60.4 & 0.5 & 75.4 & 0.2 & 69.9 & 0.3 & 74.1 & 0.4 & 63.0 & 0.5 & 91.9 & 0.2 & 95.2 & 0.1\\ 
\cline{4-17}
			& \multirow{5}{*}{n = 500} & p = 200 &38.6 & 0.4 & 61.4 & 0.2 & 50.9 & 0.3 & 74.5 & 0.1 & 73.5 & 0.1 & 70.5 & 0.2 & 63.0 & 0.3 \\
			&                          & p = 500 &38.7 & 0.4 & 61.6 & 0.2 & 51.1 & 0.3 & 74.7 & 0.1 & 73.4 & 0.1 & 74.0 & 0.2 & 69.0 & 0.3\\
			&                          & p = 800 &39.3 & 0.4 & 61.3 & 0.2 & 50.8 & 0.2 & 77.5 & 0.2 & 74.9 & 0.1 & 75.6 & 0.2 & 72.5 & 0.3\\
   &                          & p = 1200 &39.5 & 0.4 & 61.3 & 0.2 & 51.0 & 0.2 & 54.9 & 0.4 & 40.2 & 0.4 & 77.1 & 0.2 & 76.0 & 0.3\\

			&                          & p = 2000 &39.7 & 0.4 & 61.5 & 0.2 & 51.3 & 0.3 & 56.1 & 0.4 & 40.6 & 0.4 & 78.7 & 0.2 & 79.7 & 0.3\\ 
\cline{4-17}
			& \multirow{5}{*}{n =800}  & p = 200 &30.7 & 0.3 & 53.1 & 0.2 & 42.2 & 0.2 & 73.1 & 0.1 & 72.4 & 0.0 & 61.3 & 0.2 & 51.4 & 0.3\\
			&                          & p = 500 &30.5 & 0.3 & 52.9 & 0.2 & 42.0 & 0.2 & 73.1 & 0.1 & 72.4 & 0.0 & 64.1 & 0.2 & 54.4 & 0.3\\
			&                          & p = 800 &30.6 & 0.3 & 53.7 & 0.2 & 42.3 & 0.2 & 73.4 & 0.1 & 72.5 & 0.0 & 65.9 & 0.2 & 58.2 & 0.3
			\\
   	&                          & p = 1200 &30.8 & 0.3 & 53.5 & 0.2 & 42.2 & 0.2 & 45.3 & 0.4 & 31.5 & 0.3 & 67.2 & 0.2 & 60.8 & 0.3\\
			&                          & p = 2000 &30.7 & 0.3 & 53.6 & 0.2 & 42.2 & 0.2 & 45.2 & 0.4 & 31.1 & 0.3 & 68.6 & 0.2 & 64.6 & 0.3  \\\cline{2-17}
			\multirow{15}{*}{$\mathcal{M}_5$} & \multirow{5}{*}{n = 200} & p = 200 &45.5 & 0.6 & 46.5 & 0.2 & 35.7 & 0.2 & 73.9 & 0.1 & 72.4 & 0.0 & 62.6 & 0.2 & 51.8 & 0.2 \\
			&                          & p = 500 &46.9 & 0.6 & 46.9 & 0.2 & 35.8 & 0.2 & 73.9 & 0.1 & 72.4 & 0.0 & 65.6 & 0.2 & 57.4 & 0.3 \\
			&                          & p = 800 &46.2 & 0.6 & 46.4 & 0.2 & 35.5 & 0.2 & 73.8 & 0.1 & 72.4 & 0.0 & 66.5 & 0.2 & 61.4 & 0.3 \\
   && p = 1200 &46.9 & 0.6 & 46.3 & 0.2 & 35.8 & 0.2 & 33.0 & 0.3 & 23.9 & 0.2 & 67.2 & 0.3 & 65.7 & 0.3\\ 
   && p = 2000 &46.1 & 0.6 & 46.5 & 0.2 & 35.7 & 0.2 & 33.6 & 0.3 & 23.8 & 0.2 & 68.8 & 0.3 & 72.0 & 0.3\\ 
\cline{4-17}
			& \multirow{5}{*}{n = 500} & p = 200 &24.7 & 0.4 & 31.4 & 0.2 & 22.6 & 0.1 & 72.0 & 0.0 & 71.3 & 0.0 & 44.8 & 0.2 & 32.6 & 0.1 \\
			&                          & p = 500 &24.7 & 0.4 & 31.1 & 0.2 & 22.4 & 0.1 & 71.9 & 0.0 & 71.3 & 0.0 & 46.8 & 0.2 & 35.4 & 0.2  \\
			&                          & p = 800 &25.4 & 0.4 & 31.3 & 0.2 & 22.9 & 0.1 & 72.0 & 0.0 & 71.3 & 0.0 & 48.0 & 0.2 & 36.7 & 0.2 \\
   			&                          & p = 1200 &25.2 & 0.4 & 31.5 & 0.2 & 22.7 & 0.1 & 20.8 & 0.2 & 14.3 & 0.1 & 48.2 & 0.2 & 37.6 & 0.2\\ 

			&                          & p = 2000 &25.5 & 0.4 & 31.7 & 0.2 & 22.8 & 0.1 & 20.9 & 0.2 & 14.4 & 0.1 & 49.1 & 0.2 & 39.5 & 0.2\\ 
\cline{4-17}
			& \multirow{5}{*}{n =800}  & p = 200 &18.7 & 0.3 & 25.6 & 0.1 & 18.0 & 0.1 & 71.5 & 0.0 & 71.1 & 0.0 & 36.7 & 0.1 & 25.6 & 0.1\\
			&                          & p = 500 &18.1 & 0.3 & 25.3 & 0.1 & 17.7 & 0.1 & 71.5 & 0.0 & 71.1 & 0.0 & 38.4 & 0.2 & 27.1 & 0.1 \\
			&                          & p = 800 &18.6 & 0.3 & 25.4 & 0.1 & 18.0 & 0.1 & 71.5 & 0.0 & 71.1 & 0.0 & 39.3 & 0.2 & 28.4 & 0.1 \\
   &                          & p = 1200 &18.7 & 0.3 & 25.3 & 0.1 & 17.9 & 0.1 & 16.3 & 0.1 & 11.0 & 0.1 & 40.2 & 0.2 & 29.2 & 0.1  \\

			&                          & p = 2000 &19.1 & 0.3 & 25.3 & 0.1 & 18.0 & 0.1 & 16.5 & 0.1 & 11.2 & 0.1 & 41.0 & 0.2 & 30.3 & 0.1  \\\hline
		\end{tabular}
	}
	\caption{$d(V,\hat{V})$ and corresponding standard errors (in $10^{-2}$) for multiple index models (AR-type variance)\label{tab:multi2} }
\end{table}

\begin{table}[h!]
	\resizebox{\textwidth}{!}{
		
		\begin{tabular}{ccccc|cc|cc|cc|cc|cc|cc}
			\hline
			\multicolumn{3}{c}{}&  \multicolumn{2}{c}{MDDM}                           & \multicolumn{2}{c}{Oracle-SIR(3)}                         & \multicolumn{2}{c}{Oracle-SIR(10)}                         & \multicolumn{2}{c}{Rifle-SIR(3)}                    & \multicolumn{2}{c}{Rifle-SIR(10)}                                       & \multicolumn{2}{c}{LassoSIR(3)}                   & \multicolumn{2}{c}{LassoSIR(10)}                                      \\
			
			\multicolumn{3}{c}{}    & \multicolumn{1}{c}{Error} & \multicolumn{1}{c}{SE} & \multicolumn{1}{c}{Error} & \multicolumn{1}{c}{SE} & \multicolumn{1}{c}{Error} & \multicolumn{1}{c}{SE} &  \multicolumn{1}{c}{Error} & \multicolumn{1}{c}{SE} & \multicolumn{1}{c}{Error} & \multicolumn{1}{c}{SE} & \multicolumn{1}{c}{Error} & \multicolumn{1}{c}{SE} & \multicolumn{1}{c}{Error} & \multicolumn{1}{c}{SE}  \\\cline{2-17}
			\multirow{15}{*}{$\mathcal{M}_6$} & \multirow{5}{*}{n = 200} & p = 200 &34.3&0.5&49.1&0.5&33.5&0.4&50.0&0.7&30.6&0.5&70.7&0.0&70.7&0.0 \\ 
			&                          & p = 500 &34.2&0.5&48.7&0.5&32.8&0.4&49.5&0.7&30.1&0.5&70.7&0.0&70.7&0.0\\ 
			&                          & p = 800 &34.6&0.6&48.9&0.5&33.4&0.5&50.1&0.7&30.8&0.6&70.7&0.0&70.7&0.0 \\
   && p = 1200 &44.5 & 0.7 & 49.3 & 0.5 & 33.6 & 0.5 & 67.0 & 0.7 & 38.9 & 0.7 & 70.7 & 0.0 & 70.7 & 0.0 \\ 

    & & p = 2000 &44.8 & 0.8 & 48.1 & 0.5 & 33.2 & 0.5 & 66.7 & 0.8 & 39.4 & 0.7 & 70.7 & 0.0 & 70.7 & 0.0\\ 
\cline{4-17}
			& \multirow{5}{*}{n = 500} & p = 200 &22.0&0.4&35.5&0.4&22.6&0.3&33.0&0.5&19.0&0.3&70.7&0.0&70.7&0.0\\ 
			&                          & p = 500 &21.8&0.4&34.7&0.4&22.3&0.3&32.4&0.5&18.8&0.4&70.7&0.0&70.7&0.0\\
			&                          & p = 800 &21.7&0.3&34.6&0.4&22.5&0.3&32.3&0.5&18.9&0.3&70.7&0.0&70.7&0.0\\
   &                          & p = 1200 &27.1 & 0.5 & 35.5 & 0.4 & 22.5 & 0.3 & 43.4 & 0.7 & 23.4 & 0.4 & 70.7 & 0.0 & 70.7 & 0.0 \\

			&                          & p = 2000 &26.9 & 0.5 & 34.8 & 0.4 & 22.4 & 0.3 & 42.3 & 0.7 & 23.6 & 0.5 & 70.7 & 0.0 & 70.7 & 0.0 \\\cline{4-17}
			& \multirow{5}{*}{n =800}  & p = 200 &17.2&0.3&29.2&0.3&18.6&0.3&25.8&0.4&14.9&0.2&70.7&0.0&70.7&0.0\\ 
			&                          & p = 500 &16.7&0.3&28.5&0.3&18.1&0.2&25.3&0.4&14.3&0.3&70.7&0.0&70.7&0.0\\
			&                          & p = 800 &16.5&0.2&28.4&0.3&17.9&0.2&25.0&0.4&14.1&0.2&70.7&0.0&70.7&0.0\\
   &                          & p =1200 &20.8 & 0.4 & 29.0 & 0.4 & 18.2 & 0.3 & 31.7 & 0.5 & 18.5 & 0.4 & 70.7 & 0.0 & 70.7 & 0.0 \\
			&                          & p = 2000 &21.2 & 0.4 & 28.7 & 0.4 & 18.3 & 0.3 & 32.4 & 0.6 & 18.8 & 0.3 & 70.7 & 0.0 & 70.7 & 0.0 \\\hline
		\end{tabular}
	}
	\caption{$d(V,\hat{V})$ and corresponding standard errors (in $10^{-2}$) for isotropic PFC models ($\mathcal{M}_6$) \label{tab:pfc}}
\end{table}

\begin{table}[h!]
	\centering
	\resizebox{\textwidth}{!}{
		\begin{tabular}{cccccc|cccc|cccc}
			\hline
			&        & \multicolumn{4}{c}{$n = 100$}        & \multicolumn{4}{c}{$n = 200$}  & \multicolumn{4}{c}{$n = 400$} \\  \cline{3-14}
			&        &\multicolumn{2}{c}{$p = 800$}&\multicolumn{2}{c}{$p = 1200$}&\multicolumn{2}{c}{$p = 800$}&\multicolumn{2}{c}{$p = 1200$}&\multicolumn{2}{c}{$p = 800$}&\multicolumn{2}{c}{$p = 1200$}\\
			&&Error&SE&Error&SE&Error&SE&Error&SE&Error&SE&Error&SE\\\cline{2-14}
			\multirow{6}{*}{$\mathcal{M}_7$}& MDDM   & 45.0 & 0.5 & 45.9 & 0.5 & 27.5 & 0.3 & 28.7 & 0.4 & 18.8 & 0.3 & 19.6 & 0.3\\  
			& PR-Oracle-SIR(3) &26.3 & 0.2 & 26.5 & 0.2 & 18.3 & 0.1 & 18.4 & 0.2 & 12.6 & 0.1 & 12.6 & 0.1\\
			& PR-Oracle-SIR(10) & 33.0 & 0.3 & 33.0 & 0.3 & 20.1 & 0.2 & 20.2 & 0.2 & 13.0 & 0.1 & 13.1 & 0.1\\
			& PR-SIR(3) & 96.6 & 0.0 & 97.7 & 0.0 & 93.2 & 0.0 & 95.3 & 0.0 & 87.6 & 0.0 & 91.1 & 0.0  \\  
			& PR-SIR(10) & 97.4 & 0.0 & 98.2 & 0.0 & 95.0 & 0.0 & 96.6 & 0.0 & 90.0 & 0.1 & 93.7 & 0.0  \\  \cline{1-14} 		
   			\multirow{6}{*}{$\mathcal{M}_8$}& MDDM   & 93.5 & 0.6 & 94.9 & 0.5 & 73.6 & 1.1 & 77.1 & 1.1 & 36.7 & 1.1 & 37.8 & 1.2 \\  
			& PR-Oracle-SIR(3) &80.1 & 0.6 & 80.1 & 0.6 & 64.0 & 0.7 & 64.7 & 0.7 & 42.5 & 0.5 & 41.1 & 0.5\\
			& PR-Oracle-SIR(10) & 79.1 & 0.6 & 78.9 & 0.6 & 58.4 & 0.6 & 58.7 & 0.6 & 34.5 & 0.4 & 34.2 & 0.4\\
			& PR-SIR(3) & 99.9 & 0.0 & 100.0 & 0.0 & 99.9 & 0.0 & 100.0 & 0.0 & 99.9 & 0.0 & 99.9 & 0.0\\  
			& PR-SIR(10) & 99.9 & 0.0 & 100.0 & 0.0 & 99.9 & 0.0 & 100.0 & 0.0 & 99.9 & 0.0 & 100.0 & 0.0 \\  \cline{1-14}
      			\multirow{6}{*}{$\mathcal{M}_9$}& MDDM   & 17.3 & 0.4 & 17.5 & 0.4 & 10.0 & 0.1 & 10.1 & 0.1 & 7.1 & 0.1 & 7.1 & 0.1\\  
			& PR-Oracle-SIR(3) & 15.5 & 0.2 & 15.2 & 0.2 & 10.6 & 0.1 & 10.6 & 0.1 & 7.5 & 0.1 & 7.5 & 0.1 \\
			& PR-Oracle-SIR(10) & 14.2 & 0.2 & 13.9 & 0.2 & 9.6 & 0.1 & 9.6 & 0.1 & 6.8 & 0.1 & 6.8 & 0.1\\
			& PR-SIR(3) & 92.2 & 0.1 & 95.0 & 0.1 & 83.0 & 0.1 & 88.3 & 0.1 & 71.0 & 0.1 & 77.9 & 0.1\\  
			& PR-SIR(10) & 90.0 & 0.1 & 93.4 & 0.1 & 80.1 & 0.1 & 85.8 & 0.1 & 67.4 & 0.1 & 74.7 & 0.1\\  \cline{1-14} 
	\end{tabular}}
	\caption{Averaged subspace estimation errors and the corresponding standard errors (after multiplied by $100$) for multivariate response models.} \label{tab:MRM}
\end{table}

\clearpage
\newpage

\section{More on Real Data Analysis}
\label{sec:realdata}
\subsection{Choice of Tuning Parameter on Real Data}\label{sec:tuning}

To apply our proposal on real data, we need to determine the tuning parameter, $s$, i.e, the desired level of sparsity. In penalized problems such as sparse PCA and sparse SIR, tuning parameters are often chosen with cross-validation. We could also employ cross-validation to choose s. However, as with almost any procedure, cross-validation would considerably slow down the computation. Moreover, as we observe in our theoretical and simulation studies, our method is not very sensitive to $s$; the result is reasonably stable as long as $s$ is larger than d. Hence, we resort to a faster tuning method on our real data as follows. We start with a sequence of reasonable sparsity levels $\mathcal{S}$, which is set to be $\{1,\ldots, 45\}$. Then for each element in $\mathcal{S}$, we calculate $\hat{\boldsymbol{\beta}}$ and the sample distance covariance \citep{szekely2007measuring} between $\mathbf{Y}_i$ and $\hat{\boldsymbol{\beta}}^T\mathbf{X}_i$ for all $i = 1,2,...,n$. Here distance covariance is used as a model-free measure of the dependence between $\mathbf{Y}_i$ and $\hat{\boldsymbol{\beta}^T}\mathbf{X}_i$. Intuitively, the distance covariance increases as the pre-specified sparsity increases. Therefore, we plot the sample distance covariance against the sparsity levels in Figure~\ref{fig:sparsity}, and pick the sparsity corresponding to a large enough distance covariance, while any larger sparsity levels will not lead to a significantly larger distance covariance (i.e. the ``elbow method''). Based on Figure~\ref{fig:sparsity}, a sparsity level between 20 and 25 seems to be reasonable, and we pick $s = 25$ as the pre-specified sparsity.

\begin{figure}[h!]
    \centering
    \includegraphics[scale = 0.6]{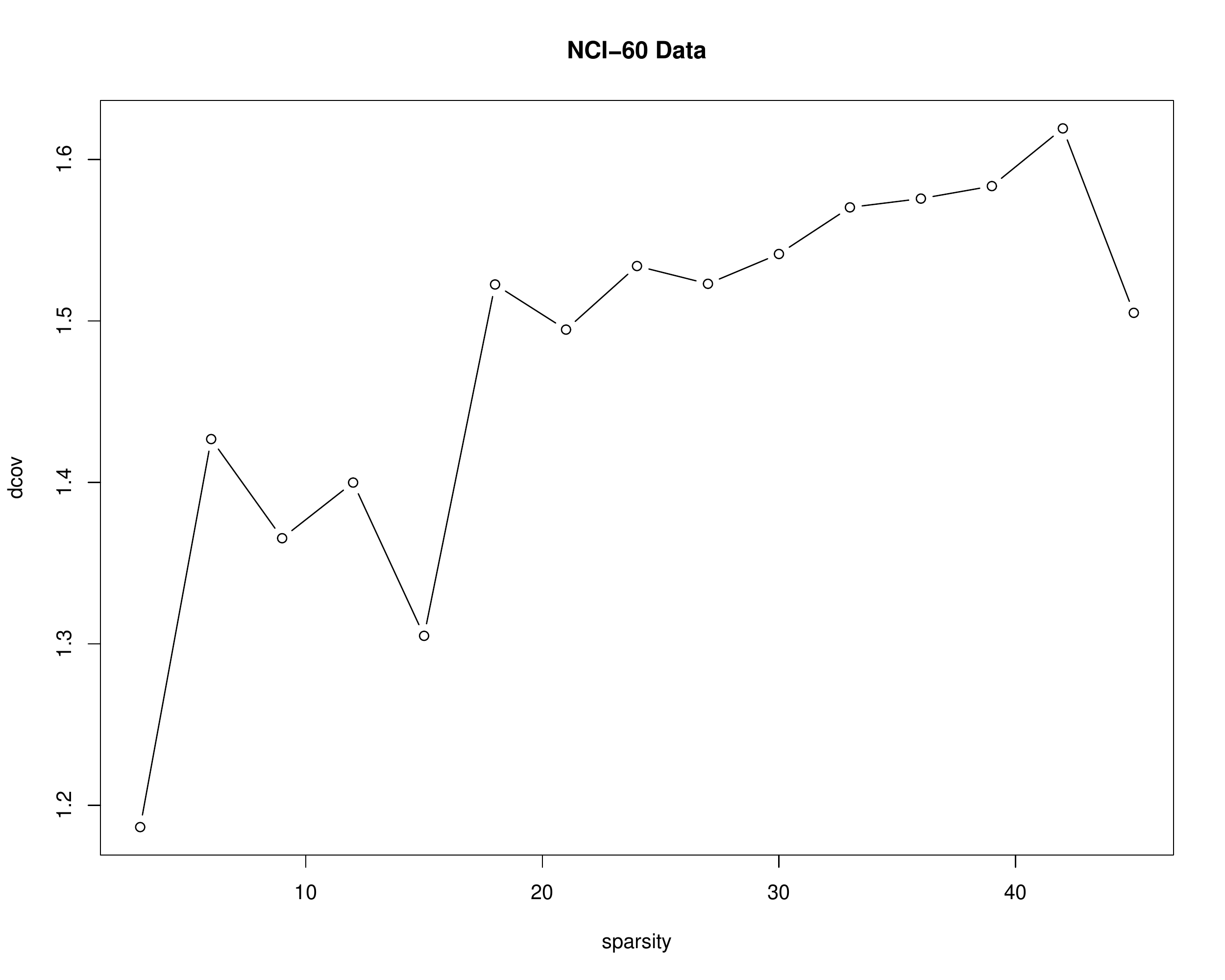}
    \caption{Distance covariance for different pre-specified sparsities.}
    \label{fig:sparsity}
\end{figure}

\subsection{Additional Real Data Analysis Results}

	To further demonstrate our methods on the real data, we also construct the scatterplot from the leading two directions $\bolbeta_1^{\T}\mbX$ and $\bolbeta_2^{\T}\mbX$ in Figure \ref{fig:b1b2}. For the first direction, the experimental units from leukemia cancer (LE) and colorectal cancer (CO) are on the different side of the plot. For the second direction, CO and LE have similar values whereas the units obtained from central nervous system cancer (CN), breast cancer (BR), lung cancer (LC) and ovarian cancer (OV) are on the opposite side of the figure. This pattern coincides with the first two canonical correlation analysis directions in a previous study of the same data set \citep[Figure 8,][]{cruz2014fast}. Such a finding is very encouraging as our slicing-free approach automatically detect the most significant associations between the response and the predictor, while directly applied to the multivariate response.

 	\begin{figure}[ht!]
		\centering
		\includegraphics[scale = 0.5]{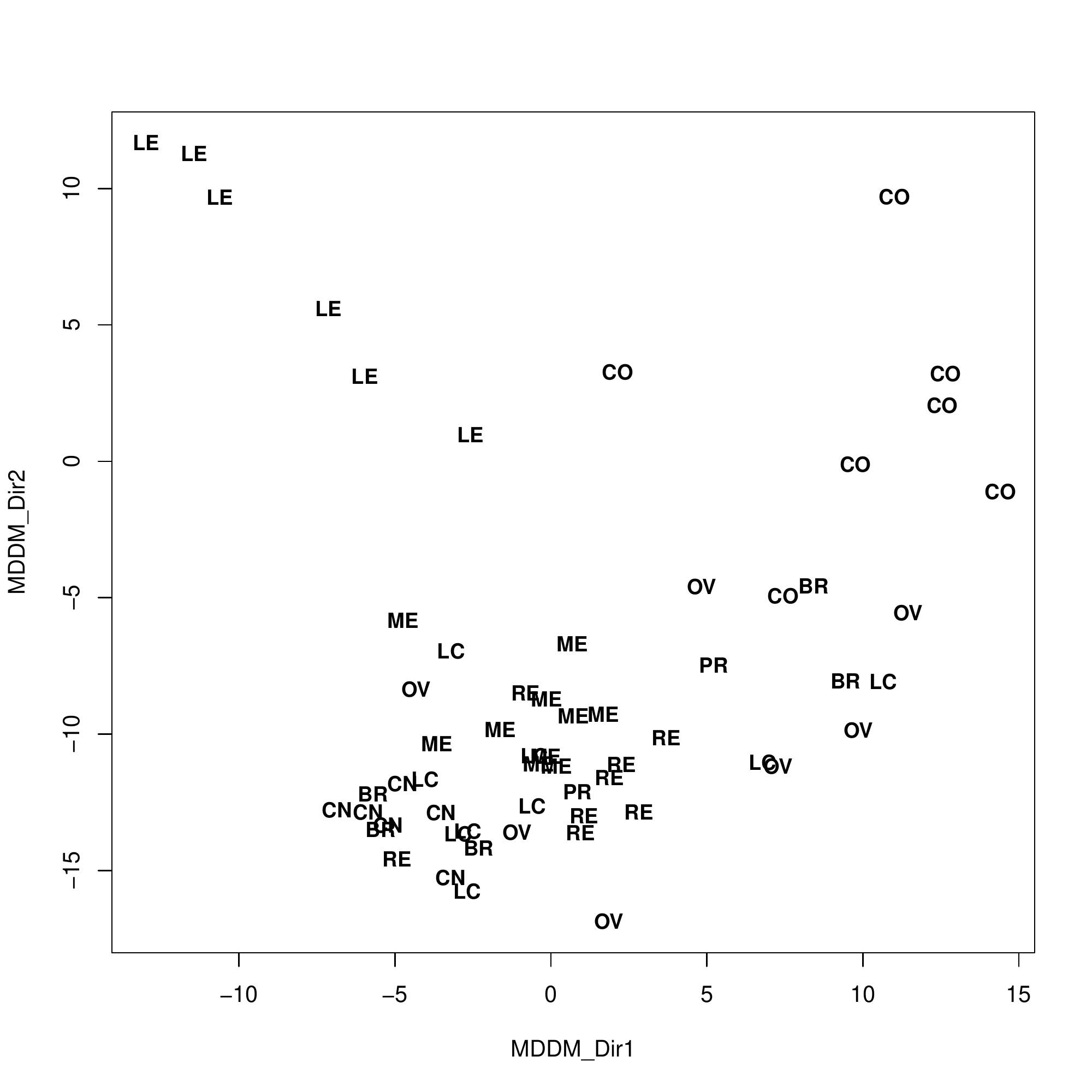}
		\caption{The scatter plot between the first leading directions.}\label{fig:b1b2}
	\end{figure}

\section{Computation complexity}
\label{sec:computation}

We briefly analyze the computational complexity of our proposed methods. For both algorithms, we need to compute the sample MDDM. The current computational complexity of sample MDDM  is $O(n^2p)$. If we adapt the fast computing algorithm of \citet{huo2016fast} developed for distance correlation to MDDM, we might be able to reduce the complexity to $O(pn\log n)$. In Algorithm~\ref{alg2}, we further need to compute the sample covariance at the complexity level of $O(np^2)$.

To apply the two algorithms, we assume that the maximum number of iterations is $T$ in finding the $K$ directions (Step 3(a) in both algorithms). After obtaining MDDM, Algorithm~\ref{alg1} has a computational complexity of $O(KT(ps+p)+(K-1)(s^2+ps))$, where $(ps+p)$ is the computation complexity of each iteration \citep{yuan2013truncated} and $(s^2+ps)$ is the computation complexity of deflating $\widehat{\mathbf{M}}_k$ in Step 3(b). Note that, in Step 3(b), we repeatedly exploit the sparsity of $\widehat\bolbeta_k$ to reduce the computation complexity. For example, to compute $\widehat\bolbeta_k\widehat\bolbeta_k^\T$, we only need to compute the $s^2$ nonzero elements, and the same applies to $\widehat\bolbeta_k^\T\widehat {\mathbf{M}}\widehat\bolbeta_k$, which has a computational complexity of $O(ps)$. For Algorithm~\ref{alg2}, the computational complexity is $O(KT(ps+p)+(K-1)(p^2+ps))$. Note that this computational complexity only differs from that of Algorithm~\ref{alg1} in the term $p^2$. This term is the cost to deflate $\widehat{\mathbf{M}}$ in Step 3(b). Since $\widehat\bSigma_{\bX}\widehat\bolbeta_k\widehat\bolbeta_k^\T\widehat\bSigma_{\bX}$ is not guaranteed to be sparse as $\widehat\bolbeta_k\widehat\bolbeta_k^\T$, the deflation is more computationally expensive. Otherwise, each iteration in Step 3(a) of Algorithm~\ref{alg2} has the same complexity as Step 3(a) of Algorithm~\ref{alg1} \citep{TWLZ18}.

\section{Proofs for Proposition~\ref{prop1} \& Lemma~\ref{sub}}
\label{sec:proplemma}

\begin{proof}[Proof of Proposition~\ref{prop1}] 

From the basic properties of MDDM (c.f. beginning of Section 1), $\E(\mbA^\T\mbX\mid Y)=\E(\mbA^\T\mbX)$
is equivalent to $\MDDM(\mbA^\T\mbX\mid Y)=0$. 

Suppose the rank of $\MDDM(\mbX\mid Y)$ is $d$, then there exists
an orthogonal basis matrix for $\mbbR^{p}$, $(\bolbeta,\bolbeta_{0})$
with $\bolbeta\in\mbbR^{p\times d}$ and $\bolbeta_{0}\in\mbbR^{p\times(p-d)}$,
such that $\spn(\bolbeta)=\spn\{\MDDM(\mbX\mid Y)\}$. This implies
$\bolbeta_{0}^\T\MDDM(\mbX\mid Y)=0$ and equivalently $\bolbeta_{0}^\T\{\E(\mbX\mid Y)-\E(\mbX)\}=0$.
Therefore, $\spn(\bolbeta_{0})\subseteq\calS_{\E(\mbX\mid Y)}^{\perp}$,
which leads to $\calS_{\E(\mbX\mid Y)}\equiv\spn\{\E(\mbX\mid Y)-\E(\mbX)\}\subseteq\spn(\bolbeta)$. 

Similarly, for any vector $\mbv\in\calS_{\E(\mbX\mid Y)}^{\perp}$
we have $\mbv^\T\{\E(\mbX\mid Y)-\E(\mbX)\}=0$ and hence $\mbv^\T\MDDM(\mbX\mid Y)\mbv=0$.
This implies that $\mbv\in\spn(\bolbeta_{0})$ and hence, $\calS_{\E(\mbX\mid Y)}^{\perp}\subseteq\spn(\bolbeta_{0})$
and $\spn(\bolbeta)\subseteq\calS_{\E(\mbX\mid Y)}$.

\end{proof}

For the proof of Lemma~\ref{sub}, we need the following elementary lemma. We include its proof for completeness.
\begin{lemma}\label{ed}
Let $\balpha_k$ be the normalized $k$th eigenvector of $\bSigma_{\mbX}^{-1/2}\bM\bSigma_{\mbX}^{-1/2}$. Then we must have $\bbeta_k=\bSigma_{\mbX}^{-1/2}\balpha_k$.
\end{lemma}
\begin{proof}[Proof of Lemma~\ref{ed}]
Let $\balpha=\bSigma_{\mbX}^{1/2}\bbeta$ in \eqref{GEP1}, we have that $\bbeta_k=\bSigma_{\mbX}^{-1/2}\balpha_k$, where
\beq
\balpha_k=\arg\max_{\balpha}\balpha^\T\bSigma_{\mbX}^{-1/2}\bM\bSigma_{\mbX}^{-1/2}\balpha \mbox{ s.t $\balpha^\T\balpha=1,\balpha^\T\balpha_l=0,l<k$}
\eeq

It is easy to see that $\balpha_k$ is the $k$th eigenvector of $\bSigma_{\mbX}^{-1/2}\bM\bSigma_{\mbX}^{-1/2}$ and the conclusion follows.
\end{proof}

\begin{proof}[Proof of Lemma~\ref{sub}]
By Lemma~\ref{ed}, we have $\bSigma_{\mbX}^{-1/2}\bM\bSigma_{\mbX}^{-1/2}=\sum_{j=1}^p\lambda_j\balpha_j\balpha_j^\T$. It follows that $\bM=\bSigma_{\mbX}\{\sum_{j=1}^p\lambda_j\bbeta_j\bbeta_j^\T\}\bSigma_{\mbX}$. Hence, $\mbM_k=\bSigma_{\mbX}\{\sum_{j=k}^p\lambda_j\bbeta_j\bbeta_j^\T\}\bSigma_{\mbX}$ and $\bolbeta_k$ is its leading generalized eigenvector subject to $\bolbeta^\T\bSigma_{\mbX}\bolbeta=1$.
\end{proof}

\section{Proof for Theorem~\ref{th:main}}
\label{sec:theorem1}

Let $(\mbV_{k\cdot},\mbU_{k\cdot})_{k=1}^{n}$ be iid sample from the joint
distribution of $(\mbV,\mbU)$, where $\mbV_{k\cdot}=(V_{k1},\cdots,V_{kp})^\T$ is the $k$th sample.
Write
$\MDDM(\mbV|\mbU)=-\E\{(\mbV-\E(\mbV))(\mbV'-\E(\mbV'))^\T|\mbU-\mbU'|_q\}=-\mbR-\mbS+2\mbT,$ where 
\begin{eqnarray*}
	\mbR&=&\E(\mbV\mbV'^\T|\mbU-\mbU'|_q)\\
	\mbS&=&\E(\mbV)\E(\mbV')^\T\E(|\mbU-\mbU'|_q)\\
	\mbT&=&\E[\mbV\mbV'^\T|\mbU'-\mbU^{''}|_q]=\E[\E(\mbV)\mbV'^\T|\mbU-\mbU'|_q]
\end{eqnarray*}
with $(\mbV',\mbU')$ and $(\mbV^{''},\mbU^{''})$ being iid copies of $(\mbV,\mbU)$. Note that $\mbV'=(V_1',\cdots,V_p')^T$.
At the sample level, we have
\[\MDDM_n(\mbV|\mbU)=-\frac{1}{n^2}\sum_{k,l=1}^{n}(\mbV_{k\cdot}-\bar{\mbV}_n)(\mbV_{l\cdot}-\bar{\mbV}_n)^T|\mbU_{k\cdot}-\mbU_{l\cdot}|_q=-\mbR_{n}-\mbS_{n}+2\mbT_{n},\]
where $\bar{\mbV}_n=n^{-1}\sum_{k=1}^{n}\mbV_{k\cdot}$, and 
\begin{eqnarray*}
	\mbR_{n}&=&n^{-2}\sum_{k,l=1}^{n}\mbV_{k\cdot}\mbV_{l\cdot}^T|\mbU_{k\cdot}-\mbU_{l\cdot}|_q\\
	\mbS_{n}&=&n^{-2}\sum_{k,l=1}^{n}\mbV_{k\cdot}\mbV_{l\cdot}^T n^{-2}\sum_{k,l=1}^{n}|\mbU_{k\cdot}-\mbU_{l\cdot}|_q\\
	\mbT_{n}&=&n^{-3}\sum_{k,l,h=1}^{n}\mbV_{k\cdot}\mbV_{h\cdot}^T|\mbU_{h\cdot}-\mbU_{l\cdot}|_q.
\end{eqnarray*}

\begin{proposition}
	\label{prop:R}
	Suppose Condition (C1) holds. There exists a positive integer
	$n_0=n_0(\sigma_0,C_0,q)<\infty$, $\gamma=\gamma(\sigma_0,C_0,q)\in (0,1/2)$ and a finite positive constant $D_0=D_0(\sigma_0,C_0,q)<\infty$ such that when $n\ge n_0$ and
	$16>\epsilon>D_0 n^{-\gamma}$, we have
	\[P(\|\mbR_n-\mbR\|_{max}>4\epsilon)\le 10 p^2\exp\left(-  \frac{\epsilon^2 n}{4\log^3{n}}\right).\]
	
\end{proposition}



\noindent Proof of Proposition~\ref{prop:R}:
Throughout  the 
proof, $C$ is a generic positive constant that vary from line to line. 
We shall find a bound for $P(\|\mbR_{n}-\mbR\|_{max}>4\epsilon)$ first.
For $i,j=1,\cdots,p$, let $R_{ij}=\E[V_iV_j'|\mbU-\mbU'|_q]$
and $R_{n,ij}=n^{-2}\sum_{k,l=1}^{n}V_{ki}V_{lj}|\mbU_{k\cdot}-\mbU_{l\cdot}|_q$. Note that
\[P(\|\mbR_{n}-\mbR\|_{max}>4\epsilon)\le p^2\max_{i,j=1,\cdots,p} P(|R_{n,ij}-R_{ij}|>4\epsilon).\]
We shall focus on the case I, $(i,j)=(1,2)$, since other cases can be treated in the
same fashion and the bound is uniformly over
all pair of $(i,j)$s.

Case I, $(i,j)=(1,2)$, write $\widetilde{R}_{n,12}=\{n(n-1)\}^{-1}\sum_{k\not=l}^{n}V_{k1}V_{l2}|\mbU_{k\cdot}-\mbU_{l\cdot}|_q$.
Let $\mbW=(\mbU^\T,V_1,V_2)^\T$ and $\mbW'=((\mbU')^\T,V_1',V_2')^\T$ and $\mbW_k=(\mbU_{k\cdot}^\T,V_{k1},V_{k2})^\T$. Define the kernel $h_1$ as
$$h_1(\mbW;\mbW')=\frac{V_1V_2'|\mbU-\mbU'|_q+V_1'V_2|\mbU-\mbU'|_q}{2}$$
Then $h_1$ is symmetric, $\widetilde{R}_{n,12}=\{n(n-1)\}^{-1}\sum_{k\not=l}^{n}h_1(\mbW_k;\mbW_l)$
is a U-statistic of order two and $R_{n,12}=\frac{n-1}{n}\widetilde{R}_{n,12}$.

Under Condition (C1), there exists a positive constant $C_1=C_1(\sigma_0,C_0)<\infty$ such that
$|R_{12}|=|\E(V_1V_2'|\mbU-\mbU'|_q)|\le \E^{1/2}(V_1^2)\E^{1/2}(V_2^{'2})\E^{1/2}(|\mbU-\mbU'|_q^2)<C_1$.
When $\epsilon$ satisfies  $\epsilon\ge C_1/(2n)$, then $|R_{12}|/n\le 2\epsilon$ and
\begin{eqnarray*}
	P(|R_{n,12}-R_{12}|\ge 4\epsilon)&&=P\left(\left|\frac{n-1}{n}(\widetilde{R}_{n,12}-R_{12})-\frac{1}{n}R_{12}\right|\ge 4\epsilon\right)\\
	&&\le P(|\widetilde{R}_{n,12}-R_{12}|+|R_{12}/n|\ge 4\epsilon) \le P(|\widetilde{R}_{n,12}-R_{12}|\ge 2\epsilon).
\end{eqnarray*}
Next we decompose
\begin{eqnarray*}
	\widetilde{R}_{n,12}&=&\{n(n-1)\}^{-1}\sum_{k\not=l}^{n}h_1(\mbW_k,\mbW_l){\bf 1}(|h_1(\mbW_k,\mbW_l)|\le M)\\
	&&+\{n(n-1)\}^{-1}\sum_{k\not=l}^{n}h_1(\mbW_k,\mbW_l){\bf 1}(|h_1(\mbW_k,\mbW_l)| > M)\\
	&=&\widetilde{R}_{n,12,1}+\widetilde{R}_{n,12,2},
\end{eqnarray*}
where the choice of $M$ will be addressed at the end of proof.
We also decompose its population counterpart $R_{12}=\E[h_1{\bf 1}(|h_1|\le M)]+\E[h_1{\bf 1}(|h_1| > M)]=R_{12,1}+R_{12,2}$.

By Lemma C on page 200 of \cite{S80}, we derive that for  $m=\lfloor n/2\rfloor$, and $t>0$,
\[\E[\exp(t\widetilde{R}_{n,12,1})]\le \E^m[\exp(t h_1{\bf 1}(|h_1|\le M)/m)]\]
which entails that
\begin{eqnarray*}
	P(\widetilde{R}_{n,12,1}-R_{12,1}\ge \epsilon)&\le& \exp(-t(\epsilon+R_{12,1}))\E[\exp(t\widetilde{R}_{n,12,1})]\\
	&\le&\exp(-t\epsilon)\E^m\{\exp(t(h_1{\bf 1}(|h_1|\le M)-R_{12,1})/m)\}\\
	&\le&\exp(-t\epsilon)\exp(t^2M^2/(2m)),
\end{eqnarray*}
where we have applied Markov's inequality and Lemma A(ii) [cf. Page 200 of \cite{S80}] in the first and third inequality above, respectively.
Applying the same argument with $h_1{\bf 1}(|h_1|\le M)$ replaced by $-h_1{\bf 1}(|-h_1|\le M)$, we can obtain
\[P(\widetilde{R}_{n,12,1}-R_{12,1}\le -\epsilon)\le \exp(-t\epsilon)\exp(t^2M^2/(2m)).\]
Choosing $t=\epsilon m/M^2$, we obtain that
\begin{eqnarray}
	\label{eq:rn121}
	P(|\widetilde{R}_{n,12,1}-R_{12,1}|\ge \epsilon)\le 2\exp(-\epsilon^2m/(2M^2))
\end{eqnarray}

Next we turn to $\widetilde{R}_{n,12,2}$. First of all, by Cauchy-Schwartz inequality,
$|R_{12,2}|\le \E^{1/2}(h_1^2)P^{1/2}(|h_1|>M)$.
Applying the inequality $|ab|\le (a^2+b^2)/2$ and $(a+b)^2\le 2(a^2+b^2)$ for any $a,b\in R$,
we derive
\begin{eqnarray*}
	h_1(\mbW_k,\mbW_l)&\le& \frac{V_{k1}V_{l2}|\mbU_{k\cdot}-\mbU_{l\cdot}|_q+V_{k2}V_{l1}|\mbU_{k\cdot}-\mbU_{l\cdot}|_q}{2}\\
	&\le&\frac{1}{4}\{(V_{k1}V_{l2}+V_{k2}V_{l1})^2+|\mbU_{k\cdot}-\mbU_{l\cdot}|_q^2\}\\
	&\le&\frac{1}{2}\{V_{k1}^2V_{l2}^2+V_{k2}^2V_{l1}^2+|\mbU_{k\cdot}|_q^2+|\mbU_{l\cdot}|_q^2\}\\
	&\le&\frac{1}{2}\{V_{k1}^4/2+V_{l2}^4/2+V_{k2}^4/2+V_{l1}^4/2+|\mbU_{k\cdot}|_q^2+|\mbU_{l\cdot}|_q^2\}
\end{eqnarray*}
Then it is easy to show that under the uniform sub-Gaussian moment
assumption in Condition (C1) and the upper bound on $h_1(\mbW_k,\mbW_l)$ above, we have that  $\E^{1/2}(h_1^2)\le C_2$ for some
$C_2=C_2(\sigma_0,C_0)<\infty$. Moreover, since $q^{1/2}\|\mbU\|_{\max}\ge \|\mbU\|_q$, we can derive that 
\begin{eqnarray}
	\label{eq:tailprobh1}
&&P(|h_1|>M)\nonumber\\
&\le& P[\max\{|V_1|,|V_1'|,|V_2|,|V_2'|,2 q^{1/2}\| \mbU\|_{\max}, 2 q^{1/2}\| \mbU'\|_{\max}\}\ge (\frac{M}{2})^{1/3}]\nonumber\\
&\le& 2 P\{|V_1|\ge (\frac{M}{2})^{1/3}\}+ 2 P\{|V_2|\ge (\frac{M}{2})^{1/3}\}+ 2 P\{ 2 q^{1/2} \| \mbU\|_{\max} \ge (\frac{M}{2})^{1/3}\}\nonumber\\
&\le&2P\{|V_1|\ge (\frac{M}{2})^{1/3}\}+ 2 P\{|V_2|\ge (\frac{M}{2})^{1/3}\} +  2\sum_{j=1}^q P\{ 2 q^{1/2}|U_j|\ge (\frac{M}{2})^{1/3}\}\label{eq.tail.bound}
\end{eqnarray}

Because $V_1$ is sub-Gaussian as assumed in Condition (C1), by Proposition~2.5.2 in \citet{vershynin2018high}, we have $P\{|V_1|\ge (\frac{M}{2})^{1/3}\}\le 2\exp\{-C(\frac{M}{2})^{2/3}\}$ for some positive constant $C$. We apply similar arguments to all the remaining terms in \eqref{eq.tail.bound} and derive that 
$$
P(|h_1|>M)\le (8+4q)\exp\{-2Cq^{-1}(\frac{M}{2})^{2/3}\}.
$$

Thus $|R_{12,2}|\le (8+4q)^{1/2}C_2 \exp\{-Cq^{-1}(\frac{M}{2})^{2/3}\}$.
If we choose $\epsilon>0$ such that $(8+4q)^{1/2}C_2 \exp\{-Cq^{-1}(\frac{M}{2})^{2/3}\}\le\epsilon/2$, then $|R_{12,2}|\le \epsilon/2$, which leads to
$P(|\widetilde{R}_{n,12,2}-R_{12,2}|\ge \epsilon)\le P(|\widetilde{R}_{n,12,2}|\ge\epsilon/2)$. To bound  $P(|\widetilde{R}_{n,12,2}|\ge\epsilon/2)$, we write 
\begin{eqnarray*}
|\widetilde{R}_{n,12,2}|&=&|\{n(n-1)\}^{-1}\sum_{k\not=l}^{n}h_1(\mbW_k,\mbW_l){\bf 1}(|h_1(\mbW_k,\mbW_l)| > M)|\\
&\le&\{n(n-1)\}^{-1}\sum_{k\not=l}^{n}|h_1(\mbW_k,\mbW_l)|{\bf 1}(\|\mbW_k\|_{\max} > q^{-1/6}(\frac{M}{2})^{1/3})\\
&&+\{n(n-1)\}^{-1}\sum_{k\not=l}^{n}|h_1(\mbW_k,\mbW_l)|{\bf 1}(\|\mbW_l\|_{\max} > q^{-1/6}(\frac{M}{2})^{1/3})\\
&\equiv& L_1+L_2.
\end{eqnarray*}
Without loss of generality, we only consider $L_1$. Define $F_k={\bf 1}(\|\mbW_k\|_{\max} > q^{-1/6}(\frac{M}{2})^{1/3})$. Note that
\begin{eqnarray*}
L_1&=&\{n(n-1)\}^{-1}\sum_{k\ne l}|V_{k1}V_{l2}|\|\mbU_k-\mbU_l\| F_k\\
&\le&\{n(n-1)\}^{-1}\sum_{k\ne l}|V_{k1}V_{l2}|\|\mbU_k\|F_k+\{n(n-1)\}^{-1}\sum_{k\ne l}|V_{k1}V_{l2}|\|\mbU_l\|F_k\\
&\le&\{n(n-1)\}^{-1} (\sum_{k=1}^n |V_{k1}|\|\mbU_k\|F_k)\cdot(\sum_{l=1}^n|V_{l2}|)+\{n(n-1)\}^{-1} (\sum_{k=1}^n|V_{k1}|F_k)\cdot \sum_{l=1}^n|V_{l2}|\|\mbU_l\|\\
&\equiv&L_{11}+L_{12}
\end{eqnarray*}

For $L_{11}$, note that, for any $\lambda>0$, $\E\exp\{\lambda|V_{l2}|^2\}=\E\exp\{\lambda V_{l2}^2\}$. Since $V_{l2}$ is sub-Gaussian by Condition (C1), we have that $|V_{l2}|$ is also sub-Gaussian by Proposition~2.5.2 in \citet{vershynin2018high}. Hence, it follows from Bernstein's inequality [Theorem 2.8.1 in \citet{vershynin2018high}] that for $\epsilon\in (0,1)$, $n\ge 2$, 
\begin{eqnarray}
	\label{eq:conc0}
	P(|\frac{1}{n-1}\sum_{l=1}^n\{|V_{l2}|-\E|V_{l2}|\}|\ge \epsilon)\le 2\exp(-Cn\epsilon^2)
	\end{eqnarray}
		Regarding $\frac{1}{n}\sum_{k=1}^n |V_{k1}|\|\mbU_k\|F_k$, we note that $|V_{k1}|\|\mbU_k\|F_k\le |V_{k1}|\|\mbU_k\|\cdot F_k\le |V_{k1}|\cdot \sum_{j=1}^q |U_{kj}|\cdot F_k$. Since $|V_{k1}|$ and $|U_{kj}|$ are sub-Gaussian, we have that $|V_{k1}|\|\mbU_k\|F_k$ is sub-exponential (Lemma 2.7.7 in \citet{vershynin2018high}). Again by Bernstein's inequality, we have that for $\epsilon\in (0,1)$, 
\begin{eqnarray}
	\label{eq:conc1}
P(|\frac{1}{n}\sum_{k=1}^n \{|V_{k1}|\|\mbU_k\|F_k-\E(|V_{k1}|\|\mbU_k\|F_k)\}|\ge\epsilon)\le 2\exp(-Cn\epsilon^2).
\end{eqnarray}

Moreover, it is easy to see that there exists a $C_3=C_3(\sigma_0,C_0)$, such that $\E^{1/2} (|V_{k1}|\|\mbU_k\|)^2\le C_3$, so $\E(|V_{k1}|\|\mbU_k\|F_k)\le \E^{1/2} (|V_{k1}|\|\mbU_k\|)^2 \E^{1/2}F_k\le C_3 [2(2+q)]^{1/2}\exp\{-Cq^{-1/3}(\frac{M}{2})^{2/3}\}$
where  we have used the fact that 
$\E(F_k)=P(\|\mbW_k\|_{\max} > q^{-1/6}(\frac{M}{2})^{1/3})\le 2(2+q)\exp\{-2Cq^{-1/3}(\frac{M}{2})^{2/3}\}$ by a union bound argument and uniform Sub-Gaussinity assumption in Condition (C1). Hence, $P(|L_{11}|\ge \epsilon/8
)\le \xi_{1}(\epsilon)+\xi_2(\epsilon)$, where $\xi_{1}(\epsilon) = P(\frac{1}{n}\sum_{k=1}^n |V_{k1}|\|\mbU_k\|F_k > (8 E|V_2|+8)^{-1}\epsilon)$ and $\xi_2(\epsilon) = P((n-1)^{-1}\sum_{l=1}^{n}|V_{l2}| >E|V_2|+1)$. Choosing $\epsilon$ such that $\epsilon>(E|V_2|+1) C_3 16 [2(2+q)]^{1/2}\exp\{-Cq^{-1/3}(\frac{M}{2})^{2/3}\}$ and $\epsilon< 16 \E|V_2|+16$, then it follows from (\ref{eq:conc1}) that 
\[\xi_1(\epsilon)\le 
P(|\frac{1}{n}\sum_{k=1}^n \{|V_{k1}|\|\mbU_k\|F_k-\E(|V_{k1}|\|\mbU_k\|F_k)\}|\ge \epsilon(16 E|V_2| + 16 )^{-1})\le 2 \exp(-C n\epsilon^2).\]
In addition, we can use (\ref{eq:conc0}) to derive that $\xi_2(\epsilon)\le 2 \exp(-Cn)$. Combining these results, we have $P(|L_{11}|\ge \epsilon/8 )\le 2\exp(-C n\epsilon^2)$, when $\epsilon\in ( (E|V_2|+1) C_3 16  [2(2+q)]^{1/2}\exp\{-Cq^{-1/3}(\frac{M}{2})^{2/3}\},
16 E|V_2|+16)$. Thus if we choose $M=(\log{n})^{3/2}$, we can find a $n_0$, $\gamma\in (0,1/2)$, and $D_0$ such that when $n\ge n_0$ and $16>\epsilon>D_0n^{-\gamma}$, we have 
$P(|L_{11}|>\epsilon/8)\le  2 \exp(-C n\epsilon^2)$. Similar arguments lead to $P(|L_{12}|>\epsilon/8)\le 2 \exp(-C n\epsilon^2)$. This implies that $P(|L_1|>\epsilon/4)\le 4 \exp(-C n\epsilon^2)$ and similarly $P(|L_2|>\epsilon/4)\le 4 \exp(-C n\epsilon^2)$. Therefore $P(|\widetilde{R}_{n,12,2}|\ge \epsilon/2)\le 8\exp(-C n\epsilon^2)$. 
In view of (\ref{eq:rn121}), the desired statement follows by choosing large enough $n_0$ such that $4\log^3(n_0)>C$.




\noindent Proof of Theorem~\ref{th:main}: Notice that for any $\epsilon>0$,
\begin{eqnarray*}
	&&P(\|\MDDM_n(\mbV|\mbU)-\MDDM(\mbV|\mbU)\|_{max}>12\epsilon)\\
	&&\le P(\|\mbR_n-\mbR\|_{max}>4\epsilon)+P(\|\mbS_n-\mbS\|_{max}>4\epsilon)+P(\|\mbT_n-\mbT\|_{max}>4\epsilon)
\end{eqnarray*}
The concentration bound for $\|\mbR_n-\mbR\|_{max}$ has been obtained in Proposition~\ref{prop:R}, and
we shall address the concentration of $\|\mbT_n-\mbT\|_{max}$ in the proof below. The proof for the concentration of $\|\mbS_n-\mbS\|_{max}$ is similar and simpler so is omitted. 
Note that
$$P(\|\mbT_{n}-\mbT\|_{max}>4\epsilon)\le p^2\max_{i,j=1,\cdots,p}P(|T_{n,ij}-T_{ij}|>4\epsilon).$$
Following the same argument as
used in the beginning of proof of Proposition~\ref{prop:R}, we shall only focus on the case $(i,j)=(1,2)$ as other cases can be treated
in exactly the same manner.

Let $T_{12}=\E[V_1V_2'|\mbU'-\mbU^{''}|_q]=\E[\E(V_1)V_2^{'}|\mbU-\mbU'|_q]$ and
$T_{n,12}=n^{-3}\sum_{k,l,h=1}^{n}V_{k1}V_{h2}|\mbU_{h\cdot}-\mbU_{l\cdot}|_q$. Let
\begin{eqnarray*}
	\widetilde{T}_{n,12}&=&\frac{1}{n(n-1)(n-2)}\sum_{k<l<h}[V_{k1}V_{h2}|\mbU_{h\cdot}-\mbU_{l\cdot}|_q+V_{k1}V_{l2}|\mbU_{l\cdot}-\mbU_{h\cdot}|\\
	&&+V_{l1}V_{k2}|\mbU_{k\cdot}-\mbU_{h\cdot}|_q+V_{l1}V_{h2}|\mbU_{h\cdot}-\mbU_{k\cdot}|_q+V_{h1}V_{l2}|\mbU_{l\cdot}-\mbU_{k\cdot}|_q\\
	&&+V_{h1}V_{k2}|\mbU_{k\cdot}-\mbU_{l\cdot}|_q]=6\{n(n-1)(n-2)\}^{-1}\sum_{k<l<h}h_3(\mbW_{k},\mbW_l,\mbW_h),
\end{eqnarray*}
where $h_3$ is a kernel function for U-statistic of order three. Following the same argument to deal with $\widetilde{R}_{n,12}$ in the proof of Proposition~\ref{prop:R},
we write $\widetilde{T}_{n,12}=\widetilde{T}_{n,12,1}+\widetilde{T}_{n,12,2}$, where
\begin{eqnarray*}
	\widetilde{T}_{n,12,1}&=&6\{n(n-1)(n-2)\}^{-1}\sum_{k<l<h}h_3{\bf 1}(|h_3|\le M),\\
	\widetilde{T}_{n,12,2}&=&6\{n(n-1)(n-2)\}^{-1}\sum_{k<l<h}h_3{\bf 1}(|h_3| > M).
\end{eqnarray*}
Correspondingly, we define $T_{12}=T_{12,1}+T_{12,2}$, where
$T_{12,1}=\E[h_3{\bf 1}(|h_3|\le M)]$ and $T_{12,2}=\E[h_3{\bf 1}(|h_3|>M)]$. By using the same argument for $\widetilde{R}_{n,12,1}$, we can show that
\[P(|\widetilde{T}_{n,12,1}-T_{12,1}|\ge \epsilon)\le 2\exp(-\epsilon^2 \lfloor n/3\rfloor/(2M^2))\]
since $\widetilde{T}_{n,12,1}$ is a third order $U$-statistic. Also we note that by the same argument used in bounding $h_1$, we can get
\[|h_3|\le \frac{1}{12}\{V_{k1}^4+V_{l1}^4+V_{h1}^4+V_{k2}^4+V_{l2}^4+V_{h2}^4+8|\mbU_{h\cdot}|_q^2+8|\mbU_{k\cdot}|_q^2+8|\mbU_{l\cdot}|_q^2\}.\]
It follows from Cauchy-Schwartz inequality that
$|T_{12,2}|\le E^{1/2}(h_3^2)P^{1/2}(|h_3|>M)$. By using exactly the same argument used for (\ref{eq:tailprobh1}), we can show that $P(|h_3|>M)\le C \exp(-C_3 M^{2/3})$ for some $C_3=C_3(\sigma_0,C_0,q)>0$. Hence $|T_{12,2}|\le C \exp(-C_3 M^{2/3})$.  
We can choose $\epsilon>0$ such that $C \exp(-C_3 M^{2/3})\le \epsilon/2$. 
Therefore, for $\epsilon\ge 2C \exp(-C_3 M^{2/3}) $, we have $
P(|\widetilde{T}_{n,12,2}-T_{12,2}|>\epsilon)\le P(|\widetilde{T}_{n,12,2}|\ge \epsilon/2)$. By setting  $M=\log^{3/2}(n)$ and adopting the same argument as used in bounding $P(|\widetilde{R}_{n,12,2}|\ge \epsilon/2)$, we can derive that $P(|\widetilde{T}_{n,12,2}|\ge \epsilon/2)\le 12\exp(-C_4n\epsilon^2)$ when $n\ge n_1$ and  $\epsilon\in (D_1 n^{-\gamma_1}, 16)$ for some $C_4=C_4(\sigma_0,C_0,q)>0$, $D_1=D_1(\sigma_0,C_0,q)$, $n_1=n_1(\sigma_0,C_0,q)$ and $\gamma_1=\gamma_1(\sigma_0,C_0,q)$.

Combining the above results, we obtain that 
for $16>\epsilon> D_1n^{-\gamma_1}$ and $n\ge n_1$, we have
\begin{eqnarray*}
	P(|\widetilde{T}_{n,12}-T_{12}|\ge 2\epsilon)&\le& 12\exp(-C_4n\epsilon^2) +2\exp(-\epsilon^2 \lfloor n/3\rfloor/(2M^2))\\
	&\le& 14 \exp(-\epsilon^2n/(6\log^3(n))),
\end{eqnarray*}
where we choose $n_1$ such that $6\log^3(n_1)>C_4^{-1}$. 

Further we note that
\[T_{n,12}-T_{12}=\frac{(n-1)(n-2)}{n^2}(\widetilde{T}_{n,12}-T_{12})-\frac{3n-2}{n^2}T_{12}+\frac{n-1}{n^2}(\widetilde{R}_{n,12}-R_{12})+\frac{n-1}{n^2}R_{12}.\]
There exists a finite positive constant $C_6=C_6(\sigma_0,C_0,q)$ such that $|R_{12}|\le C_6$ and $|T_{12}|\le C_6$ so if we choose $\epsilon\ge 3C_6/n$, then
$|\frac{3n-2}{n^2}T_{12}|\le \epsilon$ and $|\frac{n-1}{n^2}R_{12}|\le \epsilon/3$. Then for $n\ge n_1$ and $\epsilon>D_1n^{-\gamma_1}$,
\begin{eqnarray*}
	P(|T_{n,12}-T_{12}|>4\epsilon)&\le&P(|\widetilde{T}_{n,12}-T_{12}|>2\epsilon)+P(|\widetilde{R}_{n,12}-R_{12}|>2\epsilon/3)\\
	&\le&14 \exp(-\epsilon^2n/(6\log^3(n)))+ 10\exp\left\{-\frac{\epsilon^2n}{36\log^{3}(n)}\right\}\\
	&\le&24 \exp\left\{-\frac{\epsilon^2n}{36\log^{3}(n)}\right\}
\end{eqnarray*}
Thus the conclusion follows from the above inequality and Proposition~\ref{prop:R}.

\qed

\section{Proofs for Theorems~\ref{consistency} \& \ref{consistency.GEP}}\label{sec:consistency.proof}
\subsection{Two Generic Algorithms and Their Properties}
We first describe two generic algorithms and their properties that will help our proof for Theorems~\ref{consistency} \& \ref{consistency.GEP}. Consider two matrices $\bA,\bB\in\mathbb{R}^{p\times p}$, their estimates $\widehat\bA,\widehat\bB\in\mathbb{R}^{p\times p}$ and vectors $\bv,\bv_0,\bv_t\in\mathbb{R}^p$. We have Algorithm~\ref{alg3} for the penalized eigen-decomposition for $\bA$, and Algorithm~\ref{alg4} for the penalized generalized eigen-decomposition for $(\bA,\bB)$. Algorithm~\ref{alg3} is originally proposed in \citet{yuan2013truncated}, and in our Algorithm~1 we use it repeatedly for $K$ times to perform penalized eigen-decomposition for MDDM. Algorithm~\ref{alg4} is originally proposed as the RIFLE Algorithm in \citet{TWLZ18}, and we use it for $K$ times to perform penalized generalized eigen-decomposition for MDDM.

\setcounter{algorithm}{2}

\begin{algorithm}[t!]
\begin{enumerate}
\item Input: $s,\widehat\bA$. 

\item Initialize $\bv_0$. 

\begin{enumerate}
\item Iterate over $t$ until convergence:
         \item Set $\bv_t=\widehat\bA\bv_{t-1}$. 
         \item If $\Vert\bv_{t}\Vert_0\le s$, set 
         $$\bv_{t}=\dfrac{\bv_{t}}{\Vert\bv_t\Vert_2};$$ 
         else 
         $$\bv_{t}=\dfrac{\HT(\bv_{t},s)}{\Vert \HT(\bv_{t},s)\Vert_2}$$
\end{enumerate}
\item Output $\bv_{\infty}$ at convergence.
\end{enumerate}
\caption{A generic penalized eigen-decomposition algorithm.}\label{alg3}
\end{algorithm}

\citet{yuan2013truncated} proved a property for Algorithm~\ref{alg3} that is important for our proof. Assume that $\bA$ has a unique leading eigenvector $\bv^*_{\bI}$ with $\Vert\bv^*_{\bI}\Vert_0\le d$. Denote $\lambda_1^{\bI},\ldots,\lambda^{\bI}_p$ as the eigenvalues. We assume that there exists a constant $\Delta_{\bI}=\lambda_1^{\bI}-\max_{j>1}\lambda_j^{\bI}$. Also for any positive integer $k'$, define 
$$
\rho(\bE_{\bA},k')=\sup_{\Vert\bu\Vert_2=1,\Vert\bu\Vert_0\le k'} |\bu^\T\bE_{\bA}\bu|,
$$
where $\bE_{\bA}=\widehat\bA-\bA$, with $\widehat\bA$ being an estimate of $\bA$. We have the following proposition.

\begin{proposition}[\citep{yuan2013truncated} c.f~Theorem 4]\label{prop.ped}
In Algorithm~\ref{alg3}, let $s=d+2s'$ with $s'\ge d$. Assume that $\rho(\bE_{\bA},s)\le \Delta_{\bI}$. Define
$$
\gamma(s)=\dfrac{\lambda_1^{\bI}-\Delta_{\bI}+\rho(\bE_{\bA},s)}{\lambda_1^{\bI}-\rho(\bE_{\bA},s)}<1, \quad \delta_{\bI}(s)=\dfrac{\sqrt{2}\rho(\bE_{\bA},s)}{\sqrt{\rho(\bE_{\bA},s)^2+(\Delta_{\bI}-2\rho(\bE_{\bA},s))^2}}.
$$
If $|\bv_0^\T\bv^*_{\bI}|\ge \theta+\delta_{\mbI}(s)$ for some $\Vert\bv_0\Vert_0\le s',\Vert\bv_0\Vert=1,$ and $\theta\in (0,1)$ such that
$$
\mu=\sqrt{[1+2\{(\frac{d}{s'})^{1/2}+\frac{d}{s'}\}]\{1-0.5\theta(1+\theta)(1-\gamma(s)^2)\}}<1,
$$
then we either have
$$
\sqrt{1-|\bv_0^\T\bv^*_{\bI}|}<\sqrt{10}\delta_{\bI}(s)/(1-\mu),
$$
or for all $t\ge 0$,
$$
\sqrt{1-|\bv_t^\T\bv^*_{\bI}|}\le \mu^t\sqrt{1-|\bv_0^\T\bv^*_{\bI}|}+\sqrt{10}\delta_{\bI}(s)/(1-\mu).
$$
\end{proposition}

\begin{algorithm}[t!]
\begin{enumerate}
\item Input: $s, \widehat\bA,\widehat\bB$ and step size $\eta>0$. 
\item Initialize $\mbv_0$. 

\item Iterate over $t$ until convergence:
\begin{enumerate}
         \item Set $\rho^{(t-1)}=\dfrac{\mbv_t^\T\widehat\mbA\mbv_t}{\mbv_t^\T\widehat\mbB\mbv_t}$.
         \item $\mbC=\mbI+(\eta/\rho^{(t-1)})\cdot(\widehat\mbA-\rho^{(t-1)}\widehat\mbB)$.
         \item $\mbv_t=\mbC\mbv_{t-1}/\Vert\mbC\mbv_{t-1}\Vert_2$.
         \item $\mbv_t=\dfrac{\HT(\mbv_t, s)}{\Vert\HT(\mbv_t, s)\Vert_2}$.
\end{enumerate}

\item Output $\bv_{\infty}$ at convergence.
\end{enumerate}
\caption{A generic  penalized generalized eigen-decomposition algorithm.}\label{alg4}
\end{algorithm}

Based on the results in \citet{TWLZ18}, we can also derive the following useful results for Algorithm~\ref{alg4}. We assume that the matrix pair $(\bA,\bB)$ has the leading generalized eigenvector $\bv^*$ such that $\Vert\bv^*\Vert_0\le d$. The generalized eigenvalues of $(\bA,\bB)$ are referred to as $\lambda_j,j=1,\ldots,p$ and their estimates are $\widehat{\lambda}_j,j=1,\ldots,p$. We introduce the following notation:
\bea
\text{cr}(\bA,\bB)&=&\min_{\bv:\Vert\bv\Vert_2=1} \{(\bv^\T\bA\bv)^2+(\bv^\T\bB\bv)^2\}^{1/2}>0\\
\text{cr}(k')&=&\inf_{F: |F|\le k'}\text{cr}(\bA_F,\bB_F),\\
\delta(k')&=&\sqrt{\rho(\bE_{\bA},k')^2+\rho(\bE_{\bB},k')^2},
\eea
where $\bE_{\bB}=\bB-\widehat\bB$, with $\widehat\bB$ being an estimate of $\bB$.
Also denote $\kappa(\bB)$ as the condition number of $\bB$ and $\omega_1(F)=\sup_{\Vert\bu\Vert_0\subset F, \Vert\bu\Vert_2=1}\frac{\bu^\T\bA\bu}{\bu^\T\bB\bu}$ for any index set $F$. We consider the following assumption:
\begin{assumption}\label{assumption1}

For sufficiently large $n$, there are constants $b,c>0$, such that $\dfrac{\delta(s)}{cr(s)}\le b$ and $\rho(\bE_{\bB},s)\le c\lambda_{\min}(\bB)$ for any $s=o(n)$.
\end{assumption}

We also denote $c_{\text{upper}}=\dfrac{1+c}{1-c}$ for $c$ defined in Assumption~\ref{assumption1}. We estimate $\bv^*$ with the RIFLE algorithm with the step size $\eta$. We choose $\eta$ such that $\eta\lambda_{\max}(\bB)<1/(1+c)$. Further, in the RIFLE algorithm, let $s=2s'+d$ and choose $s'=Cd$ for sufficiently large $C$. The initial value $\bv_0$ satisfies that $\Vert\bv_0\Vert_2=1$.

\begin{proposition}[Based on Theorem~1 and Corollary~1 in \citet{TWLZ18}]\label{prop.gep1}  Under Assumption~\ref{assumption1}, we have the following conclusions:
\begin{enumerate}
\item For any $F$ such that $\supp(\bv_0)\subset F$, there exists a constant $a$ such that
\beq
(1-a)\omega_1(F)\le\hat \omega_1(F)\le (1+a)\omega_1(F).
\eeq

\item Choose $\eta$ such that 
\beq\label{prop.gep1.gamma}
\nu=\sqrt{1+2\{(d/s')^{1/2}+d/s'\}}\sqrt{1-\frac{1+c}{8}\eta\lambda_{\min}(\bB)\frac{1-\alpha}{c_{\text{upper}}\kappa(\bB)+\alpha}}<1,
\eeq
where $\alpha=\frac{(1+a)\lambda_2}{(1-a)\lambda_1}$. Input an initial vector $\bv_0$ such that $\dfrac{|(\bv^*)^\T\bv_0|}{\Vert\bv^*\Vert_2}\ge 1-\theta(\bA,\bB)$, where 
\beq
\theta(\bA,\bB)=\min\left[\dfrac{1}{8c_{\text{upper}}\kappa(\bB)},\dfrac{1/\alpha-1}{3c_{\text{upper}}\kappa(\bB)},\dfrac{1-\alpha}{30(1+c)c_{\text{upper}}^2\eta\lambda_{\max}(\bB)\kappa^2(\bB)\{c_{\text{upper}}\kappa(\bB)+\alpha\}}\right].
\eeq
Further denote 
\beq\label{prop.gep1.xi} 
\xi=\min_{j>1}\dfrac{\lambda_1-(1+a)\lambda_j}{\sqrt{1+\lambda_1^2}\sqrt{1+(1-a)^2\lambda_j^2}}.
\eeq
Assume that $\xi>\delta(s)/cr(s)$ and we have
\beq
\sqrt{1-\dfrac{|(\bv^*)^\T\bv_t|}{\Vert\bv^*\Vert_2}}\le \nu^t\sqrt{\theta(\bA,\bB)}+\dfrac{\sqrt{10}}{1-\nu}\dfrac{2}{\xi\{cr(s)-\delta(s)\}}\delta(s).
\eeq
\end{enumerate}

\end{proposition}

We rewrite Proposition~\ref{prop.gep1} in the following more user-friendly form. We denote
\bea
\phi&=&\lambda_1-\lambda_2,\quad a^*=\min\{1/2,\frac{\Delta}{\lambda_1+\lambda_2},\frac{\lambda_{\min}(\bB)}{2}\}\\
\xi^*&=&\dfrac{\lambda_1-\lambda_2}{2(1+\lambda_1^2)},\quad\alpha^*=\frac{(1+a^*)\lambda_2}{(1-a^*)\lambda_1}.
\eea

We have the following lemma.
\begin{lemma}\label{lem.GEP1}
Assume that Assumption~\ref{assumption1} holds. Choose $\eta$ such that 
\beqn
\nu^*=\sqrt{1+2\{(d/s')^{1/2}+d/s'\}}\sqrt{1-\frac{1+c}{8}\eta\lambda_{\min}(\bB)\frac{1-\frac{\lambda_2}{\lambda_1}}{c_{\text{upper}}\kappa(\bB)+3\frac{\lambda_2}{\lambda_1}}}<1.
\eeqn
Also assume that $\delta(s)<\min\{1/2,\lambda_{\min}(\bB)/2\}$, $\dfrac{2\delta(s)}{\lambda_{\min}(\bB)}+\dfrac{2\delta(s)}{\lambda_{\min}(\bB)\lambda_1(F)}<a^*$ and $\xi^*>2\delta(s)/cr(s)$. Input an initial vector $\bv_0$ such that $\dfrac{|(\bv^*)^\T\bv_0|}{\Vert\bv^*\Vert_2}\ge 1-\theta^*(\bA,\bB)$, where $0<\theta^*(\bA,\bB)<1$, 
\beq
\theta^*(\bA,\bB)=\min\left[\dfrac{1}{8c_{\text{upper}}\kappa(\bB)},\dfrac{1/\alpha^*-1}{3c_{\text{upper}}\kappa(\bB)},\dfrac{1-\alpha^*}{30(1+c)c_{\text{upper}}^2\eta\lambda_{\max}(\bB)\kappa^2(\bB)\{c_{\text{upper}}\kappa(\bB)+\alpha^*\}}\right].
\eeq

We have
\beq
|\sin\Theta(\bv^*,\bv_{\infty})|\le \dfrac{\sqrt{20}}{1-\nu^*}\dfrac{2}{\xi^*\{cr(s)-\delta(s)\}}\delta(s).
\eeq 
\end{lemma}

\begin{proof}[Proof of Lemma~\ref{lem.GEP1}]
We prove Lemma~\ref{lem.GEP1} by showing that all the conditions in Proposition~\ref{prop.gep1} are met. Note that
\beq
\hat\omega_1(F)=\sup_{\Vert\bu\Vert_2=1,\supp(\bu)\subset F}\dfrac{\bu^\T\widehat{\bA}\bu}{\bu^\T\widehat{\bB}\bu}= \sup_{\Vert\bu\Vert_2=1,\supp(\bu)\subset F}\dfrac{\bu^\T\bA\bu+\bu^\T(\widehat{\bA}-\bA)\bu}{\bu^\T\bB\bu+\bu^\T(\widehat{\bB}-\bB)\bu}.
\eeq

It is obvious that
\beq
\sup_{\Vert\bu\Vert_2=1,\supp(\bu)\subset F}\dfrac{\bu^\T\bA\bu-\delta(s)}{\bu^\T\bB\bu+\delta(s)}\le \hat\omega_1(F)\le \sup_{\Vert\bu\Vert_2=1,\supp(\bu)\subset F}\dfrac{\bu^\T\bA\bu+\delta(s)}{\bu^\T\bB\bu-\delta(s)}
\eeq 
Because $a^*>\dfrac{2\delta(s)}{\lambda_{\min}(\bB)}+\dfrac{2\delta(s)}{\lambda_{\min}(\bB)\lambda_1(F)}$, by Lemma~\ref{lem: a}, we have that Assumption~\ref{assumption1} implies
\beq
(1-a^*)\omega_1(F)\le\hat{\omega}_1(F)\le (1+a^*)\omega_1(F).
\eeq

Also, by our definition, $a^*\le 1/2$. It follows that $\frac{\lambda_2}{\lambda_1}\le \frac{(1+a)\lambda_2}{(1-a)\lambda_1}=\alpha\le\frac{3\lambda_2}{\lambda_1}$. Hence, $\nu\le\nu^*<1$, where $\nu$ is defined in \eqref{prop.gep1.gamma}. In addition, because $a^*\le\frac{\phi}{\lambda_1+\lambda_2}\le \frac{\phi}{2\lambda_2}$, we have $\xi\ge\xi^*$, where $\xi$ is defined in \eqref{prop.gep1.xi}. Finally, because $\dfrac{1}{8c_{\text{upper}}\kappa(\bB)}<1$, we have $\theta^*(\bA,\bB)<1$. Because $a^*\le \dfrac{\phi}{\lambda_1+\lambda_2}$, we have $1-\gamma^*>0$ and thus $\theta^*(\bA,\bB)>0$.

By Proposition~\ref{prop.gep1}, we have
\bea
\sqrt{1-\dfrac{|(\bv^*)^\T\bv_t|}{\Vert\bv^*\Vert_2}}&\le& \nu^t\sqrt{\theta(\bA,\bB)}+\dfrac{\sqrt{10}}{1-\nu}\dfrac{2}{\xi\{cr(s)-\delta(s)\}}\delta(s)\\
&\le& (\nu^*)^t\sqrt{\theta^*(\bA,\bB)}+\dfrac{\sqrt{10}}{1-\nu^*}\dfrac{2}{\xi^*\{cr(s)-\delta(s)\}}\delta(s).
\eea

Let $t\rightarrow \infty$ and we have that $\sqrt{1-\dfrac{|(\bv^*)^\T\bv_{\infty}|}{\Vert\bv^*\Vert_2}}\le \dfrac{\sqrt{10}}{1-\nu^*}\dfrac{2}{\xi^*\{cr(s)-\delta(s)\}}\delta(s).$ Finally, we note that $1-\dfrac{|(\bv^*)^\T\bv_{\infty}|}{\Vert\bv^*\Vert_2}=1-|\cos\Theta(\bv^*,\bv_{\infty})|$. Since $\sin^2\Theta(\bv^*,\bv_{\infty})=(1+|\cos\Theta(\bv^*,\bv_{\infty})|)(1-|\cos\Theta(\bv^*,\bv_{\infty})|)\le 2(1-|\cos\Theta(\bv^*,\bv_{\infty})|)$, we have the desired conclusion.
\end{proof}

\begin{lemma}\label{lem: a}
Consider two symmetric matrices $\bA,\bB$, where $\lambda_{\min}(\bB)>0$. For any $F\subset \{1,\ldots,p\}$, denote 
\beq
\lambda_1(F)=\sup_{\Vert\bu\Vert_2=1,\supp(\bu)\subset F}\dfrac{\bu^\T\bA\bu}{\bu^\T\bB\bu}.
\eeq

 For any $0<\epsilon<\min\{\frac{1}{2},\dfrac{\lambda_{\min}(\bB)}{2}\}$, we have that
\bea
\sup_{\Vert\bu\Vert_2=1,\supp(\bu)\subset F}\dfrac{\bu^\T\bA\bu+\epsilon}{\bu^\T\bB\bu-\epsilon}&\le& (1+a)\lambda_1(F),\label{lem: a.eq1}\\
\sup_{\Vert\bu\Vert_2=1,\supp(\bu)\subset F}\dfrac{\bu^\T\bA\bu-\epsilon}{\bu^\T\bB\bu+\epsilon}&\ge& (1-a)\lambda_1(F),\label{lem: a.eq2}
\eea
where $a=\dfrac{2\epsilon}{\lambda_{\min}(\bB)}+\dfrac{2\epsilon}{\lambda_{\min}(\bB)\omega_1(F)}$.
\end{lemma}

\begin{proof}[Proof of Lemma~\ref{lem: a}]
For \eqref{lem: a.eq1}, note that, for any $\bu$, we have
$\bu^\T\bB\bu-\epsilon\ge \{1-\frac{\epsilon}{\lambda_{\min}(\bB)}\}\bu^\T\bB\bu.$
So 
\bea
&&\sup_{\Vert\bu\Vert_2=1,\supp(\bu)\subset F}\dfrac{\bu^\T\bA\bu+\epsilon}{\bu^\T\bB\bu-\epsilon}\\
&\le& \dfrac{1}{1-\frac{\epsilon}{\lambda_{\min}(\bB)}}\cdot\sup_{\Vert\bu\Vert_2=1,\supp(\bu)\subset F}\dfrac{\bu^\T\bA\bu}{\bu^\T\bB\bu}+\dfrac{1}{1-\frac{\epsilon}{\lambda_{\min}(\bB)}}\cdot\sup_{\Vert\bu\Vert_2=1,\supp(\bu)\subset F}\dfrac{1}{\bu^\T\bB\bu}\\
&\le&\dfrac{1}{1-\frac{\epsilon}{\lambda_{\min}(\bB)}}\cdot\omega_1(F)+\dfrac{2\epsilon}{\lambda_{\min}(\bB)\omega_1(F)}\cdot \omega_1(F),
\eea
where in the last inequality we use the fact that $\epsilon<\dfrac{\lambda_{\min}(\bB)}{2}$. Also note that, for any $0<x<1/2$, we have $\dfrac{1}{1-x}<1+2x$. It follows that
$
\dfrac{1}{1-\frac{\epsilon}{\lambda_{\min}(\bB)}}\le 1+\dfrac{2\epsilon}{\lambda_{\min}(\bB)}.$ Hence,
\beq
\sup_{\Vert\bu\Vert_2=1,\supp(\bu)\subset F}\dfrac{\bu^\T\bA\bu+\epsilon}{\bu^\T\bB\bu-\epsilon}\le\left\{1+\dfrac{2\epsilon}{\lambda_{\min}(\bB)}+\dfrac{2\epsilon}{\lambda_{\min}(\bB)\omega_1(F)}\right\}\omega_1(F)=(1+a)\omega_1(F).
\eeq

Similarly, for \eqref{lem: a.eq2}, we note that, for any $\bu$,
$\bu^\T\bB\bu+\epsilon\le\{1+\frac{\epsilon}{\lambda_{\min}(\bB)}\}\bu^\T\bB\bu.$ So
\bea
\sup_{\Vert\bu\Vert_2=1,\supp(\bu)\subset F}\dfrac{\bu^\T\bA\bu-\epsilon}{\bu^\T\bB\bu+\epsilon}&\ge&\dfrac{1}{1+\frac{\epsilon}{\lambda_{\min}(\bB)}}\cdot \omega_1(F)-\dfrac{1}{1+\frac{\epsilon}{\lambda_{\min}(\bB)}}\cdot \dfrac{\epsilon}{\lambda_{\max}(\bB)},\\
&\ge&\dfrac{1}{1+\frac{\epsilon}{\lambda_{\min}(\bB)}}\cdot \omega_1(F)-\dfrac{\epsilon}{\lambda_{\max}(\bB)\omega_1(F)}\cdot \omega_1(F).
\eea
Because $\frac{1}{1+x}>1-x$ for any $x>0$, we have $\dfrac{1}{1+\frac{\epsilon}{\lambda_{\min}(\bB)}}>1-\dfrac{\epsilon}{\lambda_{\min}(\bB)}$. Hence,
\beq
\sup_{\Vert\bu\Vert_2=1,\supp(\bu)\subset F}\dfrac{\bu^\T\bA\bu-\epsilon}{\bu^\T\bB\bu+\epsilon}\ge(1-\dfrac{\epsilon}{\lambda_{\min}(\bB)}-\dfrac{\epsilon}{\lambda_{\max}(\bB)\omega_1(F)})\omega_1(F)\ge (1-a)\omega_1(F).
\eeq
The conclusion follows.
\end{proof}

\subsection{Additional technical lemmas}

We first derive several lemmas concerning a parameter $\bbeta_k$ (either in the penalized eigen-decomposition or the penalized generalized eigen-decomposition) and its estimate $\widehat{\bbeta}_k$. We denote $\eta_k=|\sin\Theta(\bbeta_k,\hat{\bbeta}_k)|$, and $\hat\lambda_k=\widehat{\bbeta}_k^\T\widehat{\bM}\widehat{\bbeta}_k$.

\begin{lemma}\label{lem1}
If $\Vert\bbeta_k\Vert_0\le s, k=1,\ldots,K$, we have that
\bea
\Vert\vecc(\bbeta_k\bbeta_k^\T-\hat\bbeta_k\hat\bbeta_k^\T)\Vert_1&\le& 2s\eta_k\label{lem1.l1norm}
\eea
\end{lemma}
\begin{proof}[Proof of Lemma~\ref{lem1}]
For \eqref{lem1.l1norm}, set $\bzeta_k=\vecc(\bbeta_k\bbeta_k^\T-\hat\bbeta_k\hat\bbeta_k^\T)$. We have that
\bea
\Vert\bzeta_k\Vert_2^2&=&\Vert\bbeta_k\bbeta_k^\T-\hat\bbeta_k\hat \bbeta_k^\T\Vert_F^2=\Tr((\bbeta_k\bbeta_k^\T-\hat\bbeta_k\hat\bbeta_k^\T)(\bbeta_k\bbeta_k^\T-\hat\bbeta_k\hat\bbeta_k^\T))\\
&=&2(1-(\bbeta_k^\T\hat\bbeta_k)^2)=2\eta^2_k
\eea
Hence, by the Cauchy-Schwarz inequality,
\bea
\Vert\bzeta_k\Vert_1\le\sqrt{\Vert\bzeta_k\Vert_0}\Vert\bzeta_k\Vert_2\le \sqrt{2s^2}\cdot\sqrt{2\eta_k^2}=2s\eta_k\label{lem1.eq3}
\eea

\end{proof}

\begin{lemma}\label{lem3}
If $\Vert\bbeta_k\Vert_0\le s$, we have that
\beq
\Vert\vecc(\bbeta_k\bbeta_k^\T)\Vert_1\le s.
\eeq
\end{lemma}

\begin{proof}[Proof of Lemma~\ref{lem3}]
By the Cauchy-Schwarz inequality, we have that
\beq
\Vert\vecc(\bbeta_k\bbeta_k^\T)\Vert_1\le \sqrt{\Vert \vecc(\bbeta_k\bbeta_k^\T)\Vert_0}\Vert\vecc(\bbeta_k\bbeta_k^\T)\Vert_2
\eeq

Note that $\Vert \vecc(\bbeta_k\bbeta_k^\T)\Vert_0\le s^2$ and
\beq
\Vert\vecc(\bbeta_k\bbeta_k^\T)\Vert_2^2=\Vert\bbeta_k\bbeta_k^\T\Vert_F^2=\Tr(\bbeta_k\bbeta_k^\T\bbeta_k\bbeta_k^\T)=1
\eeq
where we use the fact that $\bbeta_k^\T\bbeta_k=1$. And we have the desired conclusion.
\end{proof}

Throughout the rest of this section, we also repeatedly use the fact that, for a vector $\bu$, if $\Vert\bu\Vert_2=1$, $\Vert\bu\Vert_0\le s$, we must have that $\Vert\bu\Vert_1\le\sqrt{s}$.

\subsection{Proof for Theorem~\ref{consistency}}

In this subsection, we assume that $\bolSigma_{\bX}=\bI$, $\widehat{\bolbeta}_k$ are solutions produced by Algorithm~\ref{alg1} for the penalized eigen-decomposition problem, and $\widehat{\lambda}_k=(\widehat{\bolbeta}_k)^\T\widehat\mbM\widehat{\bolbeta}_k$. We assume all the conditions in Theorem~\ref{consistency}. We have the following result.

\begin{lemma}\label{lem2.2}
If $\Vert\widehat{\bolbeta}_k\Vert_0\le s, \Vert\bolbeta_k\Vert_0\le s$, we have that
\beq
|\hat\lambda_k-\lambda_k|\le s(2\eta_k+1)\epsilon+(\lambda_1+\lambda_{k})\eta_k^2.
\eeq
\end{lemma}

\begin{proof}[Proof of Lemma~\ref{lem2.2}]

Note that 
\bea
&&|\hat\lambda_k-\lambda_k|=|\langle\widehat{\bM},\hat\bbeta_k\hat\bbeta_k^\T\rangle-\langle\bM,\bbeta_k\bbeta_k^\T\rangle|\\
&\le& |\langle\widehat{\bM}-\bM,\hat\bbeta_k\hat\bbeta_k^\T-\bbeta_k\bbeta_k^\T\rangle|+|\langle\widehat{\bM}-\bM,\bbeta_k\bbeta_k^\T\rangle|+|\langle\bM,\hat\bbeta_k\hat\bbeta_k^\T-\bbeta_k\bbeta_k^\T\rangle|\\
&\equiv& L_1+L_2+L_3
\eea

By Lemma~\ref{lem1},
\bea
L_1\le \Vert\widehat{\bM}-\bM\Vert_{\max}\Vert\vecc(\hat\bbeta_k\hat\bbeta_k^\T-\bbeta_k\bbeta_k^\T)\Vert_1\le 2s \eta_k\Vert\widehat{\bM}-\bM\Vert_{\max}\le 2s \eta_k\epsilon.
\eea

By Lemma~\ref{lem3},
\bea
L_2\le \Vert\widehat{\bM}-\bM\Vert_{\max}\Vert\vecc(\bbeta_k\bbeta_k^\T)\Vert_1\le s\Vert\widehat{\bM}-\bM\Vert_{\max}\le s\epsilon.
\eea

For $L_3$, note that $\bM=\sum_{j=1}^p\lambda_j\bbeta_j\bbeta_j^\T$, which implies that
$$
\langle\bM,\widehat{\bbeta}_k\widehat{\bbeta}_k^\T\rangle=\sum_{j=1}^p\lambda_j(\widehat{\bbeta}_k^\T\bbeta_j)^2=\sum_{j=1}^p\lambda_j\cos^2\bTheta(\hat\bbeta_k,\bbeta_j).
$$
Also note that
$\sum_{j=1}^p\cos^2\bTheta(\hat\bbeta_k,\bbeta_j)=1$. Hence,
\bea
L_3&=&|\sum_{j=1}^p\lambda_j\cos^2\bTheta(\bbeta_j,\hat\bbeta_k)-\lambda_k| \\
&\le&\sum_{j\ne k}\lambda_j\cos^2\bTheta(\widehat{\bbeta}_k,\bbeta_j)+\lambda_k|\cos^2\bTheta(\widehat{\bbeta}_k,\bbeta_k)-1|\\
&\le&\lambda_1\sum_{j\ne k}\cos^2\bTheta(\widehat{\bbeta}_k,\bbeta_j)+\lambda_k(1-\cos^2\bTheta(\widehat{\bbeta}_k,\bbeta_k))\\
&\le& (\lambda_1+\lambda_k)(1-\cos^2\bTheta(\bbeta_k,\hat\bbeta_k))\\
&=&(\lambda_1+\lambda_{k})\eta_k^2
\eea

\end{proof}

In Lemmas~\ref{lem.beta1.I}--\ref{lem.sin.k.I}, we assume that the event $\Vert\widehat{\mbM}-\mbM\Vert_{\max}\le\epsilon$ has happened.

\begin{lemma}\label{lem.beta1.I}
For the first direction $\widehat{\bolbeta}_1$, we have that $|\sin\Theta(\widehat{\bolbeta}_1,\bolbeta_1)|\le Cs\epsilon$ and $|\hat{\lambda}_1-\lambda_1|\le Cs\epsilon$. 
\end{lemma}

\begin{proof}[Proof of Lemma~\ref{lem.beta1.I}]
It is easy to see that 
\bea
\rho(\widehat{\bM}-\bM,s)=\sup_{\Vert\bu\Vert_2=1,\Vert\bu\Vert_0\le s} s|\bu^\T(\widehat{\bM}-\bM)\bu|\le \sup_{\Vert\bu\Vert_2=1,\Vert\bu\Vert_0\le s}\Vert\bu\Vert_1^2\cdot \Vert\widehat{\bM}-\bM\Vert_{\max}\le s\epsilon.
\eea

Under our assumptions about $\epsilon$, by Proposition~\ref{prop.ped}, we have $\lvert\sin\Theta(\widehat{\bolbeta}_1,\bolbeta_1)\rvert\le Cs\epsilon$ at convergence. Lemma~\ref{lem2.2} further implies that $|\hat{\lambda}_1-\lambda_1|\le Cs\epsilon$.
\end{proof}

\begin{lemma}\label{lem.M.k+1.I}
If $|\sin\Theta(\widehat{\bolbeta}_j,\bolbeta_j)|\le Cs\epsilon$ for sufficiently small $\epsilon$, all $j\le k$, then $\rho(\widehat{\mbM}_{k+1}-\mbM_{k+1},s)\le Cs\epsilon$.
\end{lemma}

\begin{proof}[Proof of Lemma~\ref{lem.M.k+1.I}]
Define $\bN_{l}=\hat{\lambda}_l\hat{\bbeta}_l\hat{\bbeta}_l^\T-\lambda_l\bbeta_l\bbeta_l^\T$. Then
\beq
\widehat{\bM}_{k+1}-\bM_{k+1}=(\widehat{\bM}-\bM)-\sum_{l\le k} \bN_l.
\eeq
It follows that 
\beq\label{lem.M.k+1.eq1}
\rho(\widehat{\bM}_{k+1}-\bM_{k+1},s)\le \rho(\widehat{\bM}-\bM,s)+\sum_{l\le k} \rho(\bN_l,s).
\eeq
According to the proof in Lemma~\ref{lem.beta1.I}, $\rho(\widehat{\bM}-\bM,s)\le s\epsilon$. 

For any vector $\bu$, we have
\bea
\bu^\T\bN_l\bu&=&\widehat{\lambda}_l(\bu^\T\widehat{\bbeta}_l)^2-{\lambda}_l(\bu^\T{\bbeta}_l)^2\\
&=&(\widehat{\lambda}_l-\lambda_l)(\bu^\T\widehat{\bbeta}_l)^2+\lambda_l\{(\bu^\T\bbeta_l)^2-(\bu^\T\widehat\bbeta_l)^2\}\equiv L_1+L_2
\eea
By Lemma~\ref{lem2.2}, we have that $|\hat\lambda_l-\lambda_l|\le Cs\epsilon$ when $|\sin\Theta(\widehat{\bolbeta}_j,\bolbeta_j)|\le Cs\epsilon$. Also, $|\bu^\T\bbeta_l|\le \Vert\bu\Vert_2\cdot \Vert\bbeta_l\Vert_2=1$. It follows that $L_1\le Cs\epsilon$. For $L_2$, we assume that $\cos\bTheta(\widehat\bbeta_l,\bbeta_l)>0$ without loss of generality, because otherwise we can consider the proof for $-\widehat\bbeta_l$. Note that 
\bean
L_2&\le&\lambda_l|\mbu^\T(\widehat{\bolbeta}_l-\bolbeta_l)|\cdot |\bu^\T(\widehat{\bolbeta}_l+\bolbeta_l)|\\
&\le&\lambda_l\Vert\mbu\Vert_2^2\cdot \Vert\widehat{\bolbeta}_l-\bolbeta_l\Vert_2(\Vert\widehat{\bolbeta}_l\Vert_2+\Vert\bolbeta_l\Vert_2)\\
&=&C\Vert\widehat{\bolbeta}_l-\bolbeta_l\Vert_2\\
&=&C\sqrt{2-2\widehat{\bolbeta}_l^\T\bolbeta_l}=\sqrt{2(1-\cos\Theta(\widehat{\bolbeta}_l,\bolbeta_l))}\\
&=&C\sqrt{1-\cos^2\Theta(\widehat{\bolbeta}_l,\bolbeta_l)}/\sqrt{1+\cos\Theta(\widehat{\bolbeta}_l,\bolbeta_l)}\\
&\le&C|\sin\Theta(\widehat{\bolbeta}_l,\bolbeta_l)|\\
&\le& Cs\epsilon.
\eean
Hence, $\rho(\bN_l,s)\le Cs\epsilon$ and the conclusion follows.
\end{proof}

\begin{lemma}\label{lem.sin.k.I}
Assume that $|\sin\Theta(\hat{\bbeta}_l,\bbeta_l)|<Cs\epsilon$ for any $l<k$, where $0<\epsilon<\min\{\dfrac{\Delta}{4s},\theta\}$ for $\theta$ defined in Theorem~\ref{consistency}. We have that $|\sin\Theta(\hat\bbeta_k,\bbeta_k)|\le Cs\epsilon$.
\end{lemma}
\begin{proof}[Proof of Lemma~\ref{lem.sin.k.I}]
The conclusion follows from Lemma~\ref{lem.M.k+1.I} and Proposition~\ref{prop.ped}.
\end{proof}

\begin{proof}[Proof of Theorem~\ref{consistency}]
Under the event $\Vert\widehat{\mbM}-\mbM\Vert_{\max}\le\epsilon$, we have that $|\sin\Theta(\widehat{\bolbeta}_k,\bolbeta_k)|\le Cs\epsilon$ by Lemmas~\ref{lem.beta1.I}--\ref{lem.sin.k.I}. Then by Theorem~\ref{th:main} we have the desired conclusion.
\end{proof}

\subsection{Proof for Theorem~\ref{consistency.GEP}}

In this subsection we prove Theorem~\ref{consistency.GEP}, where $\bSigma_{\bX}$ could be different from the identity matrix. We first present a simple lemma, which is a modified version of Lemma~6 in \citet{mai2019iterative}.

\begin{lemma}\label{angle}
For two vectors $\bu,\bv$ and a positive definite matrix $\bSigma$, define $\xi_{\bI}=1-\cos\bTheta(\bv,\bu), \xi_{\bSigma}=1-\cos\bTheta(\bSigma^{1/2}\bv,\bSigma^{1/2}\bu)$. We have that 
$$
\dfrac{\lambda_{\min}(\bSigma)}{\lambda_{\max}(\bSigma)}\xi_{\bSigma}\le \xi_{\bI}\le \dfrac{\lambda_{\max}(\bSigma)}{\lambda_{\min}(\bSigma)}\xi_{\bSigma}
$$
\end{lemma}
\begin{proof}[Proof of Lemma~\ref{angle}]
We first show the latter half of the desired inequality. Without loss of generality, we assume that $\bu^\T\bSigma\bu=1,\bv^\T\bSigma\bv=1$, because we can always normalize $\bu,\bv$ to satisfy these conditions.

Note that
\begin{eqnarray*}
\lambda_{\min}(\bSigma)(\bu-\bv)^\T(\bu-\bv)&\le&(\bu-\bv)^\T\bSigma(\bu-\bv)\\
&=&\bu^\T\bSigma\bu-2\bu^\T\bSigma\bv+\bv^\T\bSigma\bv=2\xi_{\bSigma}\\
\lambda_{\max}(\bSigma)\bu^\T\bu&\ge&\bu^\T\bSigma\bu=1\\
\lambda_{\max}(\bSigma)\bv^\T\bv&\ge&\bv^\T\bSigma\bv=1
\end{eqnarray*}
Consequently,
\begin{eqnarray*}
&&(\bu-\bv)^\T(\bu-\bv)\le\frac{2\xi_{\bSigma}}{\lambda_{\min}(\bSigma)}\\
&&\bu^\T\bu\ge 1/\lambda_{\max}(\bSigma), \bv^\T\bv\ge 1/\lambda_{\max}(\bSigma)
\end{eqnarray*}

Now
\begin{eqnarray*}
\dfrac{2\lambda_{\max}(\bSigma)\xi_{\bSigma}}{\lambda_{\min}(\bSigma)}&\ge&\dfrac{2\xi_{\bSigma}}{\lambda_{\min}(\bSigma)\sqrt{\bu^\T\bu}\sqrt{\bv^\T\bv}}\\
&\ge&\dfrac{(\bu-\bv)^\T(\bu-\bv)}{\sqrt{\bu^\T\bu}\sqrt{\bv^\T\bv}}=\dfrac{\sqrt{\bu^\T\bu}}{\sqrt{\bv^\T\bv}}+\dfrac{\sqrt{\bv^\T\bv}}{\sqrt{\bu^\T\bu}}-2\dfrac{\bu^\T\bv}{\sqrt{\bu^\T\bu}\sqrt{\bv^\T\bv}}\\
&\ge&2-2\dfrac{\bu^\T\bv}{\sqrt{\bu^\T\bu}\sqrt{\bv^\T\bv}}=2\xi_{\bI},
\end{eqnarray*}
and we have the second half of the inequality. For the first half of the inequality, define $\bu^*=\bSigma^{1/2}\bu,\bv^*=\bSigma^{1/2}$. Apply the second half of the inequality to vectors $\bu^*,\bv^*$ and matrix $\bSigma^{-1}$, and we have the desired conclusion.
\end{proof}

We now show the following lemma parallel to Lemma~\ref{lem2.2}. Recall that $\eta_k=|\sin\bTheta(\widehat{\bbeta}_k,\bbeta_k)|$.
\begin{lemma}\label{lem2.2.sigma}
If $\Vert\widehat{\bolbeta}_k\Vert_0\le s, \Vert\bolbeta_k\Vert_0\le s$, we have that
\beq
|\hat\lambda_k-\lambda_k|\le s(2\eta_k+1)\epsilon+2\dfrac{\lambda_{\max}(\bSigma_{\bX})}{\lambda_{\min}(\bSigma_{\bX})}(\lambda_1+\lambda_k)\eta_k^2.
\eeq
\end{lemma}

\begin{proof}[Proof of Lemma~\ref{lem2.2.sigma}]

Note that 
\bea
&&|\hat\lambda_k-\lambda_k|=|\langle\widehat{\bM},\hat\bbeta_k\hat\bbeta_k^\T\rangle-\langle\bM,\bbeta_k\bbeta_k^\T\rangle|\\
&\le& |\langle\widehat{\bM}-\bM,\hat\bbeta_k\hat\bbeta_k^\T-\bbeta_k\bbeta_k^\T\rangle|+|\langle\widehat{\bM}-\bM,\bbeta_k\bbeta_k^\T\rangle|+|\langle\bM,\hat\bbeta_k\hat\bbeta_k^\T-\bbeta_k\bbeta_k^\T\rangle|\\
&\equiv& L_1+L_2+L_3
\eea

By Lemma~\ref{lem1},
\bea
L_1\le \Vert\widehat{\bM}-\bM\Vert_{\max}\Vert\vecc(\hat\bbeta_k\hat\bbeta_k^\T-\bbeta_k\bbeta_k^\T)\Vert_1\le 2s \eta_k\Vert\widehat{\bM}-\bM\Vert_{\max}\le 2s \eta_k\epsilon.
\eea

By Lemma~\ref{lem3},
\bea
L_2\le \Vert\widehat{\bM}-\bM\Vert_{\max}\Vert\vecc(\bbeta_k\bbeta_k^\T)\Vert_1\le s\Vert\widehat{\bM}-\bM\Vert_{\max}\le s\epsilon.
\eea

For $L_3$, note that $\bM=\bSigma_{\bX}\left(\sum_{j=1}^p\lambda_j\bbeta_j\bbeta_j^\T\right)\bSigma_{\bX}$ by the proof of Lemma~\ref{sub}, which implies that
$$
\langle\bM,\widehat{\bbeta}_k\widehat{\bbeta}_k^\T\rangle=\sum_{j=1}^p\lambda_j(\widehat{\bbeta}_k^\T\bSigma_{\bX}\bbeta_j)^2=\sum_{j=1}^p\lambda_j\cos^2\bTheta(\bSigma_{\bX}^{1/2}\hat\bbeta_k,\bSigma_{\bX}^{1/2}\bbeta_j)
$$
Also note that
$\sum_{j=1}^p\cos^2\bTheta(\bSigma_{\bX}^{1/2}\hat\bbeta_k,\bSigma_{\bX}^{1/2}\bbeta_j)=1$. Without loss of generality, assume that $\cos\bTheta(\bbeta_k,\hat\bbeta_k)>0$. We have that
\bean
L_3&=&|\sum_{j=1}^p\lambda_j\cos^2\bTheta(\bSigma_{\bX}^{1/2}\bbeta_j,\bSigma_{\bX}^{1/2}\hat\bbeta_k)-\lambda_k| \\
&\le&\sum_{j\ne k}\lambda_j\cos^2\bTheta(\bSigma_{\bX}^{1/2}\widehat{\bbeta}_k,\bSigma_{\bX}^{1/2}\bbeta_j)+\lambda_k|\cos^2\bTheta(\bSigma_{\bX}^{1/2}\widehat{\bbeta}_k,\bSigma_{\bX}^{1/2}\bbeta_k)-1|\\
&\le&\lambda_1\sum_{j\ne k}\cos^2\bTheta(\bSigma_{\bX}^{1/2}\widehat{\bbeta}_k,\bSigma_{\bX}^{1/2}\bbeta_j)+\lambda_k(1-\cos^2\bTheta(\bSigma_{\bX}^{1/2}\widehat{\bbeta}_k,\bSigma_{\bX}^{1/2}\bbeta_k))\\
&\le& (\lambda_1+\lambda_k)(1-\cos^2\bTheta(\bSigma_{\bX}^{1/2}\bbeta_k,\bSigma_{\bX}^{1/2}\hat\bbeta_k))\\
&=&(\lambda_1+\lambda_k)(1-\cos\bTheta(\bSigma_{\bX}^{1/2}\bbeta_k,\bSigma_{\bX}^{1/2}\hat\bbeta_k))(1+\cos\bTheta(\bSigma_{\bX}^{1/2}\bbeta_k,\bSigma_{\bX}^{1/2}\hat\bbeta_k))\\
&\le&2(\lambda_1+\lambda_k)(1-\cos\bTheta(\bSigma_{\bX}^{1/2}\bbeta_k,\bSigma_{\bX}^{1/2}\hat\bbeta_k))\\
&\le&2\dfrac{\lambda_{\max}(\bSigma_{\bX})}{\lambda_{\min}(\bSigma_{\bX})}(\lambda_1+\lambda_k)(1-\cos\bTheta(\bbeta_k,\hat\bbeta_k)),
\eean
where the last inequality follows from Lemma~\ref{angle}. Further note that $1-\cos\bTheta(\bbeta_k,\hat\bbeta_k)=\dfrac{\sin^2\bTheta(\bbeta_k,\hat\bbeta_k)}{1+\cos\bTheta(\bbeta_k,\hat\bbeta_k)}\le \eta_k^2$ and we have the desired conclusion.
\end{proof}

We consider the event $\Vert\widehat{\bM}-\bM\Vert_{\max}\le \epsilon,\Vert\widehat{\bSigma}_{\bX}-\bSigma_{\bX}\Vert_{\max}\le\epsilon$ for $\sqrt{2}s\epsilon<\min\{1/2,\lambda_{\min}(\bSigma_{\bX})\}$ and $\frac{\sqrt{2}s\epsilon}{\lambda_{\min}(\bSigma_{\bX})}+\frac{\sqrt{2}s\epsilon}{\lambda_{\min}(\bSigma_{\bX})\lambda_1}<a^*$ and $\dfrac{\sqrt{2}s\epsilon}{\lambda_{\min}(\bSigma_{\bX})}<\frac{\Delta}{2(1+\lambda_1^2)}$, where $\Delta$ is defined as in Condition~(C2) and $a^*=\min\{\frac{1}{2},\frac{\Delta}{\lambda_1+\lambda_2},\frac{\lambda_{\min}{(\bSigma_{\bX})}}{2}\}$. As a direct consequence, $\rho(\widehat{\bSigma}_{\bX}-\bSigma_{\bX},s)\le s\epsilon$. Also note that $\tilde\bolbeta_k=\dfrac{\widehat\bolbeta_k}{\Vert\hatbolbeta_k\Vert_2}$. 
In the RIFLE algorithm, $\bbeta_k^0$ is chosen to be sufficiently close to $\bbeta_k$, and the step size $\eta$ satisfies that $\eta\le\frac{\lambda_{\min}(\bSigma_{\bX})}{2}$ and
\beq
\sqrt{1+2\{(d/s')^{1/2}+d/s'\}}\sqrt{1-\frac{1}{16}\eta\lambda_{\min}(\bSigma_{\bX})\frac{\frac{\Delta}{\lambda_K}}{\kappa(\bSigma_{\bX})+1}}<1.
\eeq

Without loss of generality, in what follows we assume that $\cos\bTheta(\widehat{\bbeta}_j,\bbeta_k)>0$, because otherwise we can always consider $-\widehat{\bbeta}_j$, which spans the same susbspace as $\widehat{\bbeta}_j$.

\begin{lemma}\label{lem.GEP.beta1}
For the first direction $\widehat{\bolbeta}_1$, we have that $|\sin\Theta(\widehat{\bolbeta}_1,\bolbeta_1)|\le Cs\epsilon$ and $|\hat{\lambda}_1-\lambda_1|\le Cs\epsilon$. 
\end{lemma}

\begin{proof}[Proof of Lemma~\ref{lem.GEP.beta1}]
It is easy to see that 
\bea
\rho(\widehat{\bM}-\bM,s)=\sup_{\Vert\bu\Vert_2=1,\Vert\bu\Vert_0\le s} s|\bu^\T(\widehat{\bM}-\bM)\bu|\le \sup_{\Vert\bu\Vert_2=1,\Vert\bu\Vert_0\le s}\Vert\bu\Vert_1^2\cdot \Vert\widehat{\bM}-\bM\Vert_{\max}\le s\epsilon.
\eea

Similarly, $\rho(\widehat{\bSigma}_{\bX}-\bSigma_{\bX},s)\le s\epsilon$. Denote $\delta_1(s)=\sqrt{\rho^2(\widehat{\bM}-\bM,s)+\rho^2(\widehat{\bSigma}_{\bX}-\bSigma_{\bX},s)}$.
It follows that $\delta_1(s)\le \sqrt{2}s\epsilon$. Under our assumptions about $\epsilon$, we have that $|\sin\Theta(\widehat{\bolbeta}_1,\bolbeta_1)|\le Cs\epsilon$ by Lemma~\ref{lem.GEP1}. Lemma~\ref{lem2.2.sigma} further implies that $|\hat{\lambda}_1-\lambda_1|\le Cs\epsilon$.
\end{proof}

\begin{lemma}\label{lem.M.k+1}
If $|\sin\Theta(\widehat{\bolbeta}_j,\bolbeta_j)|\le Cs\epsilon$ for sufficiently small $\epsilon$, all $j\le k$, then $\delta(s)\le Cs\epsilon$.
\end{lemma}

\begin{proof}[Proof of Lemma~\ref{lem.M.k+1}]
It suffices to show that $\rho(\widehat{\bM}_{k+1}-\bM_{k+1},s)\le Cs\epsilon$.
Define $\bN_{l}=\widehat{\bSigma}_{\bX}\hat{\lambda}_l\hat{\bbeta}_l\hat{\bbeta}_l^\T\widehat{\bSigma}_{\bX}-\bSigma_{\bX}\lambda_l\bbeta_l\bbeta_l^\T\bSigma_{\bX}$. Then
\beq
\widehat{\bM}_{k+1}-\bM_{k+1}=(\widehat{\bM}-\bM)-\sum_{l\le k} \bN_l.
\eeq
It follows that 
\beq\label{lem.M.k+1.eq1}
\rho(\widehat{\bM}_{k+1}-\bM_{k+1},s)\le \rho(\widehat{\bM}-\bM,s)+\sum_{l\le k} \rho(\bN_l,s).
\eeq
According to the proof in Lemma~\ref{lem.GEP.beta1}, $\rho(\widehat{\bM}-\bM,s)\le s\epsilon$. 

For any vector $\bu$, we have
\bea
\bu^\T\bN_l\bu&=&\hat{\lambda}_l(\bu^\T\widehat{\bSigma}_{\bX}\hat{\bbeta}_l)^2-{\lambda}_l(\bu^\T{\bSigma}_{\bX}{\bbeta}_l)^2\\
&=&(\hat{\lambda}_l-\lambda_l)(\bu^\T\widehat{\bSigma}_{\bX}\widehat{\bbeta}_l)^2+\lambda_l\{(\bu^\T\widehat{\bSigma}_{\bX}\widehat\bbeta_l)^2-(\bu^\T\bSigma_{\bX}\bbeta_l)\}\equiv L_1+L_2
\eea
By Lemma~\ref{lem2.2.sigma}, we have that $|\hat\lambda_l-\lambda_l|\le Cs\epsilon$ when $\sin\Theta(\widehat{\bolbeta}_j,\bolbeta_j)\le Cs\epsilon$. Also,
\bea
|\bu^\T\widehat{\bSigma}_{\bX}\widehat\bbeta_l|&\le& |\bu^\T\bSigma_{\bX}\widehat\bbeta_l|+\Vert\bu\Vert_1\cdot\Vert\widehat\bbeta_l\Vert_1\Vert\widehat{\bSigma}_{\bX}-\bSigma_{\bX}\Vert_{\max}\le |\bu^\T\bSigma_{\bX}\bbeta_l|+s\epsilon\\
&\le& \sqrt{(\bu^\T\bSigma_{\bX}\bu)\cdot (\widehat\bbeta_l^\T\bSigma_{\bX}\widehat\bbeta_l)}+s\epsilon\le \lambda_{\max}(\bSigma_{\bX})+s\epsilon.
\eea
It follows that $L_1\le Cs\epsilon$. For $L_2$, note that $(\bu^\T\widehat{\bSigma}_{\bX}\widehat\bbeta_l)^2-(\bu^\T\bSigma_{\bX}\bbeta_l)=(\bu^\T\widehat{\bSigma}_{\bX}\widehat\bbeta_l+\bu^\T\bSigma_{\bX}\bbeta_l)(\bu^\T\widehat{\bSigma}_{\bX}\widehat\bbeta_l-\bu^\T\bSigma_{\bX}\bbeta_l)$. On one hand,
\bea
&&|\bu^\T\widehat{\bSigma}_{\bX}\widehat\bbeta_l+\bu^\T\bSigma_{\bX}\bbeta_l|\le |\bu^\T\widehat\bSigma_{\bX}\widehat{\bbeta}_l|+|\bu^\T\bSigma_{\bX}\bbeta_l|\\
&\le& \sqrt{\bu^\T\widehat{\bSigma}_{\bX}\bu\cdot \widehat{\bbeta}_l^\T\widehat{\bSigma}_{\bX}\widehat{\bbeta}_l}+\sqrt{\bu^\T\bSigma\bu\cdot \bbeta_l^\T\bSigma\bbeta_l}\\
&=&\sqrt{\bu^\T\widehat{\bSigma}_{\bX}\bu}+\sqrt{\bu^\T\bSigma_{\bX}\bu}=\sqrt{\bu^\T\bSigma_{\bX}\bu+\bu^\T(\widehat{\bSigma}_{\bX}-\bSigma_{\bX})\bu}+\sqrt{\bu^\T\bSigma\bu}\\
&\le& \sqrt{\lambda_{\max}(\bSigma_{\bX})+s\epsilon}+\sqrt{\lambda_{\max}(\bSigma_{\bX})}
\eea

On the other hand,
\bea
|\bu^\T\widehat{\bSigma}_{\bX}\hat\bbeta_l-\bu^\T\bSigma_{\bX}\bbeta_l|&\le& |\bu^\T(\widehat{\bSigma}_{\bX}-\bSigma_{\bX})\hat\bbeta_l|+|\bu^\T\bSigma_{\bX}(\hat{\bbeta}_l-\bbeta_l)|\\
&\le& s\epsilon+|\bu^\T\bSigma_{\bX}(\hat{\bbeta}_l-\bbeta_l)|\\
&\equiv& s\epsilon+L_3.
\eea

Further note that
\bea
L_3&\le&\sqrt{\bu^\T\bSigma_{\bX}\bu\cdot (\hat{\bbeta}_l-\bbeta_l)^\T\bSigma_{\bX}(\hat\bbeta_l-\bbeta_l)}\le \lambda_{\max}(\bSigma_{\bX})\cdot \Vert\hat{\bbeta}_l-\bbeta_l\Vert_2\\
&=&\lambda_{\max}(\bSigma_{\bX})\cdot \dfrac{1}{\sqrt{\tilde{\bbeta}_l^\T\widehat{\bSigma}_{\bX}\tilde{\bbeta}_l}}\Vert\tilde{\bbeta}_l-\sqrt{\tilde\bbeta^\T\widehat{\bSigma}_{\bX}\tilde{\bbeta}_l}\cdot\bbeta_l\Vert_2\\
&\le& \lambda_{\max}(\bSigma_{\bX})\cdot \dfrac{1}{\sqrt{\tilde{\bbeta}_l^\T\widehat{\bSigma}_{\bX}\tilde{\bbeta}_l}}\{\Vert\tilde{\bbeta}_l-\bbeta_l^*\Vert_2+|\dfrac{1}{\Vert\bbeta_l\Vert_2}-\sqrt{\tilde{\bbeta}_l^\T\widehat{\bSigma}_{\bX}\tilde{\bbeta}_l}|\cdot\Vert\bbeta_l\Vert_2\},
\eea
where $\bbeta_l^*=\dfrac{\bbeta_l}{\Vert\bbeta_l\Vert_2}$. By our assumption, $\Vert\tilde{\bbeta}_l-\bbeta_l^*\Vert_2\le \sqrt{2}\sin\Theta(\bbeta_l^*,\bbeta_l)\le Cs\epsilon$. For the second term, since $(\bbeta^*_l)^\T\bSigma_{\bX}\bbeta_l^*=\dfrac{\bbeta_l^\T\bSigma_{\bX}\bbeta_l}{\Vert\bbeta_l\Vert_2^2}=\dfrac{1}{\Vert\bbeta_l\Vert_2^2}$, we have $\Vert\bbeta_l\Vert_2=\dfrac{1}{\sqrt{(\bbeta^*_l)^\T\bSigma_{\bX}\bbeta_l^*}}$. It follows that,
\bea
&&|\dfrac{1}{\Vert\bbeta_l\Vert_2}-\sqrt{\tilde{\bbeta}_l^\T\widehat{\bSigma}_{\bX}\tilde{\bbeta}_l}|=|\sqrt{(\bbeta_l^*)^\T\bSigma_{\bX}\bbeta_l^*}-\sqrt{\tilde{\bbeta}_l^\T\widehat{\bSigma}_{\bX}\tilde{\bbeta}_l}|\\
&\le& |\sqrt{(\bbeta_l^*)^\T\bSigma_{\bX}\bbeta_l^*}-\sqrt{\tilde{\bbeta}_l^\T\bSigma_{\bX}\tilde{\bbeta}_l}|+|\tilde{\bbeta}_l^\T(\bSigma_{\bX}-\widehat{\bSigma}_{\bX})\tilde{\bbeta}_l|/(\sqrt{\tilde{\bbeta}^\T\bSigma_{\bX}\tilde{\bbeta}}+\sqrt{\tilde{\bbeta}^\T\widehat{\bSigma}_{\bX}\tilde\bbeta})\\
&\le& \dfrac{|(\bbeta_l^*)^\T\bSigma_{\bX}\bbeta_l^*-\tilde{\bbeta}_l^\T\widehat{\bSigma}_{\bX}\tilde{\bbeta}_l|}{\sqrt{(\bbeta_l^*)^\T\bSigma_{\bX}\bbeta_l^*}+\sqrt{\tilde{\bbeta}_l^\T\widehat{\bSigma}_{\bX}\tilde{\bbeta}_l}}+\dfrac{s\epsilon}{\lambda_{\min}(\bSigma_{\bX})}.
\eea
Because $\sqrt{(\bbeta_l^*)^\T\bSigma_{\bX}\bbeta_l^*}+\sqrt{\tilde{\bbeta}_l^\T\widehat{\bSigma}_{\bX}\tilde{\bbeta}_l}\ge \lambda_{\min}(\bSigma_{\bX})$, it suffices to find a bound for $|(\bbeta_l^*)^\T\bSigma_{\bX}\bbeta_l^*-\tilde{\bbeta}_l^\T\widehat{\bSigma}_{\bX}\tilde{\bbeta}_l|$.
\bea
&&|(\bbeta_l^*)^\T\bSigma_{\bX}\bbeta_l^*-\tilde{\bbeta}_l^\T\widehat{\bSigma}_{\bX}\tilde{\bbeta}_l|\le |(\bbeta_l^*)^\T\bSigma_{\bX}\bbeta_l^*-\tilde{\bbeta}_l^\T{\bSigma}_{\bX}\tilde{\bbeta}_l|+s\epsilon\\
&\le&|(\bbeta_l^*-\tilde{\bbeta}_l)^\T\bSigma_{\bX}\bbeta_l^*|+|\tilde{\bbeta}_l^\T\bSigma_{\bX}(\bbeta_l^*-\tilde{\bbeta}_l)|+s\epsilon\\
&\le& \sqrt{(\bbeta_l^*-\tilde{\bbeta}_l)^\T\bSigma_{\bX}(\bbeta_l^*-\tilde{\bbeta}_l)\cdot (\bbeta_l^*)^\T\bSigma_{\bX}\bbeta_l^*}+\sqrt{\tilde{\bbeta}_l^\T\bSigma_{\bX}\tilde{\bbeta}_l\cdot (\bbeta_l^*-\tilde{\bbeta}_l)^\T\bSigma_{\bX}(\bbeta_l^*-\tilde{\bbeta}_l)}\nonumber\\
&&+s\epsilon\\
&\le& 2\lambda_{\max}(\bSigma_{\bX})\cdot \Vert\bbeta_l^*-\tilde{\bbeta}_l\Vert_2+s\epsilon\\
&\le &\sqrt{2\{1-\cos\Theta(\bbeta^*_l,\tilde\bbeta_l)\}}+s\epsilon\le \sqrt{2}|\sin\Theta(\bbeta_l^*,\tilde\bbeta_l)|+s\epsilon\le Cs\epsilon,
\eea
which also implies that $\dfrac{1}{\sqrt{\tilde{\bbeta}_l^\T\widehat{\bSigma}_{\bX}\tilde{\bbeta}_l}}\le \dfrac{1}{\bbeta_l^*\bSigma_{\bX}\bbeta_l^*-Cs\epsilon}\le \dfrac{2}{\lambda_{\min}(\bSigma_{\bX})}$ if $s\epsilon\le \dfrac{\lambda_{\min(\bSigma_{\bX})}}{2}$. Finally, by \eqref{lem.M.k+1.eq1} we have the desired conclusion.
\end{proof}

Now we define $cr_k(s)=\inf_{F:|F|\le s}cr(\bM_{k,F},\bSigma_F)$. We have the following lemma.
\begin{lemma}\label{lem.sin.k}
Assume that $|\sin\Theta(\hat{\bbeta}_l,\bbeta_l)|<Cs\epsilon$ for any $l<k$. If $\dfrac{2s\epsilon}{\lambda_{\min}(\bSigma_{\bX})}+\dfrac{2s\epsilon}{\lambda_{\min}(\bSigma_{\bX})\lambda_1}<\min\{\frac{1}{2},\dfrac{\Delta}{\lambda_1+\lambda_2},\dfrac{\lambda_{\min}(\bSigma_{\bX})}{2},\dfrac{\Delta}{2(1+\lambda_1^2)}cr_k(s)\}$, we have that $\sin\Theta(\hat\bbeta_k,\bbeta_k)\le Cs\epsilon$.
\end{lemma}

\begin{proof}[Proof of Lemma~\ref{lem.sin.k}]
Combine Lemma~\ref{lem.M.k+1} with Lemma~\ref{lem.GEP1} and we have the desired conclusion.
\end{proof}

\begin{proof}[Proof of Theorem~\ref{consistency.GEP}]
Combining Lemma~\ref{lem.GEP.beta1} with Lemma~\ref{lem.sin.k}, we have that if $\Vert\widehat{\bM}-\bM\Vert_{\max}\le\epsilon$, $\Vert\widehat{\bSigma}_{\bX}-\bSigma_{\bX}\Vert_{\max}\le\epsilon$, then $|\sin\Theta(\hat{\bbeta}_k,\bbeta_k)|\le Cs\epsilon, k=1,\ldots,K$. By Theorem~\ref{th:main} we have the desired conclusion. 
\end{proof}


\end{document}